\documentclass[11pt]{article}

\usepackage[T1]{fontenc}

\usepackage{amsfonts,amsmath,amssymb}
\usepackage{bm}
\usepackage{breqn}
\usepackage{todonotes}
\usepackage[affil-it]{authblk}
\usepackage{amsthm}
\usepackage{multirow}
\usepackage{natbib}
\usepackage{hyperref}  
\usepackage{graphicx}
\usepackage{url}
\usepackage[ruled]{algorithm2e}
\usepackage{tikz}
\usepackage{tkz-linknodes}
\usetikzlibrary{quantikz}
\usepackage{hhline}
\usepackage[margin=2.5cm]{geometry}
\usepackage{collectbox}

\usepackage{stmaryrd}
\usepackage{ dsfont }
\usepackage{array}
\usepackage{float}
\usepackage{multirow}

\DeclareMathOperator*{\var}{Var} 

\usepackage{algorithmic}

\DeclareMathOperator{\dist}{dist}

\DeclareMathOperator{\T}{T}
\newcommand*{\cA}{\mathcal{A}}

\newcommand*{\dC}{\mathbb{C}}
\newcommand*{\cD}{\mathcal{D}}
\newcommand*{\cE}{\mathcal{E}}

\newcommand*{\F}{\mathds{F}}
\newcommand*{\rF}{\mathrm{F}}
\newcommand*{\cG}{\mathcal{G}}

\newcommand*{\cI}{\mathcal{I}}
\newcommand*{\dI}{\mathbb{I}}
\newcommand*{\cJ}{\mathcal{J}}
\newcommand*{\cK}{\mathcal{K}}
\newcommand*{\cN}{\mathcal{N}}
\newcommand*{\cM}{\mathcal{M}}

\newcommand*{\cO}{\mathcal{O}}
\newcommand*{\cP}{\mathcal{P}}

\newcommand*{\cS}{\mathcal{S}}
\newcommand*{\fS}{\mathfrak{S}}
\newcommand*{\bS}{\mathbf{S}}
\newcommand*{\bU}{\mathbb{U}}

\newcommand*{\cX}{\mathcal{X}}

\DeclareMathOperator{\TV}{TV}
\DeclareMathOperator{\W}{Wg}

\DeclareMathOperator{\Ber}{Bern}
\DeclareMathOperator{\Haar}{Haar}

\DeclareMathOperator{\KL}{KL}
\newcommand*{\eps}{\varepsilon}

\newcommand{\bigzero}{\mbox{\normalfont\huge $0$}}
\newcommand{\bigeye}{\mbox{\normalfont\huge $\dI$}}

\newcommand*{\dd}{d_{\diamond}}
\newcommand*{\din}{d_{\mathrm{in}}}
\newcommand*{\dout}{d_{\mathrm{out}}}
\newcommand*{\danc}{d_{\mathrm{anc}}}
\newcommand*{\td}{d_{\tr}}
\newcommand*{\id}{\mathrm{id}}
\newcommand*{\tr}{\mathrm{Tr}}
\newcommand*{\spr}[2]{\langle #1 | #2 \rangle}

\newcommand*{\pr}[1]{\mathds{P}\left(#1 \right)}
\newcommand*{\ex}[1]{\mathds{E}\left(#1 \right)}

\newtheorem{theorem}{Theorem}[section]
\newtheorem{lemma}[theorem]{Lemma}

\newtheorem{proposition}[theorem]{Proposition}

\newtheorem{definition}[theorem]{Definition}

\usepackage{times}
\usepackage{breqn}

\title{Quantum Channel Certification with Incoherent Strategies}
\author[1]{Omar Fawzi}
\author[2]{Nicolas Flammarion}
\author[3]{Aurélien Garivier}
\author[1]{Aadil Oufkir}
\affil[1]{\small{Univ Lyon, Inria, ENS Lyon, UCBL, LIP, Lyon, France}
}
\affil[2]{ \small{EPFL,  Lausanne, Switzerland}}
\affil[3]{\small{Univ Lyon, ENS de Lyon, UMPA UMR 5669,
 Lyon, France}}

\begin{document}

\maketitle

\begin{abstract}%
 In the problem of quantum channel certification, we have black box access to a quantum process and would like to decide if this process matches some predefined specification or is $\varepsilon$-far from this specification. The objective is to achieve this task while minimizing the number of times the black box is used.  
 Here, we focus on optimal incoherent strategies for two relevant extreme cases of channel certification. The first one is when the predefined specification is a unitary channel, e.g., a gate in a quantum circuit. 
 In this case, we show that testing whether the black box is described by a fixed unitary operator in dimension $d$ or $\varepsilon$-far from it in the trace norm requires $\Theta(d/\varepsilon^2)$ uses of the black box. The second setting we consider is when the predefined specification is a completely depolarizing channel with input dimension $d_{\text{in}}$ and output dimension $d_{\text{out}}$.
 In this case, we prove that, in the non-adaptive setting, $\Tilde{\Theta}(d_{\text{in}}^2d_{\text{out}}^{1.5}/\varepsilon^2)$ uses of the channel are necessary and sufficient to verify whether it is equal to the depolarizing channel or $\varepsilon$-far from it in the diamond norm. 
 Finally, we prove a lower bound of $\Omega(d_{\text{in}}^2d_{\text{out}}/\varepsilon^2)$ for this problem in the adaptive setting. Note that the special case $d_{\text{in}} = 1$ corresponds to the well-studied quantum state certification problem.
\end{abstract}


\section{Introduction}

We consider the problem of quantum channel certification which consists in verifying whether a quantum process to which we have black box access behaves as intended. A valid process in quantum theory is modeled by a quantum channel. A quantum channel with input dimension $\din$ and output dimension $\dout$ is a linear map $\mathbb{C}^{\din\times \din}\rightarrow \mathbb{C}^{\dout\times \dout}$ satisfying some positivity and normalization conditions (see Section~\ref{sec:prelim} for details). 
Given a complete description of a known   quantum channel $\cN_0$ and $N$ copies of an unknown quantum channel $\cN$ that can be either $\cN_0$ or $\eps$-far from it, 
 at each step $1\le t\le N$, we can choose an input quantum state, send it through the unknown process $\cN$ then measure the output quantum state. After collecting a sufficient amount of classical observations, our goal is to 
 decide in which case is the quantum channel $\cN$ with high probability and while minimizing $N$. We also call this problem testing identity to the quantum channel $\cN_0$.

This testing task is important for many reasons. 
Firstly, the building blocks of a quantum computation are unitary gates.
It is thus important to understand the complexity of checking that an unknown channel implements a given gate as specified. 
Secondly, quantum channel certification is the natural generalization of the quantum state certification. Indeed, if the channels are constant quantum states, then testing them becomes equivalent to testing those states. It might seem at first sight that using the Choi–Jamiołkowski isomorphism, quantum channel certification can be obtained by applying quantum state certification algorithms to the corresponding Choi states. 
However, there are at least two reasons this is not true: first, in models where an auxilliary system is not allowed the Choi state cannot easily be prepared from a black box implementing the channel and second, the natural notion of distance for channels does not correspond to the trace distance between Choi states.
Finally, quantum process tomography, the problem of learning completely a channel in the diamond norm is costly \citep{surawy2022projected,oufkir2023sample}, and our hope is that certification can be done with fewer copies than full tomography.

In this paper, we focus on two extreme cases, $\cN_0(\rho)=\cN_U(\rho)=U\rho U^\dagger$ is  a unitary channel where $U$ is a unitary matrix and $\cN_0(\rho)=\cD(\rho)=\tr(\rho)\frac{\dI}{\dout}$ is the completely depolarizing channel.

\textbf{Contribution.}
We propose an incoherent ancilla-free testing algorithm for testing identity to a fixed unitary channel $\cN_U$ in the trace distance using $\cO(d/\eps^2)$  measurements (here $\din=\dout=d$). 
The tester chooses a random input state and measures with the corresponding POVM conjugated by the unitary $U$. This result is stated in Thm.~\ref{thm:Test-ID}. 
The standard inequality relating the $1$-norm of a Choi state and the diamond norm of the channel only implies an upper bound $\cO(d^2/\eps^2)$. We obtain the quadratic improvement in the dimension dependency by proving a new inequality between the entanglement fidelity and the trace distance to the identity channel (Lem.~\ref{lem:fid-diamond}). 
Moreover, we establish a matching lower bound of $\Omega(d/\eps^2)$ for testing identity to a fixed unitary channel in the trace distance. This lower bound applies even for ancilla-assisted strategies.  For this, we construct a well-chosen distribution of channels $\eps$-far from the identity channel. After a sufficient number of measurements, the observations under the two hypotheses should be distinguishable, i.e., the (Kullback-Leibler) KL divergence is $\Omega(1)$. However, we can show that for this particular choice of distribution over channels, any adaptive tester can only increase the KL divergence by at most $\cO(\eps^2/d)$  after a measurement no matter the dependence on the previous observations. The lower bound is stated and proved in~Thm.~\ref{thm:LB-test-id-ind}. 

Concerning the certification of the completely depolarizing channel $\cD(\rho)=\tr(\rho)\frac{\dI}{\dout}$, we propose an incoherent ancilla-free strategy to distinguish between $\cN=\cD$ and $\cN$ is $\eps$-far from it in the diamond distance using $\cO(\din^2 \dout^{1.5}/\eps^2)$  measurements (see Thm.~\ref{thm:Test-D}). For this we show how to reduce this certification problem to the certification of the maximally mixed state (testing mixedness) $\frac{\dI}{\dout}$. 
We choose the input state $\proj{\phi}$ randomly and we compare the $2$-norm between the output state $\cN(\proj{\phi})$  and the maximally mixed state $\frac{\dI}{\dout}$. We show that with at least a probability $\Omega(1)$, we have $Y=\|\cN(\proj{\phi})-\dI/\dout\|_2^2\ge \eps^2/(4\dout\din^2)$. This inequality is sufficient to obtain the required complexity, however, it requires some work to be proved. First, we show a similar inequality  in expectation  using Weingarten calculus. Then we control the variance of the random variable $Y$ carefully in a way that this upper bound depends on the actual difficulty of the problem, mainly 
the expectation of $Y$ and the diamond distance between $\cN$ and $\cD$. Next, we obtain the anti-concentration inequality using the Paley-Zygmund inequality. 

On the other hand, we establish a lower bound of $\Omega\left( \din^2 \dout^{1.5}/(\log(\din\dout/\eps)^2\eps^2)\right)$ for testing identity to the depolarizing channel with non-adaptive strategies (Thm.~\ref{thm:LB_DEP_non-adap}).  For this, we construct a random quantum channel whose output states are almost $\cO(\eps/\din)$-close (in the $1$-norm) to the maximally mixed state $\frac{\dI}{\dout}$ except for a neighborhood of an input state chosen randomly whose output is $\eps$-far from $\frac{\dI}{\dout}$ in the $1$-norm. 
Then we use LeCam's method \citep{lecam1973convergence} as in \citep{bubeck2020entanglement} with some differences. First, we need to condition on the event that the input states chosen by the testing algorithm have very small overlaps with the best input state.  This conditioning is the main reason for the additional logarithmic factor we obtain in the lower bound. Next, with a construction using random matrices with Gaussian entries rather than $\Haar$  distributed unitaries, we can invoke hypercontractivity \citep[Proposition $5.48$]{aubrun2017alice} which  allows us to control all the  moments once we upper bound the second moment. In the special case $\din = 1$, this recovers a result of \citep{bubeck2020entanglement} while significantly simplifying their analysis.  
 Furthermore, a lower bound of $\Omega(\din^2 \dout/\eps^2)$ is proved for adaptive strategies (Thm.~\ref{thm:LB_DEP_adap}) using the same construction. In this proof, we use  Kullback-Leibler divergence instead of the Total-Variation distance. 
 We refer to Table~\ref{tab} for a summary of these results.

\begin{table}[t!]
\centering
\begin{tabular}{  c| c |  c } 
\hline
  \centering \multirow{2}{*}{\textbf{Testing identity to $\cN_0$}}  &\textbf{Lower bound} &  \textbf{Upper bound}\\
  &\textit{Ancilla-assisted} &\textit{Ancilla-free}\\
  \hline\hline
    $\cN_0=\cN_{U}$ &  \multirow{2}{*}{$\Omega\left(\frac{d}{\eps^2}\right)  $, Thm.~\ref{thm:Test-ID} } &   \multirow{2}{*}{$\cO\left(\frac{d}{\eps^2}\right)$, Thm.~\ref{thm:Test-ID}}\\
    adaptive strategies in trace distance $\td$ & &
 \\
 \hline
 $\cN_0=\cN_{U}$ & \multirow{2}{*}{$\Omega\left(\frac{d}{\eps^2}\right)  $, Thm.~\ref{thm:Test-ID}} &  \multirow{2}{*}{$\cO\left(\frac{d}{\eps^4}\right)$, Thm.~\ref{thm:Test-ID}}
 \\
  adaptive strategies in diamond distance $\dd$ && \\
 \hline\hline
   $\cN_0=\cD$ & \multirow{2}{*}{$\Omega\left(\frac{\din^{2}\dout^{1.5}}{\log(\din\dout/\eps)^2\eps^2}\right)  $, Thm.~\ref{thm:LB_DEP_non-adap}} &\multirow{2}{*}{  $\cO\left(\frac{\din^{2}\dout^{1.5}}{\eps^2}\right)$, Thm.~\ref{thm:Test-D}}
 \\
   non-adaptive  strategies &&\\
 \hline
$\cN_0=\cD$   & \multirow{2}{*}{$\Omega\left(\frac{\din^{2}\dout+ \dout^{1.5}}{\eps^2}\right)$, Thm.~\ref{thm:LB_DEP_adap} }  &\multirow{2}{*}{  $\cO\left(\frac{\din^{2}\dout^{1.5} }{\eps^2}\right)$, Thm.~\ref{thm:Test-D}}
 \\
  adaptive strategies &&\\
   \hline
\end{tabular}
\caption{Lower  and upper bounds for testing identity of quantum channels in the diamond and trace distances using incoherent  strategies. $\cN_U$ is the unitary quantum channel $\cN_U(\rho)=U\rho U^\dagger$ and $\cD$ is the depolarizing channel $\cD(\rho)= \tr(\rho)\frac{\dI}{\dout}$. Our proposed algorithms (upper bounds) are ancilla-free while the lower bounds hold even for ancilla-assisted algorithms.
}
\label{tab}
\end{table}

\textbf{Related work.} Testing identity to a unitary channel can be  seen as a generalization of the usual testing identity problem for discrete distributions \citep{valiant2016instance} and quantum states \citep{chen2022toward}. However, 
in the worst-case setting, testing identity to the identity channel requires $\Omega(d/\eps^2)$ measurements in contrast to testing identity to a rank $1$ quantum state or a Dirac distribution which can be done with only $\cO(1/\eps^2)$ measurements/samples. Also, in the definition of testing identity to a unitary channel problem, we don't require the \textit{unknown} tested channel to be unitary. In this latter setting, efficient tests can be designed easily if an auxiliary system is allowed. This can be found along with other tests on properties of unitary channels in \citep{wang2011property}. We also refer to the survey \citep{montanaro2013survey} for other examples of tests on unitary channels.
Since a unitary channel has a Choi rank equal to $1$ and the depolarizing channel has a Choi rank equal to $\din\dout$, this work is a first step to obtain instance-optimal quantum channel certification as for the classical case \citep{valiant2016instance} or quantum states \citep{chen2022toward}. On the other hand, testing identity to the completely depolarizing channel is a generalization of identity testing of distributions \citep{diakonikolas2017optimal} and testing mixedness of states \citep{bubeck2020entanglement,chen2022tight-mixedness}.  
In particular, if the input dimension is $\din=1$, the channels are constant and the problem reduces to a testing mixedness of states of dimension $\dout$. 
In this case, we recover the optimal complexity of \citep{bubeck2020entanglement, chen2022tight-mixedness}. 

Another noteworthy work is unitarity estimation of \citep{chen2022unitarity}. In this work, it is shown that ancilla-free non-adaptive strategies  could  estimate $\tr(\cJ_\cN^2)$ to within $\eps$ using $\cO(d^{0.5}/\eps^2)$  measurements where $\cJ_\cN$ is the Choi state of the channel $\cN$. In particular, this estimation can be used to distinguish between $\cN=\cN_U$ for which $\tr(\cJ_\cN^2)=1$ and $\cN=\cD$ for which $\tr(\cJ_\cN^2)=1/d^2$. A matching lower bound  (in $d$) is given for adaptive strategies in \citep{chen2022unitarity} improving the previous lower bound of \citep{chen2022exponential}. 
This complexity may seem to contradict our results but this is not the case. Indeed, such a test  cannot be used for testing identity to a fixed unitary channel for instance since we can have two unitary channels that are $\eps$-far in the diamond distance.
Finally, we note that testing problems for states were also studied for the class of coherent or entangled strategies where one can use arbitrary entangled measurements: see e.g.,~\citep{buadescu2019quantum} for state certification in this setting. In this article, we focus on incoherent strategies where entangled inputs to the different uses of the channel or entanglement measurements are not allowed.

\section{Preliminaries} 
\label{sec:prelim}
We consider quantum channels of input dimension $\din$ and output dimension $\dout$. 
We adopt the bra-ket notation: a column vector is denoted $\ket{\phi}$ and its adjoint is denoted $\bra{\phi}=\ket{\phi}^\dagger$. With this notation, $\spr{\phi}{\psi}$ is the dot product  of the vectors $\phi$ and $\psi$ and, for a unit vector $\ket{\phi}\in \bS^d$,  $\proj{\phi}$ is the projector on the space spanned by $\phi$. The canonical basis $\{e_i\}_{i\in [d]}$ is denoted $\{\ket{i}\}_{i\in [d]}:=\{\ket{e_i}\}_{i\in [d]}$. A quantum state is a positive semi-definite Hermitian matrix of trace $1$. The fidelity between two quantum states is defined $\rF(\rho,\sigma)=\left( \tr\sqrt{\sqrt{\rho}\sigma \sqrt{\rho}}\right)^2$. It is symmetric and admits the simpler expression if one of the quantum states $\rho$ or $\sigma$ has rank $1$: $\rF(\rho,\proj{\phi})= \bra{\phi}\rho\ket{\phi}$. 
 A $(\din, \dout)$-dimensional quantum channel is a map $\cN: \mathbb{C}^{\din\times \din}\rightarrow \mathbb{C}^{\dout\times \dout}$ of the form $\cN(\rho)=\sum_{k}A_k \rho A_k^\dagger$ (Kraus decomposition) where the Kraus operators $\{A_k\}_k$ satisfy $\sum_k A_k^\dagger A_k=\dI$. For instance, the identity map $\id_d(\rho)=\rho$ admits the Kraus operator $\{\dI_d\}$ and the completely depolarizing channel $\cD(\rho)=\tr(\rho)\frac{\mathds{I}}{\dout}$ admits the Kraus operators $\left\{\frac{1}{\sqrt{\dout}}\ket{i}\bra{j}\right\}_{j\in[\din], i\in [\dout]}$. Given a unitary matrix $U$, the corresponding unitary channel $\cN_U(\rho)=U\rho U^\dagger$ admits the Kraus operator $\{U\}$. We denote by $\Haar(d)$ the $\Haar$ probability  measure over the compact group of unitary $d \times d$ matrices. A $\Haar$ random vector is then any column vector of a $\Haar$ distributed unitary.

 We define the trace distance between two quantum channels $\cN$ and $\cM$ as the trace norm of their difference: $\td(\cN,\cM):= \max_{\rho }\|(\cN-\cM)(\rho) \|_1$ where the maximization is over quantum states and the Schatten $p$-norm of a matrix $M$ is defined as $\|M\|_p^p=\tr\left(\sqrt{M^\dagger M}^p\right)$. In some situations, it can be helpful to allow an auxiliary system and apply the identity on it. \sloppy In this case, we obtain the diamond distance which is defined formally as $\dd(\cN,\cM):= \max_{\rho   }\|\id_d\otimes(\cN-\cM)(\rho) \|_1$ where the maximization is over quantum states $\rho \in \mathbb{C}^{d\times d} \otimes \mathbb{C}^{\din\times \din}$. Note that because of the Schmidt decomposition we can always suppose that $d=\din$. The trace/diamond distance can be thought of as a worst-case distance, while we can define an average case distance by the Schatten $2$-norm between the corresponding Choi states. We define the Choi state of the channel $\cN$ as $\cJ_{\cN}= \id\otimes\cN(\proj{\Psi})$ where $\ket{\Psi}=\frac{1}{\sqrt{\din}}\sum_{i=1}^{\din}\ket{i}\otimes \ket{i}$ is the maximally entangled state. 
The map $\cJ : \cN \mapsto \id\otimes\cN(\proj{\Psi}) $  is an isomorphism called the Choi–Jamiołkowski isomorphism \citep{choi1975completely,jamiolkowski1972linear}. Note that for any quantum channel $\cN$, $\cJ_\cN$ is positive semi-definite and satisfy $\tr_2(\cJ_\cN)=\frac{\dI}{\din}$. Moreover, any $\cK$ satisfying these conditions is called a Choi state and we can construct a quantum channel $\cN$ such that  $\cJ_\cN=\cK$.

We consider the problem of testing identity to a fixed channel.  Given a fixed quantum channel $\cN_0$ and a precision  parameter $\eps>0$, the goal is to test whether an unknown  quantum channel $\cN$ is exactly $\cN_0$ or $\eps$-far from it  with at least a probability $2/3$:
\begin{align*}
   H_0: \cN=\cN_0  ~~~~~\text{ vs. }~~~~~ H_1: \dist(\cN,\cN_0)\ge \eps
\end{align*}
where $\dist\in \{\dd, \td\}$.
 $H_0$ is called the null hypothesis while $H_1$ is called the alternate hypothesis. 
 An algorithm $\cA$ is $1/3$-correct for this problem if it outputs $H_1$ while $H_0$ is true with a probability at most $1/3$ and  outputs $H_0$ while $H_1$ is true with a probability at most $1/3$. If $0<\dist(\cN,\cN_0)< \eps$, the algorithm $\cA$ can output any hypothesis. 
The natural figure of merit for this test is the diamond (resp. trace)  distance because it characterizes the minimal error probability to distinguish between two quantum channels when auxiliary systems are allowed (resp. not allowed) \citep{watrous2018theory}. 

A testing algorithm can only extract classical information from the unknown quantum channel $\cN$ by performing a measurement on the output state. In this article we only consider  incoherent strategies. That is, the testing algorithm  can only use one copy of the channel at each step. In other words, for a $(\din, \dout)$ dimensional channel, the testing algorithm could only use $\din$ dimensional input states and $\dout$ dimensional measurement devices.
Precisely, a $d$-dimensional measurement is defined by a POVM (positive operator-valued measure) with a finite number of elements: this is a set of positive semi-definite matrices $\mathcal{M}=\{M_x\}_{x\in \cX}$ acting on the Hilbert space $\mathbb{C}^{d}$ and satisfying $\sum_{x\in \cX} M_x=\dI$. Each element $M_x$ in the POVM $\mathcal{M}$ is associated with the outcome $x\in \cX$. The tuple $\{\tr(\rho M_x)\}_{x\in \cX}$ is non-negative and sums to $1$: it thus defines a probability. Born's rule \citep{1926ZPhy...37..863B} says that the probability that the measurement on a quantum state $\rho$ using the POVM $\cM$ will output $x$ is exactly $\tr(\rho M_x)$. 
Depending on whether an auxiliary system is allowed to be used, we distinguish two types of incoherent strategies.

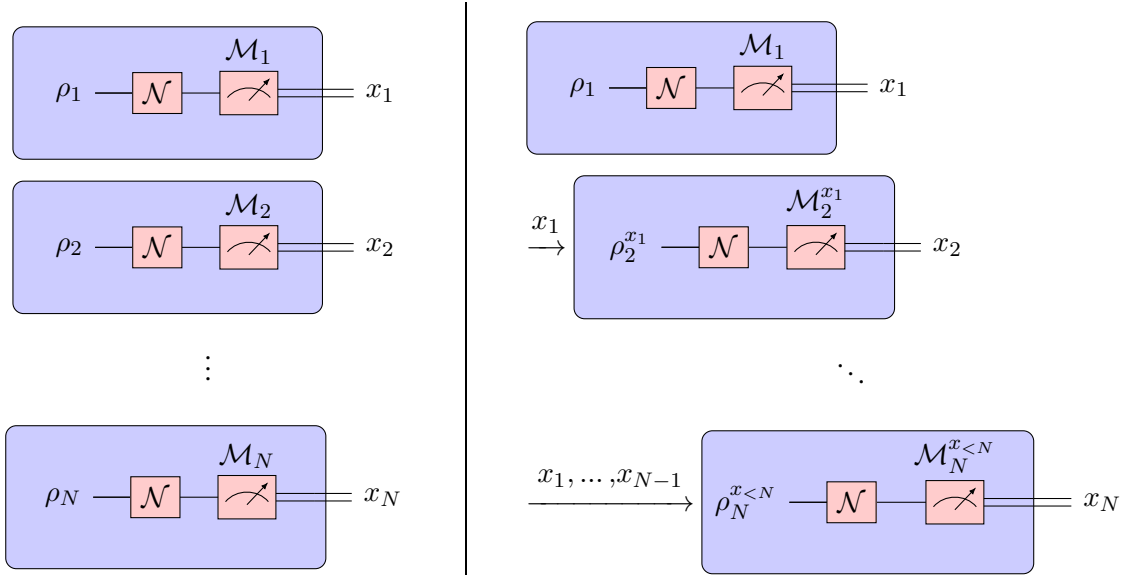
\begin{figure}
\centering
   \begin{minipage}[c]{.40\linewidth}
   \centering
   \tikzset{
     operator/.append style={fill=red!20},
     meter/.append style={fill=red!20}
    }
        \begin{quantikz}[thin lines] 
\gategroup[wires=1,steps=4,style={rounded corners,fill=blue!20,inner sep=13pt},background]{}&\lstick{$\rho_{1}$} & \gate{\cN} \qw  &\meter{$\cM_1$}   &\cw &\cw\rstick{ $x_1$}
\end{quantikz} \\
\begin{quantikz}[thin lines] 
\gategroup[wires=1,steps=4,style={rounded corners,fill=blue!20,inner sep=13pt},background]{}&\lstick{$\rho_{2}$} & \gate{\cN} \qw&\meter{$\cM_2$}   & \cw &\cw \rstick{$x_2$} 
\end{quantikz}
\begin{center}$\vdots$\end{center}
\begin{quantikz}[thin lines] 
\gategroup[wires=1,steps=4,style={rounded corners,fill=blue!20,inner sep=15pt},background]{}&\lstick{$\rho_{N}$}& \gate{\cN} \qw &\meter{$\cM_N$}  &\cw  & \rstick{ $x_N$} \cw\end{quantikz}
   \end{minipage}
   \vline \hspace*{20pt}
   \begin{minipage}[c]{.52\linewidth}
   \tikzset{
     operator/.append style={fill=red!20},
        meter/.append style={fill=red!20}
    }
\begin{quantikz}[thin lines] 
\gategroup[wires=1,steps=4,style={rounded corners,fill=blue!20,inner sep=13pt},background]{}&\lstick{$\rho_{1}$} & \gate{\cN} \qw  &\meter{$\cM_1$}   &\cw &\rstick{ $x_1$} \cw
\end{quantikz}
\\$\underset{\xrightarrow{\hspace{0.3cm}}}{x_1}$ 
\begin{quantikz}[thin lines] 
\gategroup[wires=1,steps=4,style={rounded corners,fill=blue!20,inner sep=15pt},background]{}&\lstick{$\rho_{2}^{x_1}$} & \gate{\cN} \qw&\meter{$\cM_2^{x_1}$}   &\cw& \rstick{ $x_2$} \cw 
\end{quantikz}
\begin{center}$\ddots$\end{center}
$\underset{\xrightarrow{\hspace{2cm}}}{x_1, \dots,x_{N-1}}$
\begin{quantikz}[thin lines] 
\gategroup[wires=1,steps=4,style={rounded corners,fill=blue!20,inner sep=15pt},background]{}&\lstick{$\rho_{N}^{x_{<N}}$} & \gate{\cN} \qw &\meter{$\cM_N^{x_{<N}}$} &\cw   & \rstick{ $x_N$} \cw\end{quantikz}
   \end{minipage}
\caption{Illustration of an ancilla-free incoherent  non-adaptive (left) and adaptive (right) strategies for testing identity of quantum channels. 
The observations $(x_1, \dots, x_{N})$ are then processed with a classical algorithm to answer $H_0$ or $H_1$.}
\label{Fig: (Non)Adap-ancill-free}
\end{figure}

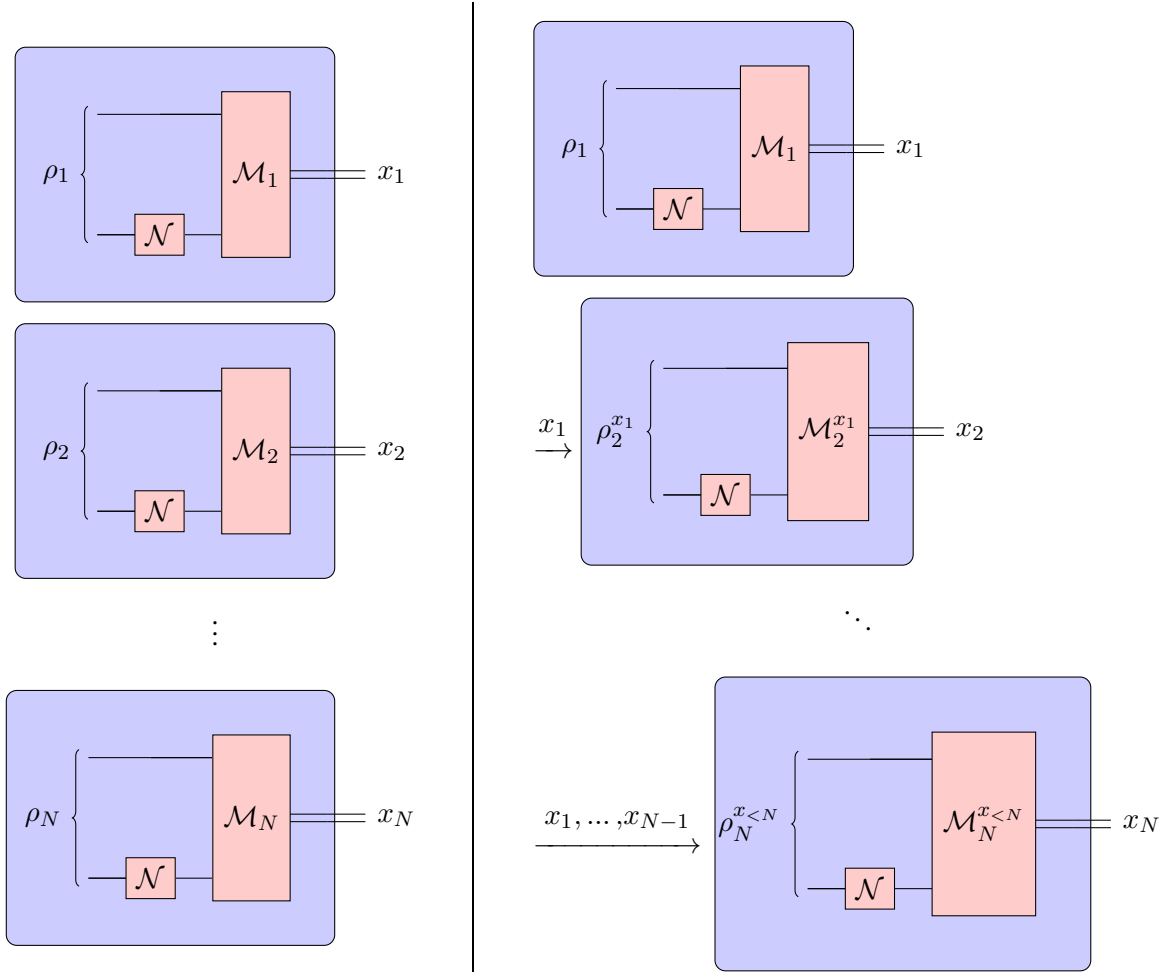
\begin{figure}[t!]
\centering
   \begin{minipage}[c]{.40\linewidth}
   \centering
   \tikzset{
     operator/.append style={fill=red!20},
     meter/.append style={fill=red!20}
    }
        \begin{quantikz}[thin lines] 
		\gategroup[wires=3,steps=4,style={rounded corners,fill=blue!20,inner sep=13pt},background]{}&\lstick[wires=3, nwires={2}]{$\rho_{1}$} & \qw &\gate[wires=3 , nwires={2}]{\cM_1 }  \qw  &      &&
		\\ & &  &  &  \cw &\cw\rstick{$x_{1}$} 
		\\&& \gate{\cN} \qw &   & &
		\end{quantikz}\\
\begin{quantikz}[thin lines] 
		\gategroup[wires=3,steps=4,style={rounded corners,fill=blue!20,inner sep=13pt},background]{}&\lstick[wires=3, nwires={2}]{$\rho_{2}$} & \qw &\gate[wires=3 , nwires={2}]{\cM_2 }  \qw  &      &&
		\\ & &  &  &  \cw &\cw\rstick{$x_{2}$} 
		\\&& \gate{\cN} \qw &   & &
		\end{quantikz}
\begin{center}$\vdots$\end{center}
\begin{quantikz}[thin lines] 
		\gategroup[wires=3,steps=4,style={rounded corners,fill=blue!20,inner sep=13pt},background]{}&\lstick[wires=3, nwires={2}]{$\rho_{N}$} & \qw &\gate[wires=3 , nwires={2}]{\cM_N }  \qw  &      &&
		\\ & &  &  &  \cw &\cw\rstick{$x_{N}$} 
		\\&& \gate{\cN} \qw &   & &
		\end{quantikz}
   \end{minipage}
   \vline \hspace*{20pt}
   \begin{minipage}[c]{.52\linewidth}
   \tikzset{
     operator/.append style={fill=red!20},
        meter/.append style={fill=red!20}
    }
\begin{quantikz}[thin lines] 
		\gategroup[wires=3,steps=4,style={rounded corners,fill=blue!20,inner sep=13pt},background]{}&\lstick[wires=3, nwires={2}]{$\rho_{1}$} & \qw &\gate[wires=3 , nwires={2}]{\cM_1 }  \qw  &      &&
		\\ & &  &  &  \cw &\cw\rstick{$x_{1}$} 
		\\&& \gate{\cN} \qw &   & &
		\end{quantikz}
\\$\underset{\xrightarrow{\hspace{0.3cm}}}{x_1}$ 
\begin{quantikz}[thin lines] 
		\gategroup[wires=3,steps=4,style={rounded corners,fill=blue!20,inner sep=13pt},background]{}&\lstick[wires=3, nwires={2}]{$\rho_{2}^{x_1}$} & \qw &\gate[wires=3 , nwires={2}]{\cM_2^{x_1} }  \qw  &      &&
		\\ & &  &  &  \cw &\cw\rstick{$x_{2}$} 
		\\&& \gate{\cN} \qw &   & &
		\end{quantikz}
\begin{center}$\ddots$\end{center}
$\underset{\xrightarrow{\hspace{2cm}}}{x_1, \dots,x_{N-1}}$
\begin{quantikz}[thin lines] 
		\gategroup[wires=3,steps=4,style={rounded corners,fill=blue!20,inner sep=17pt},background]{}&\lstick[wires=3, nwires={2}]{$\rho_{N}^{x_{<N}}$} & \qw &\gate[wires=3 , nwires={2}]{\cM_N^{x_{<N}} }  \qw  &      &&
		\\ & &  &  &  \cw &\cw\rstick{$x_{N}$} 
		\\&& \gate{\cN} \qw &   & &
		\end{quantikz}
   \end{minipage}
\caption{Illustration of an ancilla-assisted incoherent  non-adaptive (left) and adaptive (right) strategies for testing identity of quantum channels. 
The observations $(x_1, \dots, x_{N})$ are then processed with a classical algorithm to answer $H_0$ or $H_1$.}
\label{Fig: (Non)Adap-ancill-ass}
\end{figure}

\begin{itemize}
    \item \textbf{Ancilla-free strategies.} At each step $t$, the learner would choose an input $\din$-dimensional state $\rho_t$ and a $\dout$-dimensional measurement device $\cM_t=\{M_x^t\}_{x\in \cX_t} $. It thus sees the outcome $x_t \in \cX_t$ with a probability $\tr(\cN(\rho_t) M_{x_t}^t)$. If the choice of the state $\rho_t$ and measurement device $\cM_t$ can depend on the previous observations $(x_1,\dots, x_{t-1})$ the strategy is called \textbf{adaptive}, otherwise it is called \textbf{non-adaptive} (see Fig.~\ref{Fig: (Non)Adap-ancill-free} for an illustration).
    \item \textbf{Ancilla-assisted strategies.} At each step $t$, the learner would choose an input $\danc\times \din$-dimensional state $\rho_t\in \mathbb{C}^{\danc\times \danc} \otimes\mathbb{C}^{\din\times \din} $ (for some arbitrary $\danc$) and a $\danc\times \dout$-dimensional measurement device $\cM_t=\{M_x^t\}_{x\in \cX_t} $.  It thus sees the outcome $x_t \in \cX_t$ with a probability $\tr(\id_{\danc}\otimes\cN(\rho_t) M_{x_t}^t)$. If the choice of the state $\rho_t$ and measurement device $\cM_t$ can depend on the previous observations $(x_1,\dots, x_{t-1})$ the strategy is called \textbf{adaptive}, otherwise it is called \textbf{non-adaptive} (see Fig.~\ref{Fig: (Non)Adap-ancill-ass} for an illustration). 
\end{itemize}

Note that we can see an ancilla-free strategy as a special case of an ancilla-assisted strategy with $\danc=1$. But it turns out that ancilla-assisted strategies have an advantage over ancilla-free strategies for some problems e.g., \citep{chen2022quantum, chen2022exponential}.  However, in this article, we show that ancilla-free strategies suffice to achieve the optimal complexity for testing identity to a fixed unitary channel or to the depolarizing channel. 
\section{Testing identity to a unitary channel}\label{sec:test id to id}
In this section, we focus on the problem of testing identity to a fixed unitary channel in the diamond and trace distances.  Given a fixed 
unitary $U$ and a precision  parameter $\eps>0$, 
the goal is to distinguish between the hypotheses:
\begin{align*}
   H_0: \cN=\cN_U=U \cdot U^\dagger   ~~~~~\text{ vs. }~~~~~ H_1: \dist(\cN,\cN_U)\ge \eps
\end{align*}
 with at least a probability $2/3$ where $\dist\in \{\dd, \td\}$. Since we consider a unitary channel, the input and output dimensions should be equal $\din=\dout=d$. 
 
\textbf{Reduction to the case $U=\dI$.} Knowing the unitary channel $\cN_U$ is equivalent to knowing $U$. 
We can thus reduce every testing identity to $\cN_U$ to testing identity to $\cN_{\dI}=\id_d$ by conjugating the measurement device by $U$.   This is possible because the trace/diamond distance is unitary invariant: $\dist(\cN, \cN_U) = \dist(U^\dagger\cN U, \id)$ and $\tr(U^\dagger \cN(\rho) U M)= \tr( \cN(\rho) U MU^\dagger)$ for all $\rho$ and $M$. From now on, we only consider the case $U=\dI$ and we call this particular testing problem to $\cN_{\dI}=\id_d$ simply ‘‘\textit{testing identity to identity}''.

Given the nature of the diamond and trace distances, under the alternate hypothesis, a channel $\cN$ could be equal to the identity channel except on a neighbourhood of some state. In addition, this state is unknown to the learner.
When the algorithm is allowed to use an auxiliary system, it can prepare the Choi state $\cJ_{\cN}$ of the channel $\cN$ (which essentially captures everything  about the channel) by taking as input the maximally entangled state $\proj{\Psi}$. Under the null hypothesis $H_0$, the Choi state is exactly $\cJ_{\id} = \id\otimes \id(\proj{\Psi})=\proj{\Psi} $ while under the alternate hypothesis $H_1$, the Choi state $\cJ_{\cN}$ has a fidelity with $ \proj{\Psi}$ satisfying:
\begin{align*}
    \tr(\id\otimes \cN(\proj{\Psi} )\proj{\Psi} )=\rF(\cJ_{\cN},\proj{\Psi} ) \le 1-\frac{1}{4} \left\|\cJ_{\cN}-\cJ_{\id} \right\|_{1 }^2\le 1-\frac{\dd(\cN,\id)^2}{4d^2}
\end{align*}
where we use a Fuchs–van de Graaf inequality \citep{fuchs1999cryptographic} and the  standard inequality relating the diamond norm between two channels and the trace norm between their corresponding Choi states: $\|\cJ_\cN-\cJ_\cM \|_1\ge \frac{\dd(\cN,\cM)}{d}$ (e.g., \cite{jenvcova2016conditions}). 
Thus, a measurement using the POVM $\cM_\Psi=\{\proj{\Psi}, \dI-\proj{\Psi}\}$ can distinguish between the two situations.
However, if the tester is not allowed to use an auxiliary system it can  neither prepare the Choi state $\cJ_\cN$ nor measure using the POVM $\cM_{\Psi}$. 
Instead, we use a random $d$-dimensional rank-$1$  input state. Indeed, this choice is natural because the expected fidelity between the input state $\proj{\phi}$ and the output state $\cN(\proj{\phi})$ can be easily related to the fidelity between Choi states: $\mathds{E}_{\ket{\phi}\sim \Haar} \left[ \rF(\cN(\proj{\phi}),\proj{\phi}))\right] =\frac{1+d~\rF(\cJ_{\cN},\proj{\Psi} ) }{1+d}$ (see Lem.~\ref{lem:fid-ent-avg} for a proof). 
This Lemma is well known because it relates the average fidelity $\mathds{E}_{\ket{\phi}\sim \Haar} \left[ \rF(\cN(\proj{\phi}),\proj{\phi}))\right]$ and the entanglement fidelity $\rF(\cJ_{\cN},\proj{\Psi} )$. 
If we measure using the measurement device $\cM_{\phi}= \{\proj{\phi}, \dI-\proj{\phi} \}$, the error probability under $H_0$ is $0$, and under $H_1$ is $\mathds{E}_{\ket{\phi}\sim \Haar} [\bra{\phi} \cN(\proj{\phi})\ket{\phi}]=\mathds{E}_{\ket{\phi}\sim \Haar}  [\rF(\cN(\proj{\phi}),\proj{\phi})]$.  
The algorithm is detailed in Alg.~\ref{Alg}. 
\begin{algorithm}
\caption{Testing identity to identity  in the diamond/trace distance}\label{Alg}
\begin{algorithmic}[t!]
\STATE $N= \cO(d/\eps^4)$ (replace with $N= \cO(d/\eps^2)$ for testing in the trace distance).\\
\FOR{$k =1:N$}
 \STATE Sample $\phi_k$ a $\Haar$  random vector in $\bS^d$.
 \STATE Measure the output state $\cN(\proj{\phi_k})$ using the POVM $\{\proj{\phi_k}, \dI-\proj{\phi_k}\}$.
 \STATE Observe $X_k \sim \Ber(1-\bra{\phi_k}\cN( \proj{\phi_k}) \ket{\phi_k} )$.
\ENDFOR
\STATE \textbf{if} $\exists k : X_k=1$ \textbf{then return } $\cN$ is $\eps$-far from $\id$ \textbf{else return} $\cN=\id$.
\end{algorithmic}
\end{algorithm}

We can upper bound the error probability under $H_1$  using the well-known Lem.~\ref{lem:fid-ent-avg} and  our Lem.~\ref{lem:fid-diamond}: 

\begin{align*}
    \mathds{E}_{\ket{\phi}\sim \Haar}  \left[ \rF(\cN(\proj{\phi}),\proj{\phi})\right] &=\frac{1+d~\rF(\cJ_{\cN},\proj{\Psi} ) }{1+d}
    \le 1-\frac{\td(\cN, \id)^2}{4(d+1)}\le  1-\frac{\dd(\cN, \id)^4}{16(d+1)}.
\end{align*}
Observe that the standard inequality 
$\|\cJ_\cN-\cJ_\cM \|_1\ge \frac{\dd(\cN,\cM)}{d}$ implies: $  \mathds{E}_\phi \left[ \rF(\cN(\proj{\phi}),\proj{\phi})\right] \le 1-\frac{\dd(\cN,\id)^2}{4d(d+1)}$ 
which has a better dependency in the diamond distance but a worst dependency in the dimension $d$. 
This simple lemma (Lem.~\ref{lem:fid-diamond}), which relates the entanglement fidelity and the trace/diamond distance when one of the channels is unitary 
might be of independent interest. 
To obtain a $1/3$ correct algorithm it suffices to repeat the described procedure using $N=\cO(d/\eps^{2}) $  copies of $\proj{\phi}$ in the case of trace distance and 
$N=\cO(d/\eps^{4}) $  copies of $\proj{\phi}$ in the case of diamond distance. Indeed, for instance for the trace distance, 
the probability of error under $H_1$ can be controlled as follows:
    \begin{align*}
      \mathds{P}_{H_1}(\text{error})&= \mathds{P}_{H_1}(\forall k \in [N]: X_k=0) 
     = \prod_{k=1}^{N} \mathds{P}_{H_1}(X_k=0) 
      \\&\le \left( 1- \frac{\eps^2}{4(d+1)}  \right)^{N}
      \le \exp\left( -  \frac{\eps^2N}{4(d+1)} \right)
     \le \frac{1}{3}
    \end{align*}
    for $N=4\log(3)(d+1)/\eps^2=\cO(d/\eps^2)$ \footnote{All the logs of this paper are taken in base $e$ and the information is measured in ‘‘nats''.}. A similar proof shows that $\cO(d/\eps^4)$ copies are sufficient to test in the diamond distance.  This concludes the correctness of  Alg.~\ref{Alg}. 
A matching lower bound of $N=\Omega\left( \frac{d}{\eps^2}\right)$ is proved in App.~\ref{sec:LB}. For this we construct a random quantum channel $\eps$-far from the identity channel in the trace/diamond distance 
but looks similar to the identity channel for the majority of the input states. Then we compare the observations under the two hypotheses using the $\KL$ divergence. Interestingly, the analogous classical problem, testing identity to identity has a complexity $\Theta(d/\eps)$. Thus, for this task, when going to the quantum case, the dependence on the dimension $d$ remains the same whereas the dependency in the precision parameter $\eps$ changes from $\eps$ to $\eps^2$. In fact, obtaining the correct $\eps^2$ dependence is the main difficulty in this lower bound. It requires a carefully chosen construction 
inspired by the quantum skew divergence~\citep{skewdivergence} and a fine analysis using Weingarten calculus.
We summarize the main result of this section in the following theorem.

\begin{theorem}\label{thm:Test-ID}
  There is an incoherent ancilla-free algorithm for  testing identity to identity in the trace distance  using only $N=\cO\left( \frac{d}{\eps^2}\right)$  measurements. 
  Moreover, this algorithm can also solve the testing identity to identity problem in the diamond distance   using only $N=\cO\left( \frac{d}{\eps^4}\right)$  measurements. 
   On the other hand, any incoherent adaptive ancilla-assisted strategy   requires, in the worst case, a number of measurements satisfying $N=\Omega\left( \frac{d}{\eps^2}\right)$ to distinguish between  $\cN=\id$ and $\dd(\cN, \id)>\eps$ (or $\td(\cN, \id)>\eps$) with a probability at least $2/3$.
\end{theorem}

\section{Testing identity to the depolarizing channel}
In this section, we move to study the problem  of testing identity to the completely depolarizing channel $\cN_0=\cD$ in the diamond distance. Given a precision parameter $\eps>0$ and an unknown quantum channel $\cN$, we would like to  test whether  $H_0: \cN=\cD$ or $H_1:\dd(\cN,\cD)\ge \eps$ with a probability of error at most $1/3$. If $\cN=\cD$, the tester should answer the null hypothesis $H_0$ with a probability at least $2/3$ whereas if $\dd(\cN,\cD)\ge \eps$, the tester should answer the alternate hypothesis $H_1$ with a probability at least $2/3$. The alternate condition means that 
\begin{align*}
    \dd(\cN,\cD)\ge \eps &\Longleftrightarrow \exists \ket{\phi}\in \bS^{\din\times \din} : \|\id_{\din}\otimes \cN(\proj{\phi}) - \id_{\din}\otimes \cD(\proj{\phi})\|_1\ge \eps.
\end{align*}
This inequality implies a lower bound of the $1$-norm between the Choi states: $\|\cJ_\cN-\cJ_\cD\|_1\ge \frac{\eps}{\din}$. 
The simplest idea would be to use this inequality and reduce the problem of testing channels to testing states. Actually, this kind of reduction from channels to states has been used for quantum process tomography \citep{surawy2022projected} and shadow process tomography \citep{kunjummen2021shadow}. The Choi state of  the depolarizing channel $\cD$ is $\cJ_\cD= \frac{1}{\din}\sum_{i,j=1}^{\din} \ket{i}\bra{j} \otimes\tr(\ket{i}\bra{j}) \frac{\dI}{\dout}= \frac{\dI_{\din}\otimes \dI_{\dout}}{\din\dout}=\frac{\dI_{\din\dout}}{\din \dout},$ 
so by applying the previous inequality, we obtain for a quantum channel $\cN$ $\eps$-far from the depolarizing channel $\cD$: $ \left\| \cJ_\cN - \frac{\dI}{\din\dout}\right\|_1\ge \frac{\eps}{\din}.$
Then, we can apply a reduction to the testing mixedness of quantum states  \citep{bubeck2020entanglement} to design an ancilla-assisted strategy requiring  $\cO\left(\frac{\din^{1.5}\dout^{1.5}}{(\eps/\din)^2}\right)= \cO\left(\frac{\din^{3.5}\dout^{1.5}}{\eps^2}\right)$    measurements since the dimension of the states $\cJ_\cN$ and $\frac{\dI}{\din\dout}$ is $\din\dout$ and the precision parameter is $\frac{\eps}{\din}$. However, this approach has two problems. First, we need to be able to use an auxiliary system to prepare the Choi state $\cJ_\cN$, which is an additional resource. 
Next, the complexity $\cO\left(\frac{\din^{3.5}\dout^{1.5}}{\eps^2}\right)$, as we shall see later, is not optimal. If one tries to reduce to testing identity of states in the  $2$-norm (App.~\ref{app: testing states}) one obtains a slightly better bound but still not optimal and requires using an auxiliary system.

Inspired by the testing identity to identity problem (Sec.~\ref{sec:test id to id}), when we do not know one of the best input states, we choose it to be random. Let $\ket{\phi}$ be a $\Haar$  random vector. If the input state is $\proj{\phi}$, the output state under $H_0$ is $\cD(\proj{\phi})=\frac{\dI}{\dout}$ and under $H_1$ is $\cN(\proj{\phi})$. So it is natural to ask what would be the distance between $\cN(\proj{\phi})$ and $\frac{\dI}{\dout}$. Note that in general, it is much easier to compute the expectation of the $2$-norms than the $1$-norms. For this reason, we start by computing the expectation of the $2$-norm
between $\cN(\proj{\phi})$ and $\frac{\dI}{\dout}$ (see Lem.~\ref{lem-app:2-2 norms} for a proof).
\begin{lemma}\label{lem:2-2 norms}
Let  $\cM=\cN-\cD$ and   $\cJ_\cM= \id\otimes \cM(\proj{\Psi})$. We have:
\begin{align*}
   \mathds{E}_{\ket{\phi}\sim \Haar}\left(\left\|\cM(\proj{\phi}) \right\|_2^2\right)&=\frac{\left\|\cM(\dI) \right\|_2^2  +\din^2\left\|\cJ_\cM \right\|_2^2}{\din(\din+1)}   \ge \frac{\din}{\din+1}\left\|\cJ_\cM \right\|_2^2.
\end{align*}
\end{lemma}
We now need to relate the $2$-norm  between the  Choi state  of the channel $\cN$  and the Choi state of the depolarizing channel $\|\cJ_\cN-\cJ_\cD\|_2$ with the diamond distance between them $\dd(\cN, \cD)$. This is done in the following Lemma whose proof can be found in App.~\ref{app-lem:diamond-2-choi}:
\begin{lemma}\label{lem:diamond-2-choi}
Let $\cN_1$ and $\cN_2$ be two  $(\din, \dout)$-quantum channels. We have:
\begin{align*}
   \|\cJ_{\cN_1}-\cJ_{\cN_2}  \|_2 \ge \frac{d_\diamond(\cN_1,\cN_2)}{\din\sqrt{\dout}}.
\end{align*}
\end{lemma}
  
Let $\cN$ be a channel satisfying $\dd(\cN, \cD)\ge \eps$ and $X= \left\|\cN(\proj{\phi}) -\frac{\dI}{\dout}\right\|_2^2$, we obtain from Lem.~\ref{lem:2-2 norms}, \ref{lem:diamond-2-choi}:
\begin{align}\label{ineq-expectation}
    \ex{X}= \mathds{E}_{\ket{\phi}\sim \Haar}\left(\left\|\cN(\proj{\phi}) -\frac{\dI}{\dout}\right\|_2^2\right)&\ge \frac{\din}{\din+1}\left\|\cJ_\cN -\frac{\dI}{\din\dout}\right\|_2^2\ge \frac{\eps^2}{2\din^2\dout}.
\end{align}
If we could show that $X$ is larger than  
$\Omega(\ex{X})$ with constant probability, then we can reduce our problem  to the usual testing identity of quantum states in the $2$-norm (quantum state certification in App.~\ref{app: testing states}) and obtain an incoherent \textit{ancilla-free} algorithm using  $\cO\left(\frac{\sqrt{\dout}}{(\eps/\din\sqrt{\dout})^2}\right)= \cO\left(\frac{\din^{2} \dout^{1.5}}{\eps^2}\right)$    measurements.  
Establishing this turns out to be the most technical part of the proof as is summarized in the following theorem.

\begin{theorem}\label{thm:PZ-D} Let $\ket{\phi} $ be  a  $\Haar$  distributed  vector in $\bS^d$ (or any $4$-design). Let  $X=\left\|\cN(\proj{\phi}) -\frac{\dI}{\dout}\right\|_2^2$. We have:
    \begin{align*}
    \var(X)= \cO\left( [\ex{X}]^2\right). 
\end{align*}
\end{theorem}
This Theorem is the most technical part of the paper and we believe that it can be generalized for any difference of channels with a similar approach.
Moreover, applying this inequality along with the Paley-Zygmund inequality are sufficient for our reduction: 
we only need to repeat our test $\cO(1)$  times to reduce the error probability to $1/3$ for testing identity to the depolarizing channel. 

\textbf{Sketch of the proof.} We observe first that $  \var(X)= \var\left(\tr \big[\left(\cN(\proj{\phi}) \right)^2\big]\right)$. Then using Weingarten calculus (Lem.~\ref{lem:Wg},\ref{lem:wg3}), we can compute the expectation $$\ex{\left(\tr \big[\left(\cN(\proj{\phi}) \right)^2\big] \right)^2 }=\frac{1}{\din(\din+1)(\din+2)(\din+3)} \sum_{\alpha\in \fS_4} F(\alpha) $$ 
\sloppy where for $\alpha\in \fS_4$,  ${F(\alpha)= \sum_{i,j,k,l,k',l'}\tr_\alpha(A_{l'}^\dagger \ket{j}\bra{i}A_k, A_k^\dagger A_l , A_l^\dagger \ket{i}\bra{j} A_{k'} , A_{k'}^\dagger  A_{l'} )}$ and $\tr_{\alpha}(M_1,\dots,M_n)=\Pi_j \tr(\Pi_{i\in C_j} M_i)$ for $\alpha=\Pi_j C_j $ and $C_j$ are cycles. Let $m= \left\|\cN(\dI)-\frac{\din}{\dout}\dI\right\|_2$ and $\eta= \din\left\|\cJ_\cN -\frac{\dI}{\din\dout}\right\|_2$.  Next, we upper bound the function $F$ as shown in Table~\ref{tab:summary-F}. 
 \begin{table}[ht]
\centering
\begin{tabular}{  |c  |  c| c|} 
\hline
  \textbf{Permutation } $\alpha$  & \textbf{Upper bound on } $F(\alpha)$ &\textbf{Reference}\\
  \hline\hline
   $(13)$ & $ \left(\frac{\din}{\dout}+\eta^2\right)^2 $&
   \\\hline
   $\id, (132), (314), (24)(13)$ & $ \frac{\din/\dout+\eta^2}{\dout} + \frac{\eta^2}{\dout}+5\eta^4$ &(Lem.~\ref{lem: F(id)})
	\\\hline 
$(312), (134) $& $ \left(\frac{\din^2}{\dout}+m^2\right)\left(\frac{\din}{\dout}+\eta^2\right)$&
	\\\hline
$ (1234) $& $ \left(\frac{\din^2}{\dout}+m^2\right)^2$&
	\\ \hline
$ (24), (1432)$ &$ \frac{\din^2}{\dout^2}+\frac{2m^2}{\dout}+25\eta^4$ &(Lem.~\ref{lem: F(24)})
	\\\hline
 $(142), (243)$& $\frac{\din/\dout+\eta^2}{\dout} + \frac{\eta^2}{\dout}+5\eta^4$ &(Lem.~\ref{lem: F(142)})
	\\\hline
 $(14), (12), (23), (34), (1324), $&\multirow{2}{*}{$  \frac{\din^2}{\dout^2}+\frac{\din}{\dout}\eta^2+ \frac{m^2}{\dout}+5m\eta^3$}& \multirow{2}{*}{(Lem.~\ref{lem: F(14)})}\\
$(1423), (1243), (1342)$& &
		\\\hline
  $(12)(34), (14)(23), (234), (124)$& $\frac{\din^3}{\dout^2}+2\frac{\din}{\dout}m^2+m^2\eta^2 $ &(Lem.~\ref{lem: F(12)(34))})
  \\\hline
\end{tabular}
\caption{Upper bounds on the function $F$ for different input permutations.}
\label{tab:summary-F}
\end{table}
This is the hardest step of the proof and requires a fine analysis for many of the $F(\alpha)'s$. A particularly useful trick we use repeatedly is known as the \textit{replica trick} and says that if $\F=\sum_{i,j=1}^d  (\ket{i}\otimes \ket{j})(  \bra{j}\otimes \bra{i})$ is the flip operator then we have for all $A,B\in \mathbb{C}^{d\times d}: \tr((A\otimes B)  \F)=\tr(AB)$ and similarly  $\tr(A\otimes B )=\tr(A)\tr(B)$. Moreover, another trick we need frequently is to use the partial transpose to make appear the positive semi-definite matrices $M^\dagger M=\sum_{k,l} A_k^\dagger A_l\otimes A_k^{\top}\bar{A}_l$ and $MM^\dagger=\sum_{k,l} A_k A_l^\dagger\otimes \bar{A}_k A^{\top}_l$ where $M=\sum_k A_k\otimes \bar{A}_k$ is defined using the Kraus operators $\{A_k\}_k$ of the channel $\cN$.  Furthermore we can prove the approximation $\left\|M^\dagger M-  \frac{\din}{\dout}\proj{\Psi}\right\|_1
\le 5\eta^2$ where $\ket{\Psi}=\frac{1}{\sqrt{\din}}\sum_{i=1}^{\din} \ket{i}\otimes \ket{i}$ is the maximally entangled state (see Lem.~\ref{lem: M-Psi}). This is used for $\alpha= (142)$ and $(24)$. An application of the data processing inequality on the previous approximation gives:  $\left\|\tr_2(M^\dagger M)-  \frac{1}{\dout}\dI\right\|_1
\le 5\eta^2$ which is used for the permutations $\alpha= \id$ and $(14)$ . It turns out that applying such approximation will not give a sufficiently good upper bound of $F((12)(34))$ because it can be written as $F((12)(34))=\tr\left(\left( MM^\dagger \right)^{T_2}  \cdot\cM(\dI)^{\otimes 2} \F\right) +\frac{\din^3}{\dout^2}  +  2\frac{\din}{\dout}\tr(\cM(\dI)^2)  $
which depends 
instead on the matrix $MM^\dagger$. In this case we proceed by projecting $MM^\dagger$ onto the hyperplane orthogonal to $\ket{\Psi}$. This can be interpreted using representation theory. In fact,  the space spanned by $\ket{\Psi}$ and its complementary are irreducible representations that decompose the space $\left(\mathbb{C}^d\right)^{\otimes 2}$ for the action of $U\otimes \bar{U}$ where $U$ is a unitary matrix. On the other hand, changing the Kraus operators $A_k\leftrightarrow U A_k $ or $A_k\leftrightarrow A_k U $ in the expression of the channel $\cN$ does not affect the variance of $X $ because $\Haar$ measure is invariant under left and right multiplication with a unitary matrix. 
The detailed  proof  of  these  bounds can be found in App.~\ref{app-thm:PZ-D}.

We have now the required tools to design and prove the correctness of an algorithm for testing identity to the depolarizing channel.
\begin{theorem}\label{thm:Test-D}
 There is an incoherent ancilla-free algorithm requiring a number of  measurements $N=\cO\left(\frac{\din^2\dout^{1.5}}{\eps^2}\right)$
to distinguish between $\cN=\cD$ and $\dd(\cN,\cD)\ge \eps$ with a   success probability $2/3$.
\end{theorem}

As explained before, our algorithm is a reduction to the testing identity of quantum states. For the convenience of the reader we include this latter with a proof of its correctness in App.~\ref{app: testing states}. Note that we need to test quantum states in the $2$-norm which is different than the usual quantum state certification~\citep{bubeck2020entanglement}. The algorithm for testing identity to the depolarizing channel is  described  in  Alg.~\ref{alg-dep}.
\begin{algorithm}[!h]
\caption{Testing identity to the depolarizing channel  in the diamond norm}\label{alg-dep}
\begin{algorithmic}[t]
\STATE $M= 2200$.
\FOR{$k =1:M$}
 \STATE Sample $\phi_k$ a $\Haar$ random vector in $\bS^{\din}$.
 \STATE Test whether $h_0:\cN(\proj{\phi_k})= \frac{\dI}{\dout}$ or $h_1:\left\| \cN(\proj{\phi_k})- \frac{\dI}{\dout}\right\|_2\ge \frac{\eps}{2\sqrt{\dout}\din}$ using the testing identity of quantum states  Alg.~\ref{alg-states}, with an error probability $\delta=1/(3M)$, that answers the hypothesis $h_{i_k}$, $i_k\in \{0,1\}$.
\ENDFOR
\STATE \textbf{if} $\exists k : i_k=1$ \textbf{then return } $\cN$ is $\eps$-far from $\cD$ \textbf{else return} $\cN=\cD$.
\end{algorithmic}
\end{algorithm}

We remark that  Alg.~\ref{alg-dep} uses the channel only on a constant number of  random input states $\{\proj{\phi_k}\}_k$. 
One could think that querying the channel $\cN$ on more diverse inputs could lead to a more efficient algorithm. 
However, it turns out that Alg.~\ref{alg-dep} is basically optimal as we prove a matching lower bound up to a poly-logarithmic factor. 

\begin{theorem}\label{thm:LB_DEP_non-adap}
     Let $\eps\le 1/32$, $\din \ge 80$ and $\dout\ge 10$. Any ancilla-assisted non-adaptive algorithm for testing identity to the depolarizing channel (for both trace and diamond distances)  requires, in the worst case, a number of measurements satisfying:
    \begin{align*}
        N=\Omega\left(\frac{\din^2\dout^{1.5}}{\log(\din\dout/\eps)^2\eps^2}\right).
    \end{align*}
\end{theorem}
This theorem shows that our proposed  Alg.~\ref{alg-dep} is almost optimal in the dimensions $(\din, \dout)$  and the precision parameter $\eps$  thus the complexity of testing identity to the depolarizing channel is $\Tilde{\Theta}(\din^2 \dout^{1.5}/\eps^2)$
 which is slightly surprising. Indeed, we can remark that the complexity of testing identity of discrete distributions and quantum states is the square root (for constant $\eps$) of the complexity of the corresponding learning problems in the same setting. This rule does not apply for quantum channels since we know from \citep{surawy2022projected,oufkir2023sample} that the complexity of learning quantum channels in the diamond distance with  non-adaptive incoherent strategies is  $\Tilde{\Theta}(\din^3\dout^3/\eps^2)$.

\textbf{Sketch of the proof.} Under the null hypothesis $\cN=\cD$. Under the alternate hypothesis, we construct randomly the quantum channel $\cN\sim \cP$ of the form:
 $	\cN(\rho)= \cD(\rho)  +\frac{\eps}{\dout\|u\|_2^2} \bra{u}\rho \ket{u} U$
where $\ket{u}$ is a standard Gaussian vector and $U$ has Gaussian entries: for all $i\le j\in [\dout]$,  $U_{j,i}=\bar{U}_{i,j} \sim \mathbf{1}\{i\neq j\} \cN_c(0, 16/\dout)$  conditioned on the event $\cG=\{\|U\|_1\ge \dout, \|U\|_\infty \le 32\}$. Note that the usual construction applied on the Choi state gives a sub-optimal lower bound in the trace or diamond  distances. Using a concentration inequality of Lipschitz functions of Gaussian random variables~\citep{wainwright2019high}, we show that with a high probability $\cN$ is $\eps$-far from $\cD$ in the trace and diamond  distances. 
Then we use LeCam's method to lower bound the $\TV$ distance between the distribution of the observations under the two hypotheses: $\TV\left(\mathds{P}_{\cD}^{I_1,\dots, I_N}\Big\| \mathds{E}_{\cN\sim \cP}\mathds{P}_{\cN}^{I_1,\dots, I_N}\right)\ge \frac{1}{3}$ where $I_1,\dots, I_N$ are the observations the algorithm obtains after the measurements and $N$ is a sufficient number of measurements for the correctness of the algorithm. We can suppose w.l.o.g. that the input states are pure $\rho_t=\proj{\psi_t} $ and the measurement devices are given by POVMs of the form  $\cM_t=\{\lambda_{i_t}^t \proj{\phi_{i_t}^t}\}$. Also, we can write  \[\ket{\psi_t}= A_t\otimes \dI \ket{\Psi_{\din}}, \ket{\phi_{i_t}^t}=  B_{i_t}^t\otimes \dI \ket{\Psi_{\dout}}  \quad \text{where} \quad  \ket{\Psi_d}= \frac{1}{\sqrt{d} } \sum_{i=1}^{d} \ket{ii}.\] 
Then we condition on the event $\cE$ that $u$ satisfies  $\|u\|_2\ge \sqrt{\frac{\din}{6}}$ and 
\begin{equation}\label{eq:Pu}
    \forall t\in [N]:P_u^t=\sum_{i_t,j_t}\frac{\lambda^t_{i_t}\lambda^t_{j_t}\bra{u}A_t^\dagger B_{i_t}^t  B_{j_t}^{t,\dagger}A_t  \ket{u}^4 }{\tr(B_{i_t}^{t,\dagger}A_tA_t^\dagger B_{i_t}^{t})\tr(B_{j_t}^{t,\dagger}A_tA_t^\dagger B_{j_t}^{t})}\le (7\log(N))^4\din^2\dout^3.
\end{equation}
Note that  $P_u^t$ is a polynomial of degree $8$ in the entries of $u$. We can prove that the event $\cE$ occurs with a probability at least $9/10$ using a concentration inequality deduced from the Hypercontractivity of Gaussian polynomials \cite[Corollary 5.49]{aubrun2017alice} and a union bound. Observe that the union bound here adds only a factor of $\log(N)^4$ because for a non-adaptive strategy, there are at most $N$ couples of inputs/measurements.
To carry out the analysis, we need to bound the moments of the random variables:
 \[Z_t(U,V)= \eps^2 \mathds{E}_{i_t} \Big[\Phi^{t,i_t}_{u, U}\Phi^{t,i_t}_{v, V}\Big]\quad \text{where} \quad \Phi^{t,i_t}_{u, U}=\frac{\tr(B_{i_t}^{t,\dagger}A_t  \proj{u} A_t^\dagger B_{i_t}^{t}  U)}{\|u\|_2^2\tr(B_{i_t}^{t,\dagger}A_tA_t^\dagger B_{i_t}^{t})}.\]  
 Since $Z_t$ is a polynomial of degree $2$ (in the entries of $U$ and $V$) of expectation $0$, the Hypercontractivity \cite[Proposition $5.48$]{aubrun2017alice} implies for all $k\in \{1,\dots,N\}:\ex{|Z_t|^{k}}\le k^k  \ex{Z_t^{2}}^{k/2}$. Hence, it is sufficient to upper bound the second moment which can be done using the inequalities \eqref{eq:Pu}: $\ex{Z_t^{2}}\le  \cO\left(\frac{\eps^4\log(N)^4}{\din^4\dout^3}\right)$. By a contradiction argument and grouping all these elements, we can prove that $N\log(N)^2\ge \Omega\left(\frac{\din^{2}\dout^{1.5}}{\eps^2}\right)$ and finally $N\ge \Omega\left(\frac{\din^{2}\dout^{1.5}}{\log(\din\dout/\eps)^2\eps^2}\right).$ The detailed proof can be found in App.~\ref{thm-app:LB_DEP_non-adap}.

This proof relies crucially on the non-adaptiveness of the strategy. A natural question arises then, can adaptive strategies outperform their non-adaptive counterpart? We do not settle completely this question in this article. Yet, we propose a lower bound for adaptive strategies showing that, if a separation exists, the advantage would be at most $\cO(\sqrt{\dout})$.
\begin{theorem}\label{thm:LB_DEP_adap}
      Let $\eps\le 1/32$ and $\dout\ge 10$. Any incoherent ancilla-assisted adaptive algorithm for testing identity to the depolarizing channel requires, in the worst case, $N=\Omega\left(\frac{\din^2\dout+\dout^{1.5}}{\eps^2}\right)$   measurements.
\end{theorem}
The proof of this theorem uses the same construction as the one for the non-adaptive lower bound. The main difference is the use of the $\KL$ divergence to compare the observations $(I_1,\dots, I_N)$ under the two hypotheses instead of the $\TV$ distance. On the one hand, the data processing implies that $\KL\left(\mathds{P}_{\cD}^{I_1,\dots, I_N}\Big\| \mathds{E}_{\cN\sim \cP}\mathds{P}_{\cN}^{I_1,\dots, I_N}\right)
     \ge \frac{1}{3}\log(2).$ On the other hand, using Jensen's inequality and  inequalities of the $\log$ function we upper bound the $\KL $ divergence: $\KL\left(\mathds{P}_{\cD}^{I_1,\dots, I_N}\Big\| \mathds{E}_{\cN\sim \cP}\mathds{P}_{\cN}^{I_1,\dots, I_N}\right)\le  \frac{256N\eps^2}{\din^2\dout}. $ These two inequalities implies $\din^2 \dout$ part of the lower bound. The $\dout^{1.5}$ part of the lower bound follows easily  from the hardness of quantum state certification~\citep{chen2022tight-mixedness}. 
The detailed proof can be found in App.~\ref{thm-app:LB_DEP_adap}.

\section{Conclusion and open problems}
We have generalized the problem of testing identity to quantum channels. We have in particular identified the optimal complexity $\Theta(d/\eps^2)$ for testing identity to a unitary channel in the adaptive setting. Moreover, we have shown that the complexity for testing identity to the depolarizing channel in the non-adaptive setting is $\tilde{\Theta}(\din^2\dout^{1.5}/\eps^2)$.
These results open up several interesting questions: can the gap between non-adaptive and adaptive strategies for certification of the depolarizing channel be closed? How to achieve  the instance optimality (as in \citep{valiant2016instance,chen2022toward})? This would allow to adapt the complexity to the tested process $\cN_0$. Another interesting issue deals with the fact that $4$-designs can replace $\Haar$ distributed unitaries in Alg.~\ref{alg-dep}. But can the same complexity be achieved for $3$ (and lower) designs? 
Finally, it would be interesting to consider general strategies allowing entanglement between the uses of the channel, as was done for states in \citep{o2015quantum}.

\section*{Acknowledgments}
{We acknowledge support from the European Research Council (ERC Grant AlgoQIP, Agreement No. 851716) and from the European Union’s Horizon 2020 research and innovation programme under Grant Agreement No 101017733 within the QuantERA II Programme. 
\\Aurélien Garivier acknowledges the support of the Chaire SeqALO (ANR-20-CHIA-0020-01).
\\Aadil Oufkir would like to thank Guillaume Aubrun for helpful discussions and  Moritz Weber for his hospitality and organization of the Focus Semester on Quantum Information at Saarland University.}

\bibliography{bib}

\appendix

\section{Deferred proofs of the analysis of testing identity to identity Alg.~\ref{Alg}}
\begin{lemma}\label{lem:fid-ent-avg}
    Let $\ket{\phi}$ be a random $\Haar$ vector of dimension $d$. We have
    \begin{align*}
    \mathds{E}_\phi \left[ \rF(\cN(\proj{\phi}),\proj{\phi}))\right] =\frac{1+d~\rF(\cJ_{\cN},\proj{\Psi} ) }{1+d}.
\end{align*}
\end{lemma}
\begin{proof}Using Lem.~\ref{lem:Wg} and Lem.~\ref{lem:wg3}, we have 
\begin{align*}
    \mathds{E}_\phi \left[ \rF(\cN(\proj{\phi}),\proj{\phi}))\right]
    &= \sum_k \mathds{E}_U \left[ \bra{0}U^\dagger A_k U\proj{0}U^\dagger A_k^\dagger U\ket{0}\right]
    \\&= \sum_k \frac{(\tr(A_kA_k^\dagger)  +\tr(A_k)\tr(A_k^\dagger)   ) }{d(d+1)}
        \\&=  \frac{d  +\sum_k |\tr(A_k)|^2    }{d(d+1)}=\frac{1+d~\rF(\cJ_{\cN},\proj{\Psi} ) }{1+d}
\end{align*}
where we use Lem.~\ref{ent fid expr}.
\begin{lemma}\label{ent fid expr}
Let $\cN$ be a quantum channel of Kraus operators $\{A_k\}_k$. Let $S=\sum_{k} |\tr(A_k)|^2$.
 We can  relate the average fidelity and $S$ as follows:
\begin{align*}
  \rF(\cJ_{\cN} ,\proj{\Psi})=\frac{S}{d^2}. 
\end{align*}
\end{lemma}
\begin{proof} We have:
    \begin{align*}
  \rF(\cJ_{\cN} ,\proj{\Psi})
  &=\frac{1}{d^2}\sum_{i,j,k,l} \bra{ii}\id\otimes \cN( \ket{kk}\bra{ll})\ket{jj}\notag
    = \frac{1}{d^2}\sum_{i,j,k,l} \bra{i}I\ket{k} \bra{l}\ket{j} \bra{i}\cN( \ket{k}\bra{l})\ket{j}\notag
     \\&= \frac{1}{d^2}\sum_{i,j} \bra{i}\cN( \ket{i}\bra{j})\ket{j}
     = \frac{1}{d^2}\sum_{i,j,k} \bra{i}A_k \ket{i}\bra{j}A_k^\dagger\ket{j}
     = \frac{1}{d^2}\sum_{k} |\tr(A_k)|^2=\frac{S}{d^2}. 
\end{align*}
\end{proof}

\end{proof}
\begin{lemma}\label{lem:fid-diamond}
We have for all quantum channels $\cN$:
\begin{align*}
    \rF(\cJ_{\cN} ,\proj{\Psi})\le 1- \frac{\td(\cN, \id)^2}{4d}\le  1-\frac{\dd(\cN, \id)^4}{16d}.
\end{align*}
\end{lemma}
    \begin{proof}
The following inequality permits to prove the second inequality \citep[Theorem 3.56,  rephrased]{watrous2018theory} 
\begin{align*}
  \td(\cN,\id)\ge \frac{ \dd(\cN,\id)^2}{2}.
\end{align*}
\sloppy It remains to prove the first inequality. For this, let $\eps= \td(\cN,\id)=\max_{\ket{\phi} \in \bS^d} \|\cN(\proj{\phi})-\proj{\phi}\|_1$. 
Let $\ket{\phi}$ be a unit vector satisfying the previous maximization, we show that using Fuchs–van de Graaf inequality \citep{fuchs1999cryptographic}:
\begin{align*}
    \bra{\phi}\cN(\proj{\phi})\ket{\phi}&=\rF(\cN(\proj{\phi}),\proj{\phi}) 
    \le 1-\frac{1}{4}\|\cN(\proj{\phi})-\proj{\phi}\|_1^2
    \le 1-\frac{\eps^2}{4}.
\end{align*}
On the other hand, we use the Kraus decomposition to describe the quantum channel $\cN(\rho)= \sum_{k} A_k\rho A_k^\dagger.$
We can write the previous fidelity in terms of the Kraus operators:
\begin{align*}
    \bra{\phi}\cN(\proj{\phi})\ket{\phi}= \sum_{k} \bra{\phi}A_k\proj{\phi}A_k^\dagger\ket{\phi}=\sum_k |\bra{\phi}A_k\ket{\phi}|^2.
\end{align*}
Hence:
\begin{align}\label{inequ-phi1}
    \sum_k |\bra{\phi}A_k\ket{\phi}|^2\le 1-\frac{\eps^2}{4}.
\end{align}
Let $\ket{\phi_1}=\ket{\phi}$ and we can complete it to have an ortho-normal basis $\{\ket{\phi_i}\}_{i=1}^d$. Moreover, 
we have $ \rF(\cJ_{\cN} ,\proj{\Psi})=\frac{1}{d^2}\sum_{k} |\tr(A_k)|^2$ (Lem.~\ref{ent fid expr}).
By applying the Cauchy-Schwarz inequality and using the inequality~\eqref{inequ-phi1}:
\begin{align*}
    \sum_{k} |\tr(A_k)|^2&= \sum_k \left| \sum_{i=1}^d \bra{\phi_i}A_k\ket{\phi_i}\right|^2
    \le  \sum_{k,i} d|\bra{\phi_i}A_k\ket{\phi_i}|^2
    \\&=
    d\sum_k|\bra{\phi_1}A_k\ket{\phi_1}|^2+ d\sum_{i=2}^d\sum_k|\bra{\phi_i}A_k\ket{\phi_i}|^2
    \\&\le d(1-\eps^2/4)+d(d-1)=d(d-\eps^2/4)
\end{align*}
because for all $i\ge 2$:
\begin{align*}
    \sum_k|\bra{\phi_i}A_k\ket{\phi_i}|^2\le \sum_k \bra{\phi_i}A_k^\dagger A_k\ket{\phi_i}=  \bra{\phi_i}\sum_k A_k^\dagger A_k\ket{\phi_i}=1.
\end{align*}
Finally, $\rF(\cJ_{\cN} ,\proj{\Psi})= \frac{1}{d^2}\sum_{k} |\tr(A_k)|^2\le 1-\frac{\eps^2}{4d}.$
\end{proof}

\section{Lower bound for testing identity to identity}\label{sec:LB}
In this section, we establish a  general lower bound  for any algorithm (possibly adaptive) for testing identity to identity using ancilla-assisted strategies.
\begin{theorem}\label{thm:LB-test-id-ind}
  Any incoherent adaptive ancilla-assisted  strategy  requires a number of steps satisfying:
\begin{align*}
    N=\Omega\left( \frac{d}{\eps^2}\right)
\end{align*}
to distinguish between  $\cN=\id$ and $\dd(\cN, \id)\ge\eps$ (or $\td(\cN, \id)\ge\eps$) with a probability at least $2/3$.
\end{theorem}
This theorem shows that  Alg~\ref{Alg} has an optimal complexity. 
\begin{proof}
 Under the null hypothesis $H_0$, the quantum channel $\cN=\id$. Under the alternate hypothesis $H_1$, we can choose $\cN$ so that $\dd(\cN, \id)\ge \td(\cN, \id)\ge\eps$. A difficult to test  channel is a channel sending almost every vector of a basis to itself. With this intuition, we choose $V\in\Haar(d)$, 
 and construct the channel $\cN_{V}(\rho)= \frac{1}{2}\rho +\frac{1}{2} U_{V} \rho U^\dagger_{V}$ where  $U_{V}$ satisfies:
     \begin{align*}
      U_{V }V\ket{l} =
    \begin{cases*}
      \sqrt{1-\eps^2} V\ket{0}+\eps V\ket{1} & if $l=0$ \\
        \sqrt{1-\eps^2} V\ket{1}-\eps V\ket{0}       & if  $l=1$\\
        V\ket{l}& otherwise.
    \end{cases*}
  \end{align*}
 Taking a mixture of the identity channel and the unitary channel $U_V\cdot U_V^\dagger$ in the definition of $\cN_V$ is crucial in this proof and is inspired by the quantum skew divergence~\citep{skewdivergence}. 
 We need to show first that such a channel is $\eps$-far from the identity channel. Indeed, let $\ket{\phi}= V\ket{0}$, we have:
 \begin{align*}
     \dd(\cN_{V}, \id) &\ge\td(\cN_{V}, \id)\ge   \left\|\ \cN_V(\proj{\phi}) -\id(\proj{\phi})  \right\|_1 \notag
     \\&=   \left\| \frac{1}{2} \proj{\phi} +\frac{1}{2}U_V\proj{\phi}U_V^\dagger- \proj{\phi}   \right\|_1
     \\&= \frac{1}{2} \left\|  V\proj{0}V^\dagger - U_VV\proj{0}V^\dagger U_V ^\dagger \right\|_1
     \\&= \frac{1}{2} \left\|  \proj{0} - \left(\sqrt{1-\eps^2} \ket{0}+\eps \ket{1}\right)\left(\sqrt{1-\eps^2} \bra{0}+\eps \bra{1}\right) \right\|_1
     \\&=  \frac{1}{2}\left\|   \eps^2  \proj{0}       -  \eps\sqrt{1-\eps^2} ( \ket{0}\bra{1} +\ket{1}\bra{0} )-     \eps^2\proj{1}  \right\|_1\notag
     =\eps.
 \end{align*}

Hence a $1/3$-correct algorithm should distinguish between the identity channel and $\cN_{V}$ with at least a probability $2/3$ of success. This algorithm can only choose an input $\rho_t$ which we can suppose (by the convexity of the $\KL$ divergence) of rank $1$, that is $\rho_t= \proj{\psi_t}$   at each step $t$ and perform a measurement using the POVM $\cM_t=\{\lambda_i^t\proj{\phi^t_i}\}_{i\in \cI_t}$ on the output quantum state $\id\otimes \cN(\rho_t)$. Note that the identity channel $\id$ acts on the ancilla space. These choices can depend on the previous observations, that is, the algorithm can be adaptive. Let $I_{\le N}=(I_1,\dots,I_N)$ be the observations of this algorithm where $N$ is a sufficient number of steps to decide correctly with a  probability at least $2/3$. 
We can compare the distributions of the observations under the two hypotheses using the Kullback-Leibler divergence.  
Let $P$ (resp. $Q$) be the distribution of $(I_1,\dots,I_N)$ under $H_0$ (resp. $H_1$). The distribution of $(I_1,\dots,I_N)$ under $H_0$ is: 
 \begin{align*}
     P:=\Bigg\{\prod_{t=1}^N \lambda^t_{i_t}\bra{\phi^t_{i_t}}\rho_t\ket{\phi^t_{i_t}}   \Bigg\}_{i_1,\dots,i_N}=\Bigg\{\prod_{t=1}^N \lambda^t_{i_t}|\spr{\phi^t_{i_t}}{\psi_t}|^2  \Bigg\}_{i_1,\dots,i_N} .
 \end{align*}
 On the other hand, the distribution of $(I_1,\dots,I_N)$ under $H_1$ conditioned on $V$ is:
  \begin{align*}
    Q_V:= \Bigg\{\prod_{t=1}^N \lambda^t_{i_t}\bra{\phi^t_{i_t}} \id\otimes \cN_V(\rho_t)\ket{\phi^t_{i_t}}   \Bigg\}_{i_1, \dots, i_N}.
 \end{align*}
Moreover, we can write each rank one input state and measurement vector as follows:
\begin{align*}
    \ket{\psi_t}= A_t\otimes \dI \ket{w} ~~\text{ and }~~\ket{\phi_{i_t}^t}=  B_{i_t}^t\otimes \dI \ket{w}
\end{align*}
where $\ket{w}=  \sum_{i=1}^{d} \ket{ii}$ and the matrices $A_t \in \dC^{\danc \times d }$ and $B_{i_t}^t \in \dC^{\danc \times d }$ verify: 
\begin{align*}
   \tr(A_tA_t^\dagger )= 1  ~~\text{ and }~~  \tr(B_{i_t}^tB_{i_t}^{t,\dagger} )= 1.
\end{align*}
Note that we have for all $t$, for all $X\in \dC^{\danc\times \danc}$ we have   $\sum_{i_t} \lambda_{i_t}^t B_{i_t}^{t,\dagger}XB_{i_t}^{t} = \tr(X)\dI$. Indeed, the condition of the POVM $\cM_t$ implies: 
\begin{align*}
 X\otimes \dI=    X\otimes \dI\sum_{i_t}\lambda_{i_t}^t \proj{\phi_{i_t}^t} =  \sum_{i_t}\lambda_{i_t}^t XB_{i_t}^t\otimes \dI \proj{w} B_{i_t}^{t, \dagger}\otimes\dI 
\end{align*}
hence by taking the partial trace on the first system we obtain
\begin{align*}
    \tr(X) \dI = \sum_{i_t, i, j}\lambda_{i_t}^t \bra{j}  B_{i_t}^{t, \dagger}X B_{i_t}^{t} \ket{i}  \ket{i}\bra{j}.
\end{align*}
Finally
\begin{align*}
    \sum_{i_t} \lambda_{i_t}^t B_{i_t}^{t,\dagger}XB_{i_t}^{t} = \tr(X)\dI.
\end{align*}
By taking $X=\dI$  and the partial trace on the second system we obtain 
\begin{align*}
    \sum_{i_t} \lambda_{i_t}^t B_{i_t}^{t} B_{i_t}^{t, \dagger} = d~\dI_{\danc}.
\end{align*}
Recall that for an adaptive strategy, for all $t\in [N]$,  $\rho_t$ and $\cM_t=\{\lambda_{i_t}^t\proj{\phi_{i_t}^t}\}_{i_t}$ depend on $(i_1, \dots, i_{t-1})$. 
So, the $\KL$ divergence between $P$ and $Q_V$ can be expressed as follows:
 \begin{align*}
     \KL(P\|Q_V)&=\mathds{E}_{i\sim P} (-\log)\left( \frac{Q_{V,i}}{P_i}\right)
     \\&= \sum_{t=1}^N \mathds{E}_{i\le N} (-\log) \left( \frac{\bra{\phi_{i_t}^t}\id\otimes \cN_V(\rho_t)\ket{\phi_{i_t}^t}}{\bra{\phi_{i_t}^t}\rho_t\ket{\phi_{i_t}^t}}\right)
    \\& = \sum_{t=1}^N \mathds{E}_{i\le t} (-\log) \left( \frac{\bra{\phi_{i_t}^t}\id\otimes \cN_V(\rho_t)\ket{\phi_{i_t}^t}}{\bra{\phi_{i_t}^t}\rho_t\ket{\phi_{i_t}^t}}\right)
 \end{align*}
where we use the notation for $t\in [N]$,  $\mathds{E}_{i\le t}(X(i_1, \dots, i_t)) = \sum_{i_1, \dots, i_t} \prod_{k=1}^t \lambda_{i_k}^k\bra{\phi_{i_k}^k}\rho_k\ket{\phi_{i_k}^k} X(i_1, \dots, i_t) $ and the fact that  the term $\frac{\bra{\phi_{i_t}^t}\id\otimes \cN_V(\rho_t)\ket{\phi_{i_t}^t}}{\bra{\phi_{i_t}^t}\rho_t\ket{\phi_{i_t}^t}}$  depends only on $(i_1, \dots, i_t)$. 
\\Let $\cE$ be the event that the algorithm accepts $H_0$, we apply the Data-Processing inequality on the $\KL$ divergence (see Prop.~\ref{KL-sec}):
 \begin{align*}
     \KL(P\| Q_V)
     &\ge \KL(P\left(\cE\right)\|Q_V\left(\cE\right))
     \\&\ge \KL\left(\frac{2}{3}\Big\|\frac{1}{3}\right)=\frac{2}{3}\log(2)-\frac{1}{3}\log(2)=\frac{1}{3}\log(2)
 \end{align*}
where $\KL(p||q)= \KL(\Ber(p)\|\Ber(q)) $. Hence 
\begin{align*}
      \mathds{E}_{V\sim \Haar(d)}\KL(P\| Q_V)
     \ge \frac{1}{3}\log(2).
\end{align*}

 Let $M_{V}= \dI- U_{V}$  and $S_V= \dI-\frac{1}{2}M_V$. 
 We can write the logarithmic term in the expression of $\KL(P\| Q_V)$ as follows: 
 \begin{align*}
    & (-\log) \left( \frac{\bra{\phi_{i_t}^t}\id \otimes \cN_V(\rho_t)\ket{\phi_{i_t}^t}}{\bra{\phi_{i_t}^t}\rho_t\ket{\phi_{i_t}^t}}\right)
     \\&= (-\log) \left( \frac{\bra{\phi_{i_t}^t}(\frac{1}{2}\rho_t+ \frac{1}{2}(\dI\otimes U_V)\rho_t (\dI\otimes U_V^\dagger))\ket{\phi_{i_t}^t}}{\bra{\phi_{i_t}^t}\rho_t\ket{\phi_{i_t}^t}}\right)
     \\&=(-\log) \left(1-  \frac{\Re(\bra{\phi_{i_t}^t}(\dI\otimes M_V) \rho_t\ket{\phi_{i_t}^t})}{\bra{\phi_{i_t}^t}\rho_t\ket{\phi_{i_t}^t}} + \frac{1}{2}\frac{\bra{\phi_{i_t}^t}(\dI\otimes M_V)\rho_t(\dI\otimes M_V^\dagger)\ket{\phi_{i_t}^t}}{\bra{\phi_{i_t}^t}\rho_t\ket{\phi_{i_t}^t}}     \right)
     \\&=(-\log) \left(1 - \frac{1}{2} \frac{\bra{\phi_{i_t}^t}(\dI\otimes M_V)\rho_t (\dI\otimes S_V^\dagger) \ket{\phi_{i_t}^t}}{\bra{\phi_{i_t}^t}\rho_t\ket{\phi_{i_t}^t}}   -\frac{1}{2} \frac{\bra{\phi_{i_t}^t}(\dI\otimes S_V)\rho_t( \dI\otimes M_V^\dagger) \ket{\phi_{i_t}^t}}{\bra{\phi_{i_t}^t}\rho_t\ket{\phi_{i_t}^t}}    \right).
 \end{align*}
For $t\in [N]$ and $i_{\le t}=(i_1, \dots, i_t)$, define the event $\cG(t, i_{\le t})= \left\{ \bra{\phi_{i_t}^t}\rho_t\ket{\phi_{i_t}^t} \le \frac{\eps^2}{d^2} \tr(A_t^\dagger B_{i_t}^{t}B_{i_t}^{t, \dagger} A_t)\right\}$. 
We can distinguish whether the event $\cG$ is satisfied or not:
\begin{align*}
    \mathds{E}_{V\sim \Haar(d)}\KL(P\| Q_V)&= \sum_{t=1}^N \mathds{E}_{V\sim \Haar(d)} \mathds{E}_{i\le t}(\bm{1}\{\cG(t, i_{\le t})\}+ \bm{1}\{\cG^c(t, i_{\le t})\})(-\log) \left( \frac{\bra{\phi_{i_t}^t}\id \otimes\cN_V(\rho_t)\ket{\phi_{i_t}^t}}{\bra{\phi_{i_t}^t}\rho_t\ket{\phi_{i_t}^t}}\right).
\end{align*}
Let us first analyze the setting when the event $\cG$ holds. Fix $t\in [N]$, observe that we have the inequality:
\begin{align*}
     (-\log) \left( \frac{\bra{\phi_{i_t}^t}\id \otimes\cN_V(\rho_t)\ket{\phi_{i_t}^t}}{\bra{\phi_{i_t}^t}\rho_t\ket{\phi_{i_t}^t}}\right)
     &= (-\log) \left( \frac{\bra{\phi_{i_t}^t}(\frac{1}{2}\rho_t+ \frac{1}{2}(\dI\otimes U_V) \rho_t(\dI\otimes U_V^\dagger))\ket{\phi_{i_t}^t}}{\bra{\phi_{i_t}^t}\rho_t\ket{\phi_{i_t}^t}}\right)\notag
     \\&= (-\log) \left( \frac{1}{2}+ \frac{\bra{\phi_{i_t}^t}(\frac{1}{2}(\dI\otimes U_V) \rho_t (\dI\otimes U_V^\dagger))\ket{\phi_{i_t}^t}}{\bra{\phi_{i_t}^t}\rho_t\ket{\phi_{i_t}^t}}\right)\le \log(2)\le 1.
\end{align*}
Then  we can control the expectation under the event $\cG$ as follows:
\begin{align}\label{eq: E1}
    &\mathds{E}_{V\sim \Haar(d)}\mathds{E}_{i\le t}\bm{1}\{\cG(t, i_{\le t})\}(-\log) \left( \frac{\bra{\phi_{i_t}^t}\id \otimes\cN_V(\rho_t)\ket{\phi_{i_t}^t}}{\bra{\phi_{i_t}^t}\rho_t\ket{\phi_{i_t}^t}}\right) \notag
    \\&\le \mathds{E}_{V\sim \Haar(d)}\mathds{E}_{i\le t-1} \sum_{i_t}\lambda_{i_t}^t \bra{\phi_{i_t}^t}\rho_t\ket{\phi_{i_t}^t}  \bm{1}\{\cG(t, i_{\le t})\}    \notag
    \\&\le \mathds{E}_{V\sim \Haar(d)}\mathds{E}_{i\le t-1}\sum_{i_t}\lambda_{i_t}^t \left( \frac{\eps^2}{d^2}\tr(A_t^\dagger B_{i_t}^{t}B_{i_t}^{t, \dagger} A_t)\right) \bm{1}\{\cG(t, i_{\le t})\}~~~~ 
    \notag
    \\&\le \mathds{E}_{V\sim \Haar(d)}\mathds{E}_{i\le t-1}\sum_{i_t}\lambda_{i_t}^t \left( \frac{\eps^2}{d^2}\tr(A_t^\dagger B_{i_t}^{t}B_{i_t}^{t, \dagger} A_t)\right) \notag
    \\&= \mathds{E}_{V\sim \Haar(d)}\mathds{E}_{i\le t-1}\left( \frac{\eps^2}{d^2}\tr(A_t^\dagger  d A_t)\right)= \mathds{E}_{V\sim \Haar(d)}\mathds{E}_{i\le t-1} \left( \frac{\eps^2}{d}\right) =    \frac{\eps^2}{d}  
\end{align}
where we use the fact that under the event $\cG$ we have $\bra{\phi_{i_t}^t}\rho_t\ket{\phi_{i_t}^t} \le \frac{\eps^2}{d^2}\tr(A_t^\dagger B_{i_t}^{t}B_{i_t}^{t, \dagger} A_t) $,  $\tr(A_t^\dagger   A_t)=1$ and  $\sum_{i_t} \lambda_{i_t}^t B_{i_t}^{t}B_{i_t}^{t, \dagger} = d\,\dI$ which is an implication of the fact that  $\cM_t=\{\lambda_i^t\proj{\phi_i^t}\}_{i\in \cI_t}$ is a POVM. 
\\On the other hand under $\cG^c(t, i_{\le t})$, we will use instead the inequality 
\begin{align*}
    (-\log)(x)\le (1-x)+(x-1)^2\quad \text{valid for all}\quad x\in \Big[\frac{1}{2},+\infty\Big).
\end{align*}
We apply this inequality for $x=\frac{\bra{\phi_{i_t}^t}\id \otimes \cN_V(\rho_t)\ket{\phi_{i_t}^t}}{\bra{\phi_{i_t}^t}\rho_t\ket{\phi_{i_t}^t}}=\frac{1}{2}+ \frac{\bra{\phi_{i_t}^t}(\frac{1}{2}(\dI\otimes U_V)\rho_t (\dI\otimes U_V^\dagger))\ket{\phi_{i_t}^t}}{\bra{\phi_{i_t}^t}\rho_t\ket{\phi_{i_t}^t}}\ge \frac{1}{2} $, the first term of the upper bound is:
\begin{align}\label{eq: E2}
    (1-x)&= 1-\frac{\bra{\phi_{i_t}^t}\id \otimes \cN_V(\rho_t)\ket{\phi_{i_t}^t}}{\bra{\phi_{i_t}^t}\rho_t\ket{\phi_{i_t}^t}} \notag
    \\&= \frac{1}{2} \frac{\bra{\phi_{i_t}^t}(\dI\otimes M_V)\rho_t (\dI\otimes S_V^\dagger) \ket{\phi_{i_t}^t}}{\bra{\phi_{i_t}^t}\rho_t\ket{\phi_{i_t}^t}}   +\frac{1}{2} \frac{\bra{\phi_{i_t}^t}(\dI\otimes S_V)\rho_t (\dI\otimes M_V^\dagger) \ket{\phi_{i_t}^t}}{\bra{\phi_{i_t}^t}\rho_t\ket{\phi_{i_t}^t}} \notag
    \\&=\frac{\Re\bra{\phi_{i_t}^t}(\dI\otimes M_V)\ket{\psi_t} \bra{\psi_t} (\dI\otimes S_V^\dagger) \ket{\phi_{i_t}^t}}{\bra{\phi_{i_t}^t}\rho_t\ket{\phi_{i_t}^t}}  \notag   
     \\&=\frac{\Re\bra{w}( B_{i_t}^{t, \dagger}  A_t \otimes M_V)\ket{w} \bra{w} ( A_t^\dagger  B_{i_t}^{t} \otimes S_V^\dagger) \ket{w}}{\bra{w}( B_{i_t}^{t, \dagger}  A_t \otimes \dI)\ket{w} \bra{w} ( A_t^\dagger  B_{i_t}^{t} \otimes \dI) \ket{w}}  \notag
      \\&=\frac{\Re\,\tr(B_{i_t}^{t, \dagger}  A_t  M_V^\top) \tr( A_t^\dagger  B_{i_t}^{t}  \bar{S}_V)}{\tr(B_{i_t}^{t, \dagger}  A_t)\tr( A_t^\dagger  B_{i_t}^{t})} 
\end{align}
and by using first the inequality $(x+y)^2\le 2(x^2+y^2)$ and then the Cauchy Schwarz inequality applied for the vectors $\sqrt{\rho_t}\ket{\phi_{i_t}^t }$ and $\sqrt{\rho_t}M_V^{\dagger}\ket{\phi_{i_t}^t }$  we can upper bound the second term as follows:
\begin{align}\label{eq: E3}
    (x-1)^2&= \left(\frac{\bra{\phi_{i_t}^t}\id \otimes \cN_V(\rho_t)\ket{\phi_{i_t}^t}}{\bra{\phi_{i_t}^t}\rho_t\ket{\phi_{i_t}^t}}-1\right)^2\notag
    \\&=  \left(  \frac{\Re(\bra{\phi_{i_t}^t}(\dI \otimes M_V) \rho_t\ket{\phi_{i_t}^t})}{\bra{\phi_{i_t}^t}\rho_t\ket{\phi_{i_t}^t}} - \frac{1}{2}\frac{\bra{\phi_{i_t}^t}(\dI \otimes M_V)\rho_t (\dI \otimes M_V^\dagger)\ket{\phi_{i_t}^t}}{\bra{\phi_{i_t}^t}\rho_t\ket{\phi_{i_t}^t}} \right)^2 \notag
     \\&\le 2\left(  \frac{|\bra{\phi_{i_t}^t}(\dI \otimes M_V) \rho_t\ket{\phi_{i_t}^t}|}{\bra{\phi_{i_t}^t}\rho_t\ket{\phi_{i_t}^t}}\right)^2 + 2\left(\frac{1}{2}\frac{\bra{\phi_{i_t}^t}(\dI \otimes M_V)\rho_t (\dI \otimes M_V^\dagger)\ket{\phi_{i_t}^t}}{\bra{\phi_{i_t}^t}\rho_t\ket{\phi_{i_t}^t}}     \right)^2 \notag
      \\&\le 2\left(   \frac{\bra{\phi_{i_t}^t}(\dI \otimes M_V) \rho_t (\dI \otimes M_V^\dagger)\ket{\phi_{i_t}^t}}{\bra{\phi_{i_t}^t}\rho_t\ket{\phi_{i_t}^t}}\right) + 2\left(\frac{1}{2}\frac{\bra{\phi_{i_t}^t}(\dI \otimes M_V)\rho_t  (\dI \otimes M_V^\dagger)\ket{\phi_{i_t}^t}}{\bra{\phi_{i_t}^t}\rho_t\ket{\phi_{i_t}^t}}     \right)^2  \notag
      \\&= 2\frac{\tr(B_{i_t}^{t, \dagger}  A_t  M_V^\top) \tr( A_t^\dagger  B_{i_t}^{t}  \bar{M}_V)}{\tr(B_{i_t}^{t, \dagger}  A_t)\tr( A_t^\dagger  B_{i_t}^{t})} + \frac{1}{2}\left(\frac{\tr(B_{i_t}^{t, \dagger}  A_t  M_V^\top) \tr( A_t^\dagger  B_{i_t}^{t}  \bar{M}_V)}{\tr(B_{i_t}^{t, \dagger}  A_t)\tr( A_t^\dagger  B_{i_t}^{t})}     \right)^2
\end{align}
 Let us compute the expectation of~\eqref{eq: E2}. Let $M, S$ such that $M_V=VMV^\dagger $ and  $S_V=VSV^\dagger $. Concretely 
	
\[
M=\left(\begin{array}{@{}c|c@{}}
  \begin{matrix}
	1-\sqrt{1-\eps^2} & -\eps    \\ 
	\eps &  1-\sqrt{1-\eps^2}    
	\end{matrix} 
  & \bigzero \\
\hline \\
  \bigzero &
  \bigzero_{d-2}  
\end{array}\right) ~~\text{ and }~~ S=\left(\begin{array}{@{}c|c@{}}
 \begin{matrix}
	\frac{1}{2}+\frac{\sqrt{1-\eps^2}}{2} & \frac{\eps}{2}    \\ 
	\frac{- \eps}{2} & \frac{1}{2}+\frac{\sqrt{1-\eps^2}}{2}    
	\end{matrix} 
  & \bigzero \\
\hline \\
  \bigzero &
  \bigeye_{d-2}  
\end{array}\right). 
\]
 Note that $\tr(M)= 2(1-\sqrt{1-\eps^2})$, $\tr(S)=d-1+\sqrt{1-\eps^2} $, 
 $\tr(MS^\dagger)=\tr( M^\dagger S)=0$ and  $MM^\dagger = M+M^\dagger= M^\dagger M $. 
Let $\eps'=(1-\sqrt{1-\eps^2})= \Theta(\eps^2) $, we have by  Weingarten calculus (Lem.~\ref{lem:Wg}):
\begin{align*}
    &\left|\mathds{E}_{V\sim \Haar(d)} \left( \frac{\Re\,\tr(B_{i_t}^{t, \dagger}  A_t  M_V^\top) \tr( A_t^\dagger  B_{i_t}^{t}  \bar{S}_V)}{\tr(B_{i_t}^{t, \dagger}  A_t)\tr( A_t^\dagger  B_{i_t}^{t})}  \right)\right|
    \\&=\left|\frac{1}{\tr(B_{i_t}^{t, \dagger}  A_t)\tr( A_t^\dagger  B_{i_t}^{t})}\Re \mathds{E}_{V\sim \Haar(d)} \left( \sum_{x,y}\bra{x} B_{i_t}^{t, \dagger}  A_t VM^\top V^\dagger \ket{x}\bra{y} A_t^\dagger  B_{i_t}^{t} V S V^\dagger \ket{y}\right)\right|\notag
    \\&=\left|\frac{1}{\tr(B_{i_t}^{t, \dagger}  A_t)\tr( A_t^\dagger  B_{i_t}^{t})}\Re \sum_{\alpha, \beta \in \fS_2} \sum_{x,y}\W(\alpha\beta) \tr_\alpha(M^\top,S)\tr_{\beta (12)}(\ket{x}\bra{y} A_t^\dagger  B_{i_t}^{t}, \ket{y}\bra{x} B_{i_t}^{t, \dagger}  A_t )\right|\notag
    \\&=\left|\Re\left(\frac{\tr(MS^\dagger)(d\tr( A_t^\dagger  B_{i_t}^{t}  B_{i_t}^{t, \dagger}  A_t )- |\tr( A_t^\dagger  B_{i_t}^{t})|^2  ) + \tr(M)\tr(S)(d|\tr( A_t^\dagger  B_{i_t}^{t})|^2-\tr( A_t^\dagger  B_{i_t}^{t}  B_{i_t}^{t, \dagger}  A_t ) ) }{d(d^2-1)\tr(B_{i_t}^{t, \dagger}  A_t)\tr( A_t^\dagger  B_{i_t}^{t})}\right) \right|\notag
   \\&=\left|\Re\left(\frac{ 2\eps'(d-\eps')(d|\tr( A_t^\dagger  B_{i_t}^{t})|^2-\tr( A_t^\dagger  B_{i_t}^{t}  B_{i_t}^{t, \dagger}  A_t ) ) }{d(d^2-1)\tr(B_{i_t}^{t, \dagger}  A_t)\tr( A_t^\dagger  B_{i_t}^{t})}\right) \right|\notag
    \\&\le \frac{2\eps^2 }{d} + \frac{2\eps^2\tr( A_t^\dagger  B_{i_t}^{t}  B_{i_t}^{t, \dagger}  A_t )}{(d^2-1) \tr(B_{i_t}^{t, \dagger}  A_t)\tr( A_t^\dagger  B_{i_t}^{t}) }
\end{align*}
Recall the notation $\mathds{E}_{i\le t}(X(i_1, \dots, i_t)) = \sum_{i_1, \dots, i_t} \prod_{k=1}^t \lambda_{i_k}^k\bra{\phi_{i_k}^k}\rho_k\ket{\phi_{i_k}^k} X(i_1, \dots, i_t) $. If we take the expectation $\mathds{E}_{i\le t}$ under the event $\cG^c(t, i_{\le t})$, we obtain
\begin{align}\label{eq: E2 UB}
     & \mathds{E}_{V\sim \Haar(d)} \mathds{E}_{i\le t} \bm{1}\{\cG^c(t, i_t)\} \left( \frac{\Re\bra{\phi_{i_t}^t}(\dI\otimes M_V)\rho_t (\dI\otimes S_V^\dagger) \ket{\phi_{i_t}^t}}{\bra{\phi_{i_t}^t}\rho_t\notag\ket{\phi_{i_t}^t}}    \right)\notag 
     \\&\le \mathds{E}_{i\le t} \bm{1}\{\cG^c(t, i_t)\} \left|\mathds{E}_{V\sim \Haar(d)} \left( \frac{\Re\bra{\phi_{i_t}^t}(\dI \otimes M_V)\rho_t (\dI\otimes S_V^\dagger) \ket{\phi_{i_t}^t}}{\bra{\phi_{i_t}^t}\rho_t\notag\ket{\phi_{i_t}^t}}    \right)\right|\notag 
     \\&\le \mathds{E}_{i\le t} \bm{1}\{\cG^c(t, i_t)\}\left( \frac{2\eps^2 }{d}+  \frac{2\eps^2\tr( A_t^\dagger  B_{i_t}^{t}  B_{i_t}^{t, \dagger}  A_t )}{(d^2-1) \tr(B_{i_t}^{t, \dagger}  A_t)\tr( A_t^\dagger  B_{i_t}^{t}) }\right)\notag
     \\&\le  \mathds{E}_{i\le t} \left( \frac{2\eps^2 }{d}+  \frac{2\eps^2\tr( A_t^\dagger  B_{i_t}^{t}  B_{i_t}^{t, \dagger}  A_t )}{(d^2-1) \tr(B_{i_t}^{t, \dagger}  A_t)\tr( A_t^\dagger  B_{i_t}^{t}) } \right)\notag
     \\&= \frac{2\eps^2}{d}+ \mathds{E}_{i_{\le t-1}} \sum_{i_t}  \lambda_{i_t}^t \bra{\phi_{i_t}^t}\rho_t\ket{\phi_{i_t}^t}\times  \frac{2\eps^2\tr( A_t^\dagger  B_{i_t}^{t}  B_{i_t}^{t, \dagger}  A_t )}{(d^2-1) \bra{\phi_{i_t}^t}\rho_t\ket{\phi_{i_t}^t} } \notag
    \\ &= \frac{2\eps^2}{d} + \mathds{E}_{i_{\le t-1}} \frac{2\eps^2}{d^2-1 }\sum_{i_t}  \lambda_{i_t}^t\tr( A_t^\dagger  B_{i_t}^{t}  B_{i_t}^{t, \dagger}  A_t )  \notag
 \\ &= \frac{2\eps^2}{d} + \mathds{E}_{i_{\le t-1}} \frac{2\eps^2}{d^2-1 }\tr( A_t^\dagger  d  A_t )  
\le \frac{5\eps^2}{d}
\end{align}
where we use  $\sum_{i_t} \lambda_{i_t}^t B_{i_t}^{t}B_{i_t}^{t, \dagger} = d\,\dI$ and $\tr(A_tA_t^\dagger)=1$. We move to the expectation $\mathds{E}_{i\le t} $ of the first term of~\eqref{eq: E3}, it is non negative so we can safely remove the condition $\bm{1}\{\cG^c(t, i_t)\} $:
 \begin{align}\label{eq: E3-1 UB}
 & \mathds{E}_V\mathds{E}_{i\le t} \bm{1}\{\cG^c(t, i_t)\} \frac{2\bra{\phi_{i_t}^t}(\dI\otimes M_{V})\rho_t (\dI\otimes M_{V}^\dagger)\ket{\phi_{i_t}^t}}{\bra{\phi_{i_t}^t}\rho_t\ket{\phi_{i_t}^t}} \notag
 \\&\le \mathds{E}_V\mathds{E}_{i\le t} \frac{2\bra{\phi_{i_t}^t}(\dI\otimes M_{V})\rho_t (\dI\otimes M_{V}^\dagger)\ket{\phi_{i_t}^t}}{\bra{\phi_{i_t}^t}\rho_t\ket{\phi_{i_t}^t}}\notag
  \\&= \mathds{E}_V\mathds{E}_{i\le t-1}\sum_{i_t}2\lambda_{i_t}^t  \bra{\phi_{i_t}^t}(\dI\otimes M_{V})\rho_t (\dI\otimes  M_{V}^\dagger)\ket{\phi_{i_t}^t}\notag
  \\&= \mathds{E}_{i\le t-1}\mathds{E}_V 2\tr((\dI\otimes M_{V})\rho_t(\dI\otimes M_{V}^\dagger))\notag
  \\& = \mathds{E}_{i\le t-1}\mathds{E}_V 2\tr\big[\left(\dI\otimes (M_{V}+M_{V}^\dagger)\right)\rho_t\big] \notag
   \\&=\mathds{E}_{i\le t-1} \frac{8\eps'}{d}\le  \frac{8\eps^2}{d}
 \end{align}
because $\mathds{E}_{V} M_{V} = \frac{2(1-\sqrt{1-\eps^2})}{d} \dI\preccurlyeq \frac{2\eps^2}{d} \dI$. 
\\Concerning the expectation of the second term of~\eqref{eq: E3}, we apply again the Weingarten calculus (Lem.~\ref{lem:Wg}) after denoting $C=A_t^\dagger B^t_{i_t}$:
\begin{align*}
   & \mathds{E}_{V\sim \Haar(d)}  \frac{1}{2}\left( \frac{\tr(B_{i_t}^{t, \dagger}  A_t  M_V^\top) \tr( A_t^\dagger  B_{i_t}^{t}  \bar{M}_V)}{\bra{\phi_{i_t}^t}\rho_t\ket{\phi_{i_t}^t}}      \right)^2
   =   \frac{1}{2}\frac{\mathds{E}_{V\sim \Haar(d)}\tr(B_{i_t}^{t, \dagger}  A_t  M_V^\top)^2 \tr( A_t^\dagger  B_{i_t}^{t}  \bar{M}_V)^2}{\bra{\phi_{i_t}^t}\rho_t\ket{\phi_{i_t}^t}^2} 
   \\&= \frac{1}{2}\sum_{x,y,z,t}\frac{\mathds{E}_{V\sim \Haar(d)}\tr(\ket{t}        \bra{x}C^\dagger VM^\top V^\dagger\ket{x}\bra{z}C VM V^\dagger\ket{z}\bra{y}C^\dagger VM^\top V^\dagger\ket{y} \bra{t}C VM V^\dagger)  }{\bra{\phi_{i_t}^t}\rho_t\ket{\phi_{i_t}^t}^2} 
   \\&= \frac{1}{2\bra{\phi_{i_t}^t}\rho_t\ket{\phi_{i_t}^t}^2}\sum_{x,y,z,t} \sum_{\alpha, \beta \in \fS_4}\W(\alpha\beta) \tr_{\beta}(M^\top, M, M^\top,M)\tr_{\alpha \gamma} (\ket{t}\bra{x}C^\dagger, \ket{x}\bra{z}C, \ket{z}\bra{y}C^\dagger, \ket{y} \bra{t}C)
\end{align*}
We can remark that for all $\alpha\in \fS_4$ there is $\beta \in \fS_4$ such that :  
\begin{align*}
    \left| \sum_{x,y,z,t} \tr_{\alpha } (\ket{t}\bra{x}C^\dagger, \ket{x}\bra{z}C, \ket{z}\bra{y}C^\dagger, \ket{y} \bra{t}C)\right|&=\left|  \tr_{\beta }(C^\dagger, C, C^\dagger, C)\right|
     \\&\le \max\{\tr(CC^\dagger), |\tr(C)|^2\}^2
\end{align*}
where the last inequality is only non trivial for $\beta=(ijk)$. For instance if $\beta= (123)$ we have by the Cauchy Schwarz inequality:
\begin{align*}
    \tr_\beta(C^\dagger, C, C^\dagger,C)&= \tr(C^\dagger C C^\dagger)\tr(C)
    \le \sqrt{\tr(CC^\dagger CC^\dagger)\tr(CC^\dagger)} |\tr(C)| 
    \\&\le \tr(CC^\dagger)^{3/2} |\tr(C)| 
    \\&\le \max\{\tr(CC^\dagger), |\tr(C)|^2\}^{3/2}\max\{\tr(CC^\dagger), |\tr(C)|^2\}^{1/2}
    \\&= \max\{\tr(CC^\dagger), |\tr(C)|^2\}^2.
\end{align*}
 Moreover, it is clear that when $\beta $ is not a $4$-cycle, we have $|\tr_{\beta}(M, M^\top, M, M^\top)|\le \cO(\eps^4)$ since it can be written as a product of at least two elements each of them is $\cO(\eps^2)$. In the case $\beta$ is a $4$ cycle we have $\tr(MM^\top MM^\top )= \tr(MMM^\top M^\dagger )= \tr((M+M^\top)^2)= 2(2-2\sqrt{1-\eps^2})^2\le 8\eps^4$. On the other hand,  we know that for all $(\alpha, \beta)\in \fS_4^2$, $|\W(\alpha\beta)|\le \frac{2}{d^4}$~\citep{collins2006integration} so 
 \begin{align*}
     &\sum_{x,y,z,t} \sum_{\alpha, \beta \in \fS_4}\W(\alpha\beta) \tr_{\beta}(M^\top, M, M^\top,M)\tr_{\alpha \gamma} (\ket{t}\bra{x}C^\dagger, \ket{x}\bra{z}C, \ket{z}\bra{y}C^\dagger, \ket{y} \bra{t}C)
     \\&\le \cO\left(\frac{\eps^4}{d^4}\max\{\tr(CC^\dagger), |\tr(C)|^2\}^2\right).
 \end{align*}
 Therefore we have:
\begin{align*}
   &  \mathds{E}_{V\sim \Haar(d)}  \frac{1}{2}\left( \frac{\tr(B_{i_t}^{t, \dagger}  A_t  M_V^\top) \tr( A_t^\dagger  B_{i_t}^{t}  \bar{M}_V)}{\bra{\phi_{i_t}^t}\rho_t\ket{\phi_{i_t}^t}}      \right)^2
   \le \cO\left( \frac{\eps^4\max\{\tr(CC^\dagger), |\tr(C)|^2\}^2}{ d^4\bra{\phi_{i_t}^t}\rho_t\ket{\phi_{i_t}^t}^2}\right).
\end{align*}
 Recall that $C=A_t^\dagger B^t_{i_t}$, now, if we take the expectation $\mathds{E}_{i\le t} $ under the event \[\cG^c(t, i_{\le t})= \left\{ \bra{\phi_{i_t}^t}\rho_t\ket{\phi_{i_t}^t} > \frac{\eps^2}{d^2} \tr(A_t^\dagger B_{i_t}^{t}B_{i_t}^{t, \dagger} A_t)\right\} = \left\{ |\tr(C)|^2 > \frac{\eps^2}{d^2} \tr(CC^\dagger)\right\},\] we obtain:
\begin{align}\label{eq: E3-2 UB}
   & \mathds{E}_{i_{\le t}}\mathds{E}_{V\sim \Haar(d)}\bm{1}\{\cG^c(t, i_{\le t})\} \frac{1}{2}\left( \frac{\tr(B_{i_t}^{t, \dagger}  A_t  M_V^\top) \tr( A_t^\dagger  B_{i_t}^{t}  \bar{M}_V)}{\bra{\phi_{i_t}^t}\rho_t\ket{\phi_{i_t}^t}} \right)^2\notag
   \\&\le \mathds{E}_{i_{\le t}}\bm{1}\{\cG^c(t, i_{\le t})\}\cO\left( \frac{\eps^4\max\{\tr(CC^\dagger), |\tr(C)|^2\}^2}{ d^4\bra{\phi_{i_t}^t}\rho_t\ket{\phi_{i_t}^t}^2} \right)\notag
   \\&\le \mathds{E}_{i_{\le t}}\bm{1}\{\cG^c(t, i_{\le t})\}\cO\left( \frac{\eps^4|\tr(C)|^4}{d^4\bra{\phi_{i_t}^t}\rho_t\ket{\phi_{i_t}^t}^2}+\frac{\eps^4}{ d^4\bra{\phi_{i_t}^t}\rho_t\ket{\phi_{i_t}^t}}\times \frac{\tr(CC^\dagger)^2}{ \bra{\phi_{i_t}^t}\rho_t\ket{\phi_{i_t}^t}}  \right)\notag
    \\&= \mathds{E}_{i_{\le t}}\bm{1}\{\cG^c(t, i_{\le t})\}\cO\left( \frac{\eps^4}{d^4}+\frac{\eps^4}{ d^4\bra{\phi_{i_t}^t}\rho_t\ket{\phi_{i_t}^t}}\times \frac{\tr(CC^\dagger)^2}{ |\tr(C)|^2}\right)\notag
   \\&\le \mathds{E}_{i_{\le t}}\bm{1}\{\cG^c(t, i_{\le t})\}\cO\left( \frac{\eps^4}{d^4}+\frac{\eps^4}{ d^4\bra{\phi_{i_t}^t}\rho_t\ket{\phi_{i_t}^t}}\times \frac{d^2\tr(CC^\dagger)}{\eps^2}\right) ~~~~  \left(\text{under } \cG^c : |\tr(C)|^2 > \frac{\eps^2}{d^2} \tr(CC^\dagger)\right)\notag  
   \\&\le\mathds{E}_{i_{\le t-1}} \sum_{i_t}  \lambda_{i_t}^t \bra{\phi_{i_t}^t}\rho_t\ket{\phi_{i_t}^t}\cO\left(\frac{\eps^4}{d^4}+ \frac{\eps^2\tr(A_t^\dagger B^t_{i_t}B_{i_t}^{t, \dagger} A_t)}{ d^2\bra{\phi_{i_t}^t}\rho_t\ket{\phi_{i_t}^t}}\right) \notag
   \\&\le \mathds{E}_{i_{\le t-1}}  \cO\left( \frac{\eps^4}{ d^4}\right)+ \cO\left( \frac{\eps^2\sum_{i_t}  \lambda_{i_t}^t\tr(A_t^\dagger B^t_{i_t}B_{i_t}^{t, \dagger} A_t)}{ d^2}\right) \notag
    \\&\le \mathds{E}_{i_{\le t-1}}  \cO\left( \frac{\eps^4}{ d^4}\right)+ \cO\left( \frac{\eps^2\tr(A_t^\dagger d A_t)}{ d^2}\right) 
   = \mathds{E}_{i_{\le t-1}}  \cO\left( \frac{d\eps^2}{ d^2}\right)=\cO\left( \frac{\eps^2}{ d}\right)
\end{align}
where we use  $\sum_{i_t} \lambda_{i_t}^t B_{i_t}^{t}B_{i_t}^{t, \dagger} = d\,\dI$ and $\tr(A_tA_t^\dagger)=1$. By adding up \eqref{eq: E2 UB}, \eqref{eq: E3-1 UB} and \eqref{eq: E3-2 UB}, we obtain:
\begin{align*}
   & \mathds{E}_{i_{\le t}}\mathds{E}_{V\sim \Haar(d)}\bm{1}\{\cG^c(t, i_{\le t})\}(-\log) \left( \frac{\bra{\phi_{i_t}^t}\id \otimes \cN_V(\rho_t)\ket{\phi_{i_t}^t}}{\bra{\phi_{i_t}^t}\rho_t\ket{\phi_{i_t}^t}}\right)
    \\&\le \mathds{E}_{i_{\le t}}\mathds{E}_{V\sim \Haar(d)} \bm{1}\{\cG^c(t, i_{\le t})\}\left(\frac{\Re(\bra{\phi_{i_t}^t}(\dI\otimes M_V)\rho_t(\dI\otimes S_V^\dagger) \ket{\phi_{i_t}^t})}{\bra{\phi_{i_t}^t}\rho_t\ket{\phi_{i_t}^t}}    + \frac{1}{2}\frac{\bra{\phi_{i_t}^t}(\dI\otimes M_V)\rho_t(\dI\otimes M_V^\dagger)\ket{\phi_{i_t}^t}^2}{\bra{\phi_{i_t}^t}\rho_t\ket{\phi_{i_t}^t}^2}     \right)
    \\&+\mathds{E}_{i_{\le t}}\mathds{E}_{V\sim \Haar(d)} \bm{1}\{\cG^c(t, i_{\le t})\}  \left(   2\frac{\bra{\phi_{i_t}^t}(\dI\otimes M_V) \rho_t(\dI\otimes M_V^\dagger)\ket{\phi_{i_t}^t}}{\bra{\phi_{i_t}^t}\rho_t\ket{\phi_{i_t}^t}}\right)=  \cO\left(\frac{\eps^2}{ d}\right).
\end{align*}
Therefore using this upper bound and the upper bound  \eqref{eq: E1} we get an upper bound on the expected $\KL$ divergence:
\begin{align*}
\mathds{E}_{V\sim \Haar(d)}\KL(P\| Q_V)&=  \sum_{t=1}^N\mathds{E}_{i\le t}\mathds{E}_{V\sim \Haar(d)}(\bm{1}\{\cG(t, i_{\le t})\}+ \bm{1}\{\cG^c(t, i_{\le t})\})(-\log) \left( \frac{\bra{\phi_{i_t}^t}\id \otimes \cN_V(\rho_t)\ket{\phi_{i_t}^t}}{\bra{\phi_{i_t}^t}\rho_t\ket{\phi_{i_t}^t}}\right)
     \\&\le\sum_{t=1}^N \cO\left(\frac{\eps^2}{d}\right)+\cO\left(\frac{\eps^2}{d}\right)= \cO\left(\frac{N\eps^2}{d}\right).
\end{align*}
Finally since $\mathds{E}_{V\sim \Haar(d)}\KL(P\| Q_V)\ge \frac{\log(2)}{3}$ we conclude:
\begin{align*}
    N=\Omega\left( \frac{d}{\eps^2}\right).
\end{align*}
\end{proof}

\section{Deferred proofs of the analysis of testing identity to the depolarizing channel Alg.~\ref{alg-dep}}\label{app:analysis-dep}
\subsection{Proof of Lem.~\ref{lem:diamond-2-choi}} \label{app-lem:diamond-2-choi}
\begin{lemma}
    Let $\cN_1$ and $\cN_2$ be two  quantum channels and $\din$ (resp. $\dout$) be the dimension of their input (resp. output) states. We have:
\begin{align*}
 d_\diamond(\cN_1,\cN_2)\le   \din\sqrt{\dout}\|\cJ_{\cN_1}-\cJ_{\cN_2}  \|_2.
\end{align*}
\end{lemma}
\begin{proof}
    Denote by $\cM=\cN_1-\cN_2$ and $\cJ=\cJ_\cM= \cJ_{\cN_1}-\cJ_{\cN_2} $. Let $\ket{\phi}$ be a maximizing unit vector of the diamond norm, i.e., $\|\id \otimes \cM(\proj{\phi})\|_1= d_\diamond(\cN_1,\cN_2)$.
We can write $\ket{\phi}= A\otimes \dI \ket{\Psi} $ where $\ket{\Psi}=\frac{1}{\sqrt{\din}}\sum_{i=1}^{\din} \ket{i}\otimes \ket{i}$ is the maximally entangled state. $\ket{\phi}$ has norm $1$ so $\tr(A^\dagger A)=\din\bra{\Psi}A^\dagger A \otimes \dI \ket{\Psi} =  \din\spr{\phi}{\phi}=\din. $ We can write the diamond distance as follows: 
\begin{align*}
d_\diamond(\cN_1,\cN_2)  &= \|\id \otimes \cM(\proj{\phi})\|_1 =  \|\id\otimes \cM(A\otimes \dI\proj{\Psi} A^\dagger \otimes \dI)\|_1 
 \\&= \|(A\otimes \dI) \id\otimes \cM(\proj{\Psi})( A^\dagger \otimes \dI)\|_1 
 = \|(A\otimes \dI)\cJ_{\cM}( A^\dagger \otimes \dI)\|_1.
\end{align*}
$\cJ_{\cM}$ is Hermitian so can be written as : $\cJ_{\cM}=\sum_i \lambda_i \proj{\psi_i}$. Using the triangle inequality and the Cauchy Schwarz inequality, we obtain:
\begin{align*}
&\|(A\otimes \dI)\cJ_{\cM}( A^\dagger \otimes \dI)\|_1
= \left\|(A\otimes \dI)\sum_i \lambda_i \proj{\psi_i}( A^\dagger \otimes \dI)\right\|_1
\\&\le \sum_i |\lambda_i| \|(A\otimes \dI) \proj{\psi_i}( A^\dagger \otimes \dI)\|_1
\le \sqrt{\sum_i \lambda_i^2}\sqrt{\sum_i \bra{\psi_i} ( A^\dagger A \otimes \dI)\ket{\psi_i}^2}
\\&\le \|\cJ\|_2\sqrt{\sum_i \bra{\psi_i} ( A^\dagger A A^\dagger A \otimes \dI)\ket{\psi_i}}
= \|\cJ_\cM\|_2\sqrt{\tr ( A^\dagger A A^\dagger A \otimes \dI)}
\\&= \|\cJ_\cM\|_2\sqrt{\dout\tr ( A^\dagger A A^\dagger A )}
\le \|\cJ_\cM\|_2\sqrt{\dout (\tr ( A^\dagger A))^2}= \din\sqrt{\dout}\|\cJ_\cM\|_2.
\end{align*}
\end{proof}

\subsection{Proof of Lem.~\ref{lem:2-2 norms}}
\begin{lemma}\label{lem-app:2-2 norms}
Let $\cJ= \id\otimes \cN(\proj{\Psi})$ be the Choi state corresponding to the channel $\cN$.
\begin{align*}
    \ex{\left\|\cN(\proj{\phi}) -\frac{\dI}{\dout}\right\|_2^2}&=\frac{\left( \tr\left(\cN(\dI) -\frac{\din}{\dout}\dI\right)^2  +\din^2\left\|\cJ -\frac{\dI}{\din\dout}\right\|_2^2\right)}{\din(\din+1)} 
   \ge \frac{\din}{\din+1}\left\|\cJ -\frac{\dI}{\din\dout}\right\|_2^2 .
\end{align*}
\end{lemma}
\begin{proof}
We write
\begin{align*}
    \ex{\left\|\cN(\proj{\phi}) -\frac{\dI}{\dout}\right\|_2^2}= \ex{\tr(\cN(\proj{\phi})^2)} -\frac{1}{\dout}.
\end{align*}
then if we use the Kraus decomposition of the quantum channel  $\cN(\rho)=\sum_k A_k \rho A_k^\dagger$, we can compute the following expectation using Weingarten calculus (Lem.~\ref{lem:Wg} and Lem.~\ref{lem:wg3}):
\begin{align*}
    \ex{\tr(\cN(\proj{\phi})^2) }&= \ex{\tr\left(\sum_k A_k \proj{\phi}A_k^\dagger\right)^2}= \sum_{k,l} \ex{\tr(A_k\proj{\phi}A_k^\dagger A_l \proj{\phi}A_l^\dagger) }
    \\&= \frac{1}{\din(\din+1)} \left( \sum_{k,l} \tr(A_k^\dagger A_l A_l^\dagger A_k) +\sum_{k,l}|\tr(A_kA_l^\dagger)|^2\right)
    \\&= \frac{1}{\din(\din+1)} \left(  \tr\left(\sum_{k}A_kA_k^\dagger \right)^2 +\sum_{k,l}|\tr(A_kA_l^\dagger)|^2\right)
    \\&= \frac{1}{\din(\din+1)} \left( \tr(\cN(\dI)^2)+\sum_{k,l}|\tr(A_kA_l^\dagger)|^2\right)
\end{align*}
Observe that $\tr(\cN(\dI)^2)= \tr\left(\cN(\dI) -\frac{\din}{\dout}\dI\right)^2 +\frac{\din^2}{\dout}$. 
Moreover, 
\begin{align*}
    \tr(\cJ^2)&= \tr\left(\frac{1}{\din }\sum_{i,j} \ket{i}\bra{j}\otimes \cN(\ket{i}\bra{j})  \right)^2
   = \frac{1}{\din^2} \sum_{i,j} \tr(\cN(\ket{i}\bra{j}) \cN(\ket{j}\bra{i}))
    \\&= \frac{1}{\din^2} \sum_{i,j,k,l} \tr(A_k\ket{i}\bra{j}A_k^\dagger A_l\ket{j}\bra{i}A_l^\dagger)= \frac{1}{\din^2} \sum_{k,l} |\tr(A_k^\dagger A_l)|^2
\end{align*}
hence: 
\begin{align*}
    \ex{\tr(\cN(\proj{\phi})^2) }&= \ex{\tr\left(\sum_k A_k \proj{\phi}A_k^\dagger\right)^2}
    \\&=  \frac{1}{\din(\din+1)} \left(\tr\left(\cN(\dI) -\frac{\din}{\dout}\dI\right)^2 +\frac{\din^2}{\dout} +\din^2\tr(\cJ^2)\right)
\end{align*}
Finally, 
\begin{align*}
    \ex{\left\|\cN(\proj{\phi}) -\frac{\dI}{\dout}\right\|_2^2}&= \ex{\tr(\cN(\proj{\phi})^2)} -\frac{1}{\dout}
    \\&= \frac{1}{\din(\din+1)} \left( \tr\left(\cN(\dI) -\frac{\din}{\dout}\dI\right)^2 +\frac{\din^2}{\dout} +\din^2\tr(\cJ^2)\right) -\frac{1}{\dout}
    \\&= \frac{1}{\din(\din+1)} \left( \tr\left(\cN(\dI) -\frac{\din}{\dout}\dI\right)^2  +\din^2\tr(\cJ^2)-\frac{\din}{\dout}\right) 
    \\&=\frac{1}{\din(\din+1)} \left( \tr\left(\cN(\dI) -\frac{\din}{\dout}\dI\right)^2  +\din^2\left\|\cJ -\frac{\dI}{\din\dout}\right\|_2^2\right).
\end{align*}
\end{proof}
\subsection{Proof of Thm~\ref{thm:PZ-D}}\label{app-thm:PZ-D}
\begin{theorem}\label{thm-app:PZ-D} Let $\ket{\phi} $ be  a  $\Haar$  distributed  vector in $\bS^d$. Let  $X=\left\|\cN(\proj{\phi}) -\frac{\dI}{\dout}\right\|_2^2$. We have:
    \begin{align*}
    \var(X)= \cO\left( \ex{X}^2\right). 
\end{align*}
\end{theorem}
\begin{proof}
 Recall that $X=\left\|\cN(\proj{\phi}) -\frac{\dI}{\dout}\right\|_2^2$.
We can observe that:
\begin{align*}
    \var(X)= \var\left(\tr \left(\cN(\proj{\phi}) \right)^2 -\frac{1}{\dout} \right)= \var\left(\tr \left(\cN(\proj{\phi}) \right)^2\right).
\end{align*}
    We use the Kraus notation $\cN(\rho)= \sum_{k} A_k\rho A_k^\dagger$ and $d=\din$.  By the Weingarten calculus (Lem.~\ref{lem:Wg} and Lem.~\ref{lem:wg3}), we can compute the expectation:
\begin{align*}
    &\ex{\left(\tr \left(\cN(\proj{\phi}) \right)^2 \right)^2 } = \ex{ \left(\tr \sum_{k,l } A_k \proj{\phi} A_k^\dagger A_l \proj{\phi} A_l^\dagger \right)^2 }
    \\&=\sum_{i, j} \sum_{k,l,k',l'} \ex{  \tr A_{l'}^\dagger \ket{j}\bra{i}A_k \proj{\phi} A_k^\dagger A_l \proj{\phi} A_l^\dagger \ket{i}\bra{j} A_{k'} \proj{\phi} A_{k'}^\dagger  A_{l' }\proj{\phi}   }
    \\&= \sum_{i, j} \sum_{k,l,k',l'} \frac{1}{d(d+1)(d+2)(d+3)}\sum_{\alpha\in \fS_4} \tr_\alpha(A_{l'}^\dagger \ket{j}\bra{i}A_k, A_k^\dagger A_l , A_l^\dagger \ket{i}\bra{j} A_{k'} , A_{k'}^\dagger  A_{l' } )
    \\&=\frac{1}{d(d+1)(d+2)(d+3)} \sum_{\alpha\in \fS_4} \sum_{i,j,k,l,k',l'}  \tr_\alpha(A_{l'}^\dagger \ket{j}\bra{i}A_k, A_k^\dagger A_l , A_l^\dagger \ket{i}\bra{j} A_{k'} , A_{k'}^\dagger  A_{l' } )
    \\&= \frac{1}{d(d+1)(d+2)(d+3)} \sum_{\alpha\in \fS_4} F(\alpha)
\end{align*}
where for $\alpha\in \fS_4$ we adopt the notation $F(\alpha)= \sum_{i,j,k,l,k',l'}\tr_\alpha(A_{l'}^\dagger \ket{j}\bra{i}A_k, A_k^\dagger A_l , A_l^\dagger \ket{i}\bra{j} A_{k'} , A_{k'}^\dagger  A_{l'} )$ where  $\tr_{\alpha}(M_1,\dots,M_n)=\Pi_j \tr(\Pi_{i\in C_j} M_i)$ for $\alpha=\Pi_j C_j $ and $C_j$ are cycles. It is thus necessary to control each of these $24$ terms in order to upper bound the variance. Furthermore, we need to be careful so that our upper bounds on $\{F(\alpha)\}_{\alpha\in \fS_4}$ depend on the actual parameters of the testing problem. Recall that the expected value of $X$ can be expressed as follows: 
\begin{align*}
    \ex{X}=\frac{1}{\din(\din+1)} \left( \tr\left(\cN(\dI) -\frac{\din}{\dout}\dI\right)^2  +\din^2\left\|\cJ -\frac{\dI}{\din\dout}\right\|_2^2\right).
\end{align*}
Let us define $\cM=\cN-\cD$, $m=\|\cM(\dI_{\din})\|_2= \left\|\cN(\dI)-\frac{\din}{\dout}\dI\right\|_2$ and $\eta= \din\left\|\cJ -\frac{\dI}{\din\dout}\right\|_2$. We state a useful Lemma relating $\eta$, $\cM$ and the Kraus operators $\{A_k\}_k$:
\begin{lemma}\label{lem: eta= }
	Let $\eta=\din\|\cJ-\dI/(\din\dout)\|_2$ and $\cM=\cN-\cD$, 
	we have:
		\begin{itemize}
		\item $\eta^2=\sum_{k,l} |\tr(A_k^\dagger A_l)|^2 -\frac{\din}{\dout}$,
		\item $\eta^2=\sum_{x,y=1}^{\din} \|\cM(\ket{x}\bra{y})\|_2^2$.
	\end{itemize}
\end{lemma}
\begin{proof}
	Recall that we use the Kraus representation of the channel $\cN(\rho)=\sum_k A_k \rho A_k^\dagger$. We can express $\eta^2$:
	 \begin{align*}
	\eta^2&= \left\| \din\cJ-\frac{\dI}{\dout}\right\|_2^2= \din^2\tr(\cJ^2)-\frac{\din}{\dout}= \tr\left( \sum_{i,j,k} \ket{i}\bra{j}\otimes A_k \ket{i}\bra{j}A_k^\dagger\right)^2-\frac{\din}{\dout}
	\\&= \sum_{i,j,k,l} \tr( A_k \ket{i}\bra{j}A_k^\dagger A_l  \ket{j}\bra{i}A_l^\dagger  )-\frac{\din}{\dout}= \sum_{k,l} |\tr(A_k^\dagger A_l)|^2 -\frac{\din}{\dout}.
	\end{align*}
	We move to the second point, we have $\cJ-\dI/(\din\dout)=\cJ_\cN-\cJ_\cD= \cJ_\cM = \id\otimes\cM(\proj{\Psi} )$ so
	\begin{align*}
	\eta^2&= \din^2 \tr( \id \otimes \cM (\proj{\Psi} ) )^2 =  \tr( \id \otimes \cM (\din\proj{\Psi} ) )^2
	= \tr\left(\sum_{x,y} \ket{x}\bra{y}\otimes \cM(\ket{x}\bra{y})\right)^2
	\\&= \sum_{x,y} \tr(\cM(\ket{x}\bra{y})\cM(\ket{y}\bra{x})        )=\sum_{x,y} \|   \cM(\ket{x}\bra{y}) \|_2^2.
	\end{align*} 
\end{proof}
On the other hand, when dealing with some $F(\alpha)$'s, we will need to have some properties of the matrix $\sum_k A_k\otimes \Bar{A}_k$.
\begin{lemma}\label{lem: MM*}
	Let $M= \sum_k A_k \otimes \bar{A}_k$. Let $\{\lambda_i\}_i$ be the set of the eigenvalues of $M^\dagger M$ (in a decreasing order) corresponding to the eigenstates $\{\ket{\phi_i}\}_i$. 
 We have:
	\begin{itemize}
		\item $\sum_i \lambda_i = \frac{\din}{\dout}+\eta^2$,
		\item $\lambda_1\ge \frac{\din}{\dout}, \sum_{i>1}\lambda_i \le \eta^2$,
		\item $ \frac{\din^2}{\dout^2}(1-|\spr{\phi_1}{\Psi}|^2)\le  \frac{2m^2\eta^2}{\din}+2\eta^4$,
    \item $\bra{\Psi} MM^\dagger \ket{\Psi}= \frac{\din}{\dout}$.
	\end{itemize}
 \end{lemma}
\begin{proof}
 We have $\sum_i \lambda_i = \tr(M^\dagger M) = \sum_{k,l} \tr(A_k^\dagger A_l)\tr(A_k^{\top} \bar{A}_l) = \sum_{k,l} |\tr(A_k^\dagger A_l)|^2 = \frac{\din}{\dout}+\eta^2 $ by Lem.~\ref{lem: eta= }.
 Recall the definition of the unnormalized maximally entangled state $\sqrt{\din}\ket{\Psi}= \sum_i \ket{i}\otimes \ket{i}$, we can compute its image by the matrix $M^\dagger M$:
  \begin{align*}
  M^\dagger M(\sqrt{\din}\ket{\Psi}) &=\sum_{i,k,l} A_k^\dagger A_l \ket{i}\otimes A_k^{\top} \bar{A}_l\ket{i} 
  = \sum_{x,y,i,k,l} \bra{x}A_k^\dagger A_l \ket{i} \bra{y}A_k^{\top} \bar{A}_l \ket{i} \ket{xy}
  \\&= \sum_{x,y,i,k,l} \bra{x}A_k^\dagger A_l \ket{i} \bra{i}A_l^\dagger A_k\ket{y} \ket{xy}=\sum_{x,y,k} \bra{x}A_k^\dagger \cN(\dI)A_k\ket{y} \ket{xy}
  \end{align*}
  therefore 
  \begin{align*}
  \bra{\Psi}M^\dagger M\ket{\Psi}= \frac{1}{\din}\sum_{i,k} \bra{i}A_k^\dagger \cN(\dI)A_k\ket{i}= \frac{\tr(\cN(\dI)^2)}{\din}\ge \frac{\din^2}{\din\dout}=\frac{\din}{\dout}
  \end{align*}
  where we used the Cauchy-Schwarz inequality. This implies that the largest eigenvalue verifies $\lambda_1\ge \frac{\din}{\dout}$ thus $\sum_{i>1}\lambda_i = \frac{\din}{\dout}+\eta^2-\lambda_1\le \eta^2$. 
  
  We move to prove the third point. Recall the notation $\cM(\rho)= (\cN- \cD)(\rho)=\sum_k A_k\rho A_k^\dagger  -\tr(\rho)\frac{\dI}{\dout}$. We have on the one hand:
   \begin{align*}
  M^\dagger M(\sqrt{\din}\ket{\Psi}) &=\sum_{x,y,k} \bra{x}A_k^\dagger \cN(\dI)A_k\ket{y} \ket{xy}=\sum_{x,y} \tr(\cN(\dI)\cN(\ket{y}\bra{x})) \ket{xy}
  \\&= \sum_{x,y} \tr(\cN(\dI)\cM(\ket{y}\bra{x})) \ket{xy}+ \sum_{x,y} \tr(\cN(\dI)\cD(\ket{y}\bra{x})) \ket{xy}
    \\&= \sum_{x,y} \tr(\cM(\dI)\cM(\ket{y}\bra{x})) \ket{xy}+ \sum_{x} \frac{1}{\dout}\tr(\cN(\dI)\dI) \ket{xx}
     \\&= \sum_{x,y} \tr(\cM(\dI)\cM(\ket{y}\bra{x})) \ket{xy}+ \frac{\din}{\dout}\sqrt{\din}\ket{\Psi}.
  \end{align*}
  On the other hand, using the spectral decomposition of $M^\dagger M$, we can write:
    \begin{align*}
 \sum_{x,y} \tr(\cM(\dI)\cM(\ket{y}\bra{x})) \ket{xy}+ \frac{\din}{\dout}\sqrt{\din}\ket{\Psi}= M^\dagger M(\sqrt{\din}\ket{\Psi}) &=\sum_{i} \lambda_i \sqrt{\din}\spr{\phi_i}{\Psi}\ket{\phi_i}. 
  \end{align*}
  Therefore 
  \begin{align*}
  \lambda_1 \spr{\phi_1}{\Psi}\ket{\phi_1}- \frac{\din}{\dout}\ket{\Psi}= \frac{1}{\sqrt{\din}}  \sum_{x,y} \tr(\cM(\dI)\cM(\ket{y}\bra{x})) \ket{xy} - \sum_{i>1} \lambda_i \spr{\phi_i}{\Psi}\ket{\phi_i}.
  \end{align*}
  Taking the $2$-norm squared on both sides, then applying the triangle inequality, the fact $(x+y)^2\le 2(x^2+y^2)$  and the Cauchy-Schwarz inequality we obtain:
  \begin{align*}
   &\lambda_1^2| \spr{\phi_1}{\Psi}|^2+\frac{\din^2}{\dout^2}-2\frac{\din}{\dout} \lambda_1 |\spr{\phi_1}{\Psi}|^2\\&=\left\|   \frac{1}{\sqrt{\din}}  \sum_{x,y} \tr(\cM(\dI)\cM(\ket{y}\bra{x})) \ket{xy} - \sum_{i>1} \lambda_i \spr{\phi_i}{\Psi}\ket{\phi_i}     \right\|_2^2
   \\&\le  \left(  \left\|   \frac{1}{\sqrt{\din}}  \sum_{x,y} \tr(\cM(\dI)\cM(\ket{y}\bra{x})) \ket{xy}  \right\|_2+\left\|   \sum_{i>1} \lambda_i \spr{\phi_i}{\Psi}\ket{\phi_i}     \right\|_2\right)^2
   \\&\le 2 \left\|   \frac{1}{\sqrt{\din}}  \sum_{x,y} \tr(\cM(\dI)\cM(\ket{y}\bra{x})) \ket{xy}  \right\|_2^2+2\left\|   \sum_{i>1} \lambda_i \spr{\phi_i}{\Psi}\ket{\phi_i}     \right\|_2^2
   \\&\le \frac{2}{\din} \sum_{x,y} |\tr(\cM(\dI)\cM(\ket{y}\bra{x}))|^2 +2\sum_{i>1} \lambda_i^2
   \\&\le  \frac{2}{\din} \sum_{x,y} \|\cM(\dI)\|_2^2 \|\cM(\ket{y}\bra{x})\|_2^2 +2\left(\sum_{i>1} \lambda_i\right)^2
   \le \frac{2m^2\eta^2}{\din}+2\eta^4.
  \end{align*}
  Observe that the LHS can be lower bounded as follows:
  \begin{align*}
    \lambda_1^2| \spr{\phi_1}{\Psi}|^2+\frac{\din^2}{\dout^2}-2\frac{\din}{\dout} \lambda_1 |\spr{\phi_1}{\Psi}|^2&= \left(\lambda_1-\frac{\din}{\dout}\right)^2 |\spr{\phi_1}{\Psi}|^2 +\frac{\din^2}{\dout^2}- \frac{\din^2}{\dout^2}|\spr{\phi_1}{\Psi}|^2\\&\ge\frac{\din^2}{\dout^2}( 1-|\spr{\phi_1}{\Psi}|^2).
  \end{align*}
  Finally, we deduce from the two previous inequalities:
  \begin{align*}
 \frac{\din^2}{\dout^2}( 1-|\spr{\phi_1}{\Psi}|^2)\le  \frac{2m^2\eta^2}{\din}+2\eta^4.
  \end{align*}
  We move to the fourth point.  We have:
	\begin{align*}
	MM^\dagger(\sqrt{\din}\ket{\Psi}) &=\sum_{i,k,l} A_kA_l^\dagger  \ket{i}\otimes \bar{A}_kA_l^{\top}\ket{i} = \sum_{x,y,i,k,l} \bra{x}A_kA_l^\dagger  \ket{i} \bra{y}\bar{A}_kA_l^{\top}\ket{i} \ket{xy}
	\\&= \sum_{x,y,i,k,l} \bra{x}A_kA_l^\dagger  \ket{i} \bra{i}A_lA_k^\dagger \ket{y} \ket{xy}=\sum_{x,y} \bra{x}\cN(\dI)\ket{y} \ket{xy}
	\\&= \sum_{x,y} \bra{x}\cM(\dI)\ket{y} \ket{xy}+  \frac{\din}{\dout}\sqrt{\din}\ket{\Psi}.
	\end{align*}
 Hence:
 \begin{align*}
	\bra{\Psi}MM^\dagger\ket{\Psi} &=\frac{1}{\din}\sum_{x} \bra{x}\cM(\dI)\ket{x} +  \frac{\din}{\dout}\spr{\Psi}{\Psi}=\frac{1}{\din}\tr(\cM(\dI))+ \frac{\din}{\dout}= \frac{\din}{\dout}.
	\end{align*}
	\end{proof}
 The first and fourth points will be used in the proof of Lemma~\ref{lem: F(12)(34))}. The first three points  imply that the matrix $M^\dagger M$ (when normalized) is close to the maximally entangled state in the $1$-norm.
\begin{lemma}\label{lem: M-Psi}Let $M= \sum_k A_k \otimes \bar{A}_k$ and $\ket{\Psi}=\frac{1}{\sqrt{\din}}\sum_{i=1}^{\din} \ket{i}\otimes \ket{i}$.
 We have:
\begin{align*}
\left\|M^\dagger M-  \frac{\din}{\dout}\proj{\Psi}\right\|_1
&\le 5\eta^2.
\end{align*}
\end{lemma}
\begin{proof} Let $\ket{\phi_1}$ the eigenvector of $M^\dagger M$ corresponding to the largest eigenvalue.
	Using  the Fuchs–van de Graaf inequality \citep{fuchs1999cryptographic} and Lem.~\ref{lem: MM*}:
	\begin{align*}
	\|\proj{\phi_1}-\proj{\Psi}\|_1\le 2\sqrt{1-|\spr{\phi_1}{\Psi}|^2}\le 2\frac{\dout}{\din}\sqrt{\frac{2m^2\eta^2}{\din}+2\eta^4} \le 4\frac{\dout}{\din}\eta^2
	\end{align*}
	where we use the Cauchy Schwarz inequality and Lem.~\ref{lem: eta= }:
	\begin{align*}
	m^2&=\tr(\cM(\dI)^2) =\sum_{i,j} \tr(\cM(\proj{i})    \cM(\proj{j}) ) 
 \\&\le \sum_{i,j} \tr(\cM(\proj{i})^2) = \din \sum_{i} \tr(\cM(\proj{i})^2) \le \din\eta^2.
	\end{align*}
	By the triangle inequality and Lem.~\ref{lem: MM*} we deduce:
\begin{align*}
 \left\|M^\dagger M- \frac{\din}{\dout}\proj{\Psi}\right\|_1
&\le \left\| M^\dagger M- \frac{\din}{\dout}\proj{\phi_1}\right\|_1+  \frac{\din}{\dout}\big\| \proj{\phi_1}- \proj{\Psi}\big\|_1
\\&\le  \lambda_1- \frac{\din}{\dout} +\sum_{i>1} \lambda_i +4\eta^2= 5\eta^2.
\end{align*}
\end{proof}
This Lemma will be used in the proofs of Lemmas~\ref{lem: F(142)}, \ref{lem: F(24)}, \ref{lem: F(id)} and \ref{lem: F(14)}. 
We move now to upper bound different values of the function $F$.
\begin{lemma}\label{lem: F(142)}
	We can upper bound $F((142))$ and $F((243))$ as follows:
 \begin{align*}
 F((142))&= F((243))\le \frac{\din/\dout+\eta^2}{\dout}+ \frac{\eta^2}{\dout}+5 \eta^4.
 \end{align*} 
\end{lemma}
\begin{proof}
Recall the notation $\cM(\rho)= (\cN- \cD)(\rho)=\sum_k A_k\rho A_k^\dagger -\tr(\rho)\frac{\dI}{\dout}$. We will first write $F((142))$ as a sum of an ideal term reflecting the null hypothesis ($\cN=\cD$) and an error term reflecting the difference between $\cN$ and $\cD$. The ideal term is computed exactly and depends on dimensions $\din $ and $\dout$. The error term can also be splited to a simple error depending on $\eta$ and $\dout$ and a more involved term that depends on the Kraus operators $\{A_k\}_k$ and the difference of channels $\cM$. To control this latter error, we first write it in a closed form in terms of $\cM$, $M=\sum_{k}A_k\otimes A_k^\dagger $ and the flip operator $\F$. Then we can use the spectral decomposition of the matrix $M^\dagger M - \frac{\din}{\dout}\proj{\Psi}$ in order to decompose this error term into a combination of negligible elements. The final step requires to control the $\ell_1$ norm of the coefficients of this combination which is done using Lemma~\ref{lem: M-Psi}.

We have:
\begin{align*}
F((142)) &=  \sum_{k,l,k',l'}  \tr(A_k   A_{k'}^\dagger  A_{l' }     A_k^\dagger  A_lA_{l'}^\dagger A_{k'} A_l^\dagger  )
= \sum_{k,l} \tr(\cN(A_l^\dagger A_k)\cN(A_k^\dagger A_l)) 
\\&= \sum_{k,l} \tr(\cM(A_l^\dagger A_k)\cN(A_k^\dagger A_l))  + \tr(\cD(A_l^\dagger A_k)\cN(A_k^\dagger A_l))
\\&= \sum_{k,l} \tr(\cM(A_l^\dagger A_k)\cM(A_k^\dagger A_l))  + \tr(\cD(A_l^\dagger A_k)\cN(A_k^\dagger A_l))+ \tr(\cM(A_l^\dagger A_k)\cD(A_k^\dagger A_l)) 
\\&= \sum_{k,l} \tr(\cM(A_l^\dagger A_k)\cM(A_k^\dagger A_l))  + \frac{\tr(A_l^\dagger A_k)}{\dout}\tr(\cN(A_k^\dagger A_l))+ \frac{\tr(A_k^\dagger A_l)}{\dout}\tr(\cM(A_l^\dagger A_k)) 
\\&= \sum_{k,l} \tr(\cM(A_l^\dagger A_k)\cM(A_k^\dagger A_l))  + \frac{\tr(A_l^\dagger A_k)}{\dout}\tr(A_k^\dagger A_l)
\\&= \sum_{k,l} \tr(\cM(A_l^\dagger A_k)\cM(A_k^\dagger A_l))  +  \frac{\eta^2+\din/\dout}{\dout}.
\end{align*}
It remains to control the sum $\sum_{k,l} \tr(\cM(A_l^\dagger A_k)\cM(A_k^\dagger A_l))  $ .   Let $\T_2: X\otimes Y \mapsto X\otimes Y^{\top} $ be the partial transpose operator and $M=\sum_l A_l\otimes \bar{A}_l$. Let $\F=\sum_{i,j=1}^d  (\ket{i}\otimes \ket{j})(  \bra{j}\otimes \bra{i})$ be the flip operator, we have $\tr(A\otimes B \F)=\sum_{i,j,k,l} A_{i,j} B_{k,l} \tr(\ket{i}\bra{j}\otimes \ket{k}\bra{l} \F)= \sum_{i,j,k,l} A_{i,j} B_{k,l} \tr(\ket{i}\bra{l}\otimes\ket{k}\bra{j} )= \sum_{i,j }A_{i,j}B_{j,i}=\tr(AB)$
which is known as the replica trick. We have using the replica trick:
\begin{align*}
\sum_{k,l} \tr(\cM(A_l^\dagger A_k)\cM(A_k^\dagger A_l))   &= \sum_{k,l} \tr(\cM(A_l^\dagger A_k)\otimes \cM(A_k^\dagger A_l) \F)  
\\&= \tr\left(\cM\otimes \cM\left(\sum_{k,l}  A_l^\dagger A_k\otimes A_k^\dagger A_l\right) \F\right)  
\\&=  \tr\left(\cM\otimes \cM\circ\T_2\left(\sum_{k,l}  A_l^\dagger A_k\otimes A_l^{\top} \bar{A}_k\right) \F\right) 
\\&=  \tr\left(\cM\otimes \cM\circ\T_2(M^\dagger M)\F\right) 
\end{align*}
Let $\ket{\phi} $ be a unit vector, we can write $\ket{\phi}= \sum_{x,y} \phi_{x,y} \ket{x}\otimes \ket{y}$ then we can express:
\begin{align*}
|\tr\left(\cM\otimes \cM\circ\T_2(\proj{\phi})\F\right) | &= \left| \sum_{x,y,z,t} \phi_{x,y}\bar{\phi}_{z,t} \tr\left(\cM\otimes \cM\circ\T_2(\ket{x}\bra{z}\otimes \ket{y}\bra{t} )\F\right)\right|
\\&=  \left| \sum_{x,y,z,t} \phi_{x,y}\bar{\phi}_{z,t} \tr\left(\cM\otimes \cM (\ket{x}\bra{z}\otimes \ket{t}\bra{y} )\F\right)\right|
\\&=  \left| \sum_{x,y,z,t} \phi_{x,y}\bar{\phi}_{z,t} \tr\left(\cM(\ket{x}\bra{z})\otimes \cM(\ket{t}\bra{y} )\F\right)\right|
\\&=  \left| \sum_{x,y,z,t} \phi_{x,y}\bar{\phi}_{z,t} \tr\left(\cM(\ket{x}\bra{z}) \cM(\ket{t}\bra{y} )\right)\right|
\\&\le   \sqrt{\sum_{x,y,z,t} |\phi_{x,y}\bar{\phi}_{z,t}|^2 \sum_{x,y,z,t}  |\tr\left(\cM(\ket{x}\bra{z}) \cM(\ket{t}\bra{y} )\right) |^2}
\\&=  \sqrt{ \sum_{x,y,z,t} \| \cM(\ket{x}\bra{z})\|_2^2  \| \cM(\ket{t}\bra{y})\|_2^2         }=\eta^2
\end{align*}
where we use the Cauchy Schwarz inequality and Lem.~\ref{lem: eta= }.
On the other hand, we can compute:
\begin{align*}
\tr(\cM\otimes \cM \circ \T_2(\proj{\Psi}) \F)&= \frac{1}{\din}\sum_{i,j} \tr(\cM\otimes\cM \circ \T_2(\ket{ii}\bra{jj} ) \F)
\\&= \frac{1}{\din}\sum_{i,j} \tr(\cM(\ket{i}\bra{j})  \cM(\ket{j}\bra{i}))
= \frac{\eta^2}{\din}.\end{align*}
Then, we can decompose the Hermitian matrix $M^\dagger M-\frac{\din}{\dout} \proj{\Psi}=\sum_i \mu_i \proj{\psi_i}$.
Hence Lem.~\ref{lem: M-Psi} implies:
\begin{align*}
&\sum_{k,l} \tr(\cM(A_l^\dagger A_k)\cM(A_k^\dagger A_l)) =   \tr\left(\cM\otimes \cM\circ\T_2(M^\dagger M)\F\right) 
\\&= \frac{\din}{\dout}\tr\left(\cM\otimes \cM\circ\T_2(\proj{\Psi})\F\right) +  \sum_{i} \mu_i \tr\left(\cM\otimes \cM\circ\T_2(\proj{\psi_i})\F\right) 
\\&\le \frac{\eta^2}{\dout}+ \sum_i |\mu_i| \eta^2 
=  \frac{\eta^2}{\dout}+ \tr\left|M^\dagger M-\frac{\din}{\dout}\proj{\Psi}  \right| \eta^2 
\le\frac{\eta^2}{\dout}+5 \eta^4.
\end{align*}
Finally,
\begin{align*}
F((142)) &= \sum_{k,l} \tr(\cM(A_l^\dagger A_k)\cM(A_k^\dagger A_l))  + \frac{\frac{\din}{\dout}+\eta^2}{\dout}
\le \frac{\frac{\din}{\dout}+\eta^2}{\dout}+ \frac{\eta^2}{\dout}+ 5\eta^4.
\end{align*}
\end{proof}


\begin{lemma}\label{lem: F(id)}
	We can upper bound $F(\id),F((132)),F((314))$ and $F((24)(13))$ as follows:
	\begin{align*}
	F(\id)&= F((132))=F((314))=F((24)(13)) \le \frac{\din/\dout+\eta^2}{\dout} + \frac{\eta^2}{\dout}+5\eta^4.
	\end{align*} 
\end{lemma}
\begin{proof}
For the  identity permutation, the function $F$ can be expressed as follows:
\begin{align*}
&F(\id) = \sum_{k,l,k',l'} \tr(A_kA_{l'}^\dagger  A_{k'}A_l^\dagger )\tr( A_{k'}^\dagger  A_{l' } ) \tr( A_k^\dagger  A_l)
\\&=  \sum_{k,l,k',l'}\sum_{i,j} \tr(  \ket{j}\bra{i} A_{l'}^\dagger  A_{k'}A_l^\dagger A_k    \ket{i}\bra{j} A_k^\dagger A_l )\tr( A_{k'}^\dagger  A_{l' } ) 
\\&= \sum_{k',l'}\sum_{i,j} \tr(  \cN(\ket{j}\bra{i} A_{l'}^\dagger  A_{k'}) \cN( \ket{i}\bra{j} ) )\tr( A_{k'}^\dagger  A_{l' } ) 
\\&= \sum_{k',l'}\sum_{i,j} \tr(  \cN(\ket{j}\bra{i} A_{l'}^\dagger  A_{k'}) \cM( \ket{i}\bra{j} ) )\tr( A_{k'}^\dagger  A_{l' } ) +\sum_{k',l'}\sum_{i,j} \tr(  \cN(\ket{j}\bra{i} A_{l'}^\dagger  A_{k'}) \cD( \ket{i}\bra{j} ) )\tr( A_{k'}^\dagger  A_{l' } ) 
\\&= \sum_{k',l'}\sum_{i,j} \tr(  \cN(\ket{j}\bra{i} A_{l'}^\dagger  A_{k'}) \cM( \ket{i}\bra{j} ) )\tr( A_{k'}^\dagger  A_{l' } ) +\sum_{k',l'}\sum_{i}\frac{1}{\dout} \tr(  \cN(\ket{i}\bra{i} A_{l'}^\dagger  A_{k'}) )\tr( A_{k'}^\dagger  A_{l' } ) 
\\&= \sum_{k',l'}\sum_{i,j} \tr(  \cN(\ket{j}\bra{i} A_{l'}^\dagger  A_{k'}) \cM( \ket{i}\bra{j} ) )\tr( A_{k'}^\dagger  A_{l' } ) +\sum_{k',l'}\frac{1}{\dout} \tr(A_{l'}^\dagger  A_{k'} )\tr( A_{k'}^\dagger  A_{l' } ) 
\\&= \sum_{k',l'}\sum_{i,j} \tr(  \cM(\ket{j}\bra{i} A_{l'}^\dagger  A_{k'}) \cM( \ket{i}\bra{j} ) )\tr( A_{k'}^\dagger  A_{l' } ) +\frac{\din/\dout+\eta^2}{\dout}                 \text{   (because $\cM$ is a  trace less channel)}
\\&= \sum_{i,j} \tr(  \cM(\ket{j}\bra{i} N) \cM( \ket{i}\bra{j} ) )+\frac{\din/\dout+\eta^2}{\dout}  
\end{align*}
where $N =  \sum_{k,l}   \tr(A_k^\dagger A_l)A_l^\dagger A_k  $. 
Let us introduce $\tilde{N} = N-\frac{\dI}{\dout} = \sum_i \mu_x \proj{\psi_x}$ (this is possible because $N$ is Hermitian) so that we can write using Lem.~\ref{lem: eta= }:
\begin{align*}
\sum_{i,j} \tr(  \cM(\ket{j}\bra{i} N) \cM( \ket{i}\bra{j} ) ) &= \sum_{i,j} \tr(  \cM(\ket{j}\bra{i} \tilde{N}) \cM( \ket{i}\bra{j} ) ) + \frac{1}{\dout} \sum_{i,j} \tr(  \cM(\ket{j}\bra{i} ) \cM( \ket{i}\bra{j} ) )
\\&=\sum_{i,j,x} \mu_x \tr(  \cM(\ket{j}\bra{i} \proj{\psi_x}) \cM( \ket{i}\bra{j} ) ) + \frac{\eta^2}{\dout}
\\&= \sum_{i,j,x} \mu_x \spr{i}{\psi_x}\tr(  \cM(\ket{j}\bra{\psi_x}) \cM( \ket{i}\bra{j} ) ) + \frac{\eta^2}{\dout}
\\&= \sum_{i,j,x,k} \mu_x   \tr(  \spr{i}{\psi_x}\cM(\ket{j}\bra{k}) \spr{\psi_x}{k}  \cM( \ket{i}\bra{j} ) ) + \frac{\eta^2}{\dout}
\\&\le \frac{1}{2}\sum_{i,j,x,k} |\mu_x| (  \| \spr{i}{\psi_x}\cM(\ket{j}\bra{k})  \|_2^2 + \| \spr{\psi_x}{k}  \cM( \ket{i}\bra{j} )\|_2^2)+ \frac{\eta^2}{\dout}
\\&= \tr|\tilde{N}| \eta^2 + \frac{\eta^2}{\dout}.
\end{align*}
We can see the matrix $N= \sum_{k,l}   \tr(A_k^\dagger A_l)A_l^\dagger A_k $ as a partial trace of $M^\dagger M$: 
\begin{align*}
\tr_2(M^\dagger M)&= \tr_2 \left(\sum_{k,l}  A_l^\dagger A_k\otimes A_l^{\top}\bar{A}_k\right)= \sum_{k,l}  \tr(A_l^{\top}\bar{A}_k)A_l^\dagger A_k
\\&= \sum_{k,l}  \tr(A_lA^*_k)A_l^\dagger A_k=N.
\end{align*}
Moreover $\frac{\dI}{\din}= \tr_2(\proj{\Psi})$ so by the data processing inequality (the partial trace is a valid quantum channel) and Lem.~\ref{lem: M-Psi}, we deduce:
\begin{align*}
\tr|\tilde{N}| &= \left\|\tr_2(M^\dagger M )-\frac{\din}{\dout}\tr_2(\proj{\Psi})\right\|_1 \le  \left\|M^\dagger M -\frac{\din}{\dout}\proj{\Psi}\right\|_1\le 5\eta^2.
\end{align*}
Finally, 
\begin{align*}
F(\id) &= \sum_{i,j} \tr(  \cM(\ket{j}\bra{i} N) \cM( \ket{i}\bra{j} ) )+\frac{\din/\dout +\eta^2}{\dout}   \le \frac{\din/\dout +\eta^2}{\dout} + \frac{\eta^2}{\dout}+5\eta^4. 
\end{align*}
\end{proof}

\begin{lemma}\label{lem: F(24)}
We can upper bound $F((24))$ and $F((1432))$ as follows:
\begin{align*}
F((24))&=F((1432)) \le \frac{\din^2}{\dout^2}+\frac{2m^2}{\dout}+25\eta^4. 
\end{align*} 
\end{lemma}
\begin{proof}
Recall that $\bra{\Psi}M^\dagger M\ket{\Psi} = \frac{\tr(\cN(\dI)^2)}{\din}= \frac{\din}{\dout}+\frac{m^2}{\din}$. We use the fact that $\tr(X)\tr(Y)=\tr(X\otimes Y)$: 
\begin{align*}
F((24))
&= \sum_{k,l,k',l'} \tr(A_k^\dagger  A_lA_{k'}^\dagger  A_{l' })\tr(A_kA_{l'}^\dagger  A_{k'}A_l^\dagger )
\\&= \sum_{k,l,k',l'} \tr(A_k^\dagger  A_lA_{k'}^\dagger  A_{l' } \otimes A_l^\dagger A_kA_{l'}^\dagger  A_{k'})
\\&= \tr\left( \sum_{k,l} A_k^\dagger A_l\otimes A_l^\dagger A_k\right)\left( \sum_{k,l} A_k^\dagger A_l\otimes A_l^\dagger A_k\right)
\\&= \tr\left( \sum_{k,l} A_k^\dagger A_l\otimes A_k^{\top}\bar{A_l}\right)^{\T_2}\left( \sum_{k,l} A_k^\dagger A_l\otimes A_k^{\top}\bar{A_l}\right)^{\T_2}
\\&= \tr(M^\dagger M )^{\T_2}(M^\dagger M )^{\T_2}
= \tr(M^\dagger M )(M^\dagger M )
\\&=  \tr\left(M^\dagger M -\frac{\din}{\dout}\proj{\Psi}\right)M^\dagger M  + \frac{\din}{\dout}\tr(\proj{\Psi}M^\dagger M )
\\&=  \tr\left(M^\dagger M -\frac{\din}{\dout}\proj{\Psi}\right)\left(M^\dagger M -\frac{\din}{\dout}\proj{\Psi}\right)
\\&+\frac{\din}{\dout}\bra{\Psi}\left(M^\dagger M -\frac{\din}{\dout}\proj{\Psi} \right)\ket{\Psi} + \frac{\din^2}{\dout^2}+\frac{m^2}{\dout}
\\&\le\left(\tr\left|M^\dagger M -\frac{\din}{\dout}\proj{\Psi}\right|\right)^2  +\frac{\din^2}{\dout^2}+\frac{2m^2}{\dout}
\\&\le  \frac{\din^2}{\dout^2}+\frac{2m^2}{\dout}+25\eta^4   
\end{align*}
where we have used the fact that $\tr\left(A^{\T_2}  B^{\T_2}\right)=\tr(AB)  $ and Lem.~\ref{lem: M-Psi}.
\end{proof}

\begin{lemma} \label{lem: F(12)(34))}
We can upper bound $F((12)(34)), F((14)(23)), F((234))$ and $ F((124))$ as follows:
	\begin{align*}
	F((12)(34))&= F((14)(23))=F((234))= F((124))
\le \frac{\din^3}{\dout^2}+2\frac{\din}{\dout}m^2+\eta^2 m^2.
		\end{align*}
\end{lemma} 
\begin{proof}
Using the expression of $F(\alpha)$ for the permutation $\alpha=(12)(34)$ and the fact that $\sum_k A_k^\dagger A_k=\dI$, we have: 
\begin{align*}
&F((12)(34))=   \sum_{k,l,k',l'}  \tr(A_k    A_k^\dagger  A_l A_{l'}^\dagger A_{k'}    A_{k'}^\dagger  A_{l' }A_l^\dagger ) 
=   \sum_{k,l}  \tr(A_l^\dagger  \cN(\dI) A_l A_{k}^\dagger \cN(\dI) A_{k }) 
\\&=  \sum_{k,l}  \tr(A_l^\dagger  \cM(\dI) A_l A_{k}^\dagger \cN(\dI) A_{k })  +  \frac{\din}{\dout}\sum_{k,l}  \tr(A_l^\dagger   A_l A_{k}^\dagger \cN(\dI) A_{k })  
\\&= \sum_{k,l}  \tr(A_l^\dagger  \cM(\dI) A_l A_{k}^\dagger \cM(\dI) A_{k }) +  \frac{\din}{\dout}\sum_{k,l}  \tr(A_l^\dagger  \cM(\dI) A_l A_{k}^\dagger A_{k })      +  \frac{\din}{\dout} \sum_{k}  \tr(A_{k}^\dagger \cN(\dI) A_{k })  
\\&= \sum_{k,l}  \tr(A_l^\dagger  \cM(\dI) A_l A_{k}^\dagger \cM(\dI) A_{k }) + \frac{\din}{\dout}  \tr(\cM(\dI)\cN(\dI))  +  \frac{\din}{\dout}\tr(\cN(\dI)^2)  
\\&= \sum_{k,l}  \tr(A_l^\dagger  \cM(\dI) A_l A_{k}^\dagger \cM(\dI) A_{k }) +   \frac{\din^3}{\dout^2}  +  2\frac{\din}{\dout}\tr(\cM(\dI)^2)  
\end{align*}
Then, if we focus on the first term,

we can use the replica trick again to obtain:
\begin{align*}
& \sum_{k,l}  \tr(A_l^\dagger  \cM(\dI) A_l A_{k}^\dagger \cM(\dI) A_{k }) = \sum_{k,l}\tr(A_kA_l^\dagger \cM(\dI) \otimes A_lA_k^\dagger  \cM(\dI) \F)
\\&=
\tr\left( \sum_{k,l}A_kA_l^\dagger \otimes A_lA_k^\dagger   \cdot\cM(\dI)^{\otimes 2} \F\right)
= \tr\left(\T_2\left( \sum_{k,l}A_kA_l^\dagger \otimes \bar{A_k}A_l^{\top}\right)  \cdot\cM(\dI)^{\otimes 2} \F\right) 
\\&= \tr\left(\left( MM^\dagger \right)^{T_2}  \cdot\cM(\dI)^{\otimes 2} \F\right) 
\\&= \tr\left(\left( (\dI-\proj{\Psi}) MM^\dagger(\dI-\proj{\Psi}) \right)^{T_2}  \cdot\cM(\dI)^{\otimes 2} \F\right)  
+ \tr\left(\left(  MM^\dagger\proj{\Psi} \right)^{T_2}  \cdot\cM(\dI)^{\otimes 2} \F\right) 
\\&+\tr\left(\left( \proj{\Psi} MM^\dagger \right)^{T_2}  \cdot\cM(\dI)^{\otimes 2} \F\right) -\tr\left(\left( \proj{\Psi} MM^\dagger\proj{\Psi} \right)^{T_2}  \cdot\cM(\dI)^{\otimes 2} \F\right).
\end{align*}
We can simplify the latter terms. First we have $\proj{\Psi}^{T_2}=\F$  and $\F^2=\dI$ so
\begin{align*}
   & \tr\left(\left( \proj{\Psi} MM^\dagger\proj{\Psi} \right)^{T_2}  \cdot\cM(\dI)^{\otimes 2} \F\right)= \bra{\Psi}MM^\dagger \ket{\Psi} \tr\left(\left( \proj{\Psi}\right)^{T_2}  \cdot\cM(\dI)^{\otimes 2} \F\right)
    \\&= \bra{\Psi}MM^\dagger \ket{\Psi} \tr\left( \F \cM(\dI)^{\otimes 2} \F\right)=\bra{\Psi}MM^\dagger \ket{\Psi} \tr\left(  \cM(\dI)^{\otimes 2} \right)
    = \bra{\Psi}MM^\dagger \ket{\Psi} \tr\left(  \cM(\dI) \right)^2=0.
\end{align*}
Next by Lem.~\ref{lem: MM*} we have $MM^\dagger\ket{\Psi} = \frac{1}{\sqrt{\din}} \sum_{x,y} \bra{x}\cN(\dI)\ket{y} \ket{xy}$ so
\begin{align*}
    \tr\left(\left(  MM^\dagger\proj{\Psi} \right)^{T_2}  \cdot\cM(\dI)^{\otimes 2} \F\right) &=\frac{1}{\din}\sum_{x,y, z} \bra{x}\cN(\dI)\ket{y}  \tr\left(\left(  \ket{xy}\bra{zz} \right)^{T_2}  \cdot\cM(\dI)^{\otimes 2} \F\right) 
    \\&= \frac{1}{\din}\sum_{x,y, z} \bra{x}\cN(\dI)\ket{y}  \tr\left( \ket{x}\bra{z} \otimes \ket{z}\bra{y}   \cdot\cM(\dI)^{\otimes 2} \F\right) 
    \\&= \frac{1}{\din}\sum_{x,y, z} \bra{x}\cN(\dI)\ket{y}  \tr\left( \ket{x}\bra{z}\cM(\dI) \ket{z}\bra{y} \cM(\dI) \right)
     \\&= \frac{1}{\din}\sum_{x,y} \bra{x}\cN(\dI)\ket{y}  \tr\left( \ket{x}\tr(\cM(\dI)) \bra{y} \cM(\dI) \right)=0. 
\end{align*}
Similarly we prove:
\begin{align*}
    \tr\left(\left(  \proj{\Psi}MM^\dagger \right)^{T_2}  \cdot\cM(\dI)^{\otimes 2} \F\right) &=\frac{1}{\din}\sum_{x,y, z} \bra{x}\cN(\dI)\ket{y}  \tr\left(\left(  \ket{zz}\bra{xy} \right)^{T_2}  \cdot\cM(\dI)^{\otimes 2} \F\right) 
    \\&= \frac{1}{\din}\sum_{x,y, z} \bra{x}\cN(\dI)\ket{y}  \tr\left( \ket{z}\bra{x} \otimes \ket{y}\bra{z}   \cdot\cM(\dI)^{\otimes 2} \F\right) 
    \\&= \frac{1}{\din}\sum_{x,y, z} \bra{x}\cN(\dI)\ket{y}  \tr\left( \ket{z}\bra{x}\cM(\dI) \ket{y}\bra{z} \cM(\dI) \right)
     \\&= \frac{1}{\din}\sum_{x,y} \bra{x}\cN(\dI)\ket{y}   \bra{y} \cM(\dI)\ket{x} \tr(\cM(\dI)) =0. 
\end{align*}
Now the matrix $(\dI-\proj{\Psi}) MM^\dagger(\dI-\proj{\Psi}) $ is Hermitian and positive semi-definite so can be written as $(\dI-\proj{\Psi}) MM^\dagger(\dI-\proj{\Psi}) =\sum_i\lambda_i \proj{\phi_i}$, and for each $i$, 
we can write the Schmidt's decomposition of $\ket{\phi}= \sum_{x,y} \phi_{x,y} \ket{xy}$. Therefore
\begin{align*}
    & \sum_{k,l}  \tr(A_l^\dagger  \cM(\dI) A_l A_{k}^\dagger \cM(\dI) A_{k }) = \tr\left(\left( (\dI-\proj{\Psi}) MM^\dagger(\dI-\proj{\Psi}) \right)^{T_2}  \cdot\cM(\dI)^{\otimes 2} \F\right) 
    \\&=\sum_i\lambda_i \tr\left(\T_2\left(\proj{\phi_i} \right)  \cdot\cM(\dI)^{\otimes 2} \F\right)
    = \sum_i\sum_{x,y,z,t}\lambda_i \phi_{x,y} \bar{\phi}_{z,t}\tr\left(\T_2\left(\ket{xy}\bra{zt} \right)  \cdot\cM(\dI)^{\otimes 2} \F\right)
    \\&= \sum_i\sum_{x,y,z,t}\lambda_i \phi_{x,y} \bar{\phi}_{z,t}\tr\left(\ket{x}\bra{z}\otimes \ket{t}\bra{y} \cdot\cM(\dI)^{\otimes 2} \F\right)
    \\&= \sum_i\sum_{x,y,z,t}\lambda_i \phi_{x,y} \bar{\phi}_{z,t}\tr\left(\ket{x}\bra{z}\cM(\dI)\otimes \ket{t}\bra{y} \cM(\dI) \F\right)
    \\&= \sum_i\sum_{x,y,z,t}\lambda_i \phi_{x,y} \bar{\phi}_{z,t}\tr\left(\ket{x}\bra{z}\cM(\dI) \ket{t}\bra{y} \cM(\dI) \right)
     \\&= \sum_i\sum_{x,y,z,t}\lambda_i \phi_{x,y} \bar{\phi}_{z,t}\bra{z}\cM(\dI) \ket{t}\bra{y} \cM(\dI) \ket{x}
      \\&\le \sum_i\sum_{x,y,z,t}\lambda_i |\phi_{x,y}|^2 |\bra{z}\cM(\dI) \ket{t}|^2+ \sum_i\sum_{x,y,z,t}\lambda_i |\phi_{z,t}|^2|\bra{y} \cM(\dI) \ket{x}|^2
       \\&\le \sum_i\sum_{x,y,z,t}\lambda_i |\phi_{x,y}|^2 |\bra{z}\cM(\dI) \ket{t}|^2+ \sum_i\sum_{x,y,z,t}\lambda_i |\phi_{z,t}|^2|\bra{y} \cM(\dI) \ket{x}|^2
       \\&=  \tr\left|(\dI-\proj{\Psi}) MM^\dagger(\dI-\proj{\Psi}) \right| \tr(\cM(\dI)^2)
\end{align*}

By Lem.~\ref{lem: MM*} we have $\bra{\Psi}MM^\dagger\ket{\Psi}=\frac{\din}{\dout}$ and $\tr(MM^\dagger)=\frac{\din}{\dout}+\eta^2$ so we have:
\begin{align*}
    \tr\left|(\dI-\proj{\Psi}) MM^\dagger(\dI-\proj{\Psi}) \right|&= \tr(\dI-\proj{\Psi}) MM^\dagger(\dI-\proj{\Psi}) 
=\eta^2.
\end{align*}
Finally
\begin{align*}
F((12)(34))\le \frac{\din^3}{\dout^2}+2\frac{\din}{\dout}m^2+\eta^2m^2.
\end{align*}
This concludes the proof for $F((12)(34))$.
\end{proof}

\begin{lemma}\label{lem: F(14)}
We can upper bound $F((14)), F((12)),  F((23)) $ and $F((34))$ as follows:
	\begin{align*}
F((14))&= F((12))= F((23))= F((34)) \le \frac{\din^2}{\dout^2}+\frac{\din}{\dout}\eta^2+ \frac{m^2}{\dout}+5m\eta^3.
	\end{align*}
\end{lemma} 
\begin{proof}
Recall the notation $N=\sum_{k,l} \tr(A_kA_l^\dagger)A_k^\dagger A_l$. We can write the spectral decomposition of the Hermitian matrix:
\begin{align*}
    N- \frac{\dI}{\dout}=\tr_2\left( \sum_{k,l} A_l^\dagger A_k\otimes A_l^{\top}\bar{A}_k -\frac{\din}{\dout} \proj{\Psi}\right) =\sum_i \lambda_i \proj{\phi_i}.
\end{align*}
Then using the triangle inequality and the fact that $\sum_k A_k^\dagger A_k=\dI$:
\begin{align*}
F(14)&= \sum_{k,l,k',l'}  \tr(A_k^\dagger  A_l )\tr(A_l^\dagger A_k A_{k'}^\dagger  A_{l' }A_{l'}^\dagger A_{k'}) 
=  \sum_{k'}  \tr(N A_{k'}^\dagger  \cN(\dI)A_{k'}) 
\\&=  \sum_{k'}  \tr(N A_{k'}^\dagger  \cM(\dI)A_{k'}) + \frac{\din}{\dout}\tr(N)
=   \tr(\cN(N)  \cM(\dI)) +  \frac{\din}{\dout}\tr(N)
\\&=   \tr(\cM(N)  \cM(\dI)) +  \frac{\din}{\dout}\tr(N)=   \tr\left(\cM\left(N-\frac{\dI}{\dout}\right)  \cM(\dI)\right) +\frac{\tr(\cM(\dI)^2)}{\dout}+  \frac{\din}{\dout}\tr(N)
\\&=  \frac{\din}{\dout}\left(  \frac{\din}{\dout}+\eta^2\right) +\frac{\tr(\cM(\dI)^2)}{\dout}+ \sum_{i} \lambda_i \tr(\cM(\proj{\phi_i}) \cM(\dI)) 
\\&\le \frac{\din}{\dout}\left(  \frac{\din}{\dout}+\eta^2\right) +\frac{\tr(\cM(\dI)^2)}{\dout}+ \sum_{i} |\lambda_i|  \|\cM(\proj{\phi_i}) \|_2\|\cM(\dI)\|_2
\\&\le  \frac{\din}{\dout}\left(  \frac{\din}{\dout}+\eta^2\right)  +\frac{m^2}{\dout}+ m\eta \sum_{i} |\lambda_i|
=  \frac{\din}{\dout}\left(  \frac{\din}{\dout}+\eta^2\right)  +\frac{m^2}{\dout}+ m\eta\cdot \tr\left|  N- \frac{\dI}{\dout}   \right|
\end{align*}

because for all unit vector $\ket{\phi}=\sum_i \phi_i \ket{i}$ we have using the Cauchy Schwarz inequality, the AM-GM inequality, and Lem.~\ref{lem: eta= }:
\begin{align*}
\|\cM(\proj{\phi})\|_2^2 &=\sum_{i,j,k,l } \phi_i \bar{\phi}_j \phi_k\bar{\phi}_l  \tr\left( \cM(\ket{i}\bra{j} )  \cM(\ket{k}\bra{l})    \right)
\le \sum_{i,j,k,l } |\phi_i \bar{\phi}_j \phi_k\bar{\phi}_l| \|\cM(\ket{i}\bra{j} ) \|_2\| \cM(\ket{k}\bra{l}) \|_2
\\ &\le \frac{1}{2}\sum_{i,j,k,l}  |\phi_i |^2 |\phi_j|^2 \|\cM(\ket{k}\bra{l} ) \|_2^2+\frac{1}{2}\sum_{i,j,k,l}  |\phi_k |^2 |\phi_l|^2 \|\cM(\ket{i}\bra{j} ) \|_2^2
= \eta^2. 
\end{align*}
Moreover, using the data processing inequality and  Lem.~\ref{lem: M-Psi} :
\begin{align*}
\tr\left|  N- \frac{\dI}{\dout}   \right|&=\tr\left| \tr_2\left(\sum_{k,l} A_l^\dagger A_k\otimes A_l^{\top}\bar{A}_k -\frac{\din}{\dout}\proj{\Psi}\right)  \right|
\\&\le \tr\left|    \sum_{k,l} A_l^\dagger A_k\otimes A_l^{\top}\bar{A}_k -\frac{\din}{\dout}\proj{\Psi}     \right|
= \tr\left|M^\dagger M-\frac{\din}{\dout}\proj{\Psi}\right|
\le 5\eta^2.
\end{align*}
This concludes the proof.
\end{proof}
For the remaining permutations, we can obtain a closed form for the function $F$. For the transposition $(13)$, the image of the  function $F$ can be expressed as follows:
\begin{align*}
F((13)) &= \sum_{i,j,k,l,k',l'}  \tr(A_{l'}^\dagger \ket{j}\bra{i}A_kA_l^\dagger \ket{i}\bra{j} A_{k'}) \tr(A_k^\dagger  A_l)\tr(A_{k'}^\dagger  A_{l' } )
\\&= \sum_{k,l,k',l'}  \tr(A_{k'}A_{l'}^\dagger ) \tr(A_kA_l^\dagger ) \tr(A_k^\dagger  A_l)\tr(A_{k'}^\dagger  A_{l' } )
\\&= \left( \sum_{k,l} |\tr(A_kA_l^\dagger )|^2 \right)^2=\left(\frac{\din}{\dout}+\eta^2\right)^2.
\end{align*}
Then, we remark that the permutations $(312)$ and $(134)$ have the same image:
\begin{align*}
    F((312))=F((134)) &=   \sum_{i,j,k,l,k',l'}  \tr(A_l^\dagger \ket{i}\bra{j} A_{k'}A_{l'}^\dagger \ket{j}\bra{i}A_kA_k^\dagger  A_l) \tr( A_{k'}^\dagger  A_{l' })
    \\&=  \sum_{k,l,k',l'}\tr(A_{k'}A_{l'}^\dagger ) \tr(A_kA_k^\dagger  A_lA_l^\dagger ) \tr( A_{k'}^\dagger  A_{l' })
    \\&= \sum_{k,l } \tr(A_kA_k^\dagger  A_lA_l^\dagger ) \sum_{k,l }|\tr(A_kA_l^\dagger )|^2
    \\&= \tr(\cN(\dI)^2) \left(\frac{\din}{\dout}+\eta^2\right) = \left(\frac{\din^2}{\dout}+m^2\right)\left(\frac{\din}{\dout}+\eta^2\right).
\end{align*}
Next, the image of the cycle $(1234)$ has also a closed expression:
\begin{align*}
    F((1234))&=  \sum_{i,j,k,l,k',l'}  \tr( A_{l'}^\dagger \ket{j}\bra{i}A_kA_k^\dagger  A_lA_l^\dagger \ket{i}\bra{j} A_{k'}  A_{k'}^\dagger  A_{l' })
    \\&=  \sum_{k,l,k',l'} \tr(A_kA_k^\dagger  A_lA_l^\dagger  ) \tr(A_{k'}  A_{k'}^\dagger  A_{l' }A_{l'}^\dagger )
    \\&= (\tr(\cN(\dI)^2))^2=\left(\frac{\din^2}{\dout}+m^2\right)^2.
\end{align*}
To sum up, we have proved so far that:
\begin{lemma}\label{lem: F(..;)} Let $m=\|\cM(\dI)\|_2= \|\cN(\dI)-\cD(\dI)\|_2$ and $\eta= \din\left\|\cJ -\frac{\dI}{\din\dout}\right\|_2$. We have:

 \begin{table}[ht]
\centering
\begin{tabular}{  |c  |  c| c|} 
\hline
  \textbf{Permutation } $\alpha$  & \textbf{Upper bound on } $F(\alpha)$ &\textbf{Reference}\\
  \hline\hline
   $(13)$ & $ \left(\frac{\din}{\dout}+\eta^2\right)^2 $&
   \\\hline
   $\id, (132), (314), (24)(13)$ & $ \frac{\din/\dout+\eta^2}{\dout} + \frac{\eta^2}{\dout}+5\eta^4$ &(Lem.~\ref{lem: F(id)})
	\\\hline 
$(312), (134) $& $ \left(\frac{\din^2}{\dout}+m^2\right)\left(\frac{\din}{\dout}+\eta^2\right)$&
	\\\hline
$ (1234) $& $ \left(\frac{\din^2}{\dout}+m^2\right)^2$&
	\\ \hline
$ (24), (1432)$ &$ \frac{\din^2}{\dout^2}+\frac{2m^2}{\dout}+25\eta^4$ &(Lem.~\ref{lem: F(24)})
	\\\hline
 $(142), (243)$& $\frac{\din/\dout+\eta^2}{\dout} + \frac{\eta^2}{\dout}+5\eta^4$ &(Lem.~\ref{lem: F(142)})
	\\\hline
  $(14), (12), (23), (34), (1324), $&\multirow{2}{*}{$  \frac{\din^2}{\dout^2}+\frac{\din}{\dout}\eta^2+ \frac{m^2}{\dout}+5m\eta^3$}& \multirow{2}{*}{(Lem.~\ref{lem: F(14)})}\\
$(1423), (1243), (1342)$& &
		\\\hline
  $(12)(34), (14)(23), (234), (124)$& $\frac{\din^3}{\dout^2}+2\frac{\din}{\dout}m^2+m^2\eta^2 $ &(Lem.~\ref{lem: F(12)(34))})
  \\\hline
\end{tabular}
\end{table}
\end{lemma}
Therefore we have  the following upper bound on the second moment using the inequality $m^2\le d\eta^2$: 
\begin{align*}
&\ex{\left( \tr( \cN(\proj{\phi  }    ) ^2)  \right)^2  } = \frac{1}{\din(\din+1)(\din+2)(\din+3)} \sum_{\alpha \in \fS_4} F(\alpha)
\\&\le \frac{\dout^4+ 6\din^3+2\din^2\dout m^2+2\din^2\dout \eta^2+ 11\din^2+ 10\din \dout m^2+10\din\dout \eta^2 }{\dout^2\din(\din+1)(\din+2)(\din+3)}
\\&+ \frac{6\din + \dout^2m^4+2\dout^2\eta^2+ 12\dout m^2+ 40\dout^2m\eta^3 +81\dout^2\eta^4+12 \dout \eta^2   }{\dout^2\din(\din+1)(\din+2)(\din+3)}
\end{align*} 
Recall that for the random variable $X= \tr \left(\cN(\proj{\phi}) -\frac{\dI}{\dout}\right)^2= \tr \left(\cN(\proj{\phi}) \right)^2-\frac{1}{\dout}$ we have:
\begin{align*}
\ex{X}= \frac{\left( \tr\left(\cN(\dI) -\frac{\din}{\dout}\dI\right)^2  +\din^2\left\|\cJ -\frac{\dI}{\din\dout}\right\|_2^2\right)}{\din(\din+1)} =\frac{1}{\din(\din+1)}\left( \tr(\cM(\dI)^2)+\eta^2\right)=\frac{m^2+\eta^2}{\din(\din+1)}.
\end{align*}
Since $\var(X)=\var\left(X+\frac{1}{\dout}\right)$, it can be upper bounded as follows:
\begin{align*}
&\var(X) =\ex{\left( \tr( \cN(\proj{\phi  }    ) ^2)  \right)^2  } -\left(\ex{ \tr( \cN(\proj{\phi  }    ) ^2)   }\right)^2 
\\&\le \frac{\dout^4+ 6\din^3+2\din^2\dout m^2+2\din^2\dout \eta^2+ 11\din^2+ 10\din \dout m^2+10\din\dout \eta^2 }{\dout^2\din(\din+1)(\din+2)(\din+3)}
\\&+ \frac{6\din + \dout^2m^4+2\dout^2\eta^2+ 12\dout m^2+ 40\dout^2m\eta^3 +81\dout^2\eta^4+12 \dout \eta^2   }{\dout^2\din(\din+1)(\din+2)(\din+3)}- \left(\frac{m^2+\eta^2}{\din(\din+1)}+\frac{1}{\dout}\right)^2
\\&\le \left(\frac{80\din^2\eta^4+10\din^2m^2\eta^2+40\din^2\eta^3m}{\din^2(\din+1)^2(\din+2)(\din+3)}\right).
\end{align*}
Therefore the upper bound on the variance becomes using the inequalities $(m^2+\eta^2)^2\ge 4m^2\eta^2$ and  $(m^2+\eta^2)^2\ge 2m\eta^{3}$ 
 (successive AM-GM):
	\begin{align*}
\frac{\var(X)}{\ex{X}^2}&\le  \left( \frac{80\din^2\eta^4+10\din^2m^2\eta^2+40\din^2\eta^3m}{(\din+2)(\din+3)(\eta^2+m^2  )^2} \right)
\\&\le 80\frac{\din^2\eta^4}{\din^2\eta^4}+ 10\frac{\din^2m^2\eta^2}{4\din^2m^2\eta^2}+40\frac{\din^2\eta^3m}{2\din^2m\eta^{3}}
\le 105.
\end{align*} 
\end{proof}

\begin{theorem}\label{app:thm:Test-D}
 There is an incoherent ancilla-free algorithm using a number of measurements satisfying
 \begin{align*}
     N=\cO\left(\frac{\din^2\dout^{1.5}}{\eps^2}\right)
 \end{align*} 
to distinguish between $\cN=\cD$ or $\dd(\cN,\cD)>\eps$ with a probability of success at least $2/3$.
\end{theorem}

\begin{proof}
    \textbf{Complexity.} We start by showing that  Alg.~\ref{alg-dep} uses $\cO(\din^2\dout^{1.5}/\eps^2)$  measurements. Note that Alg.~\ref{alg-states} requires $\cO(\sqrt{\dout}\log(3M)/\eta^2)$ to test whether a $\dout$-dimensional quantum state $\rho$ satisfies $\rho= \frac{\dI}{\dout}$ or $\left\|\rho-\frac{\dI}{\dout}\right\|_2\ge \eta$ with an error probability $1/{(3M)}$. We apply this test for the state $\cN(\proj{\phi})$ , which is a $\dout$-dimensional quantum state, and $\eta= \frac{\eps}{2\sqrt{\dout}\din}$. So the testing identity of states algorithm will use a number of measurements $\cO(\din^2\dout^{1.5}\log(3M)/\eps^2)$. Since we repeat this test a number of $M=\cO(1)$ times, the total number of measurements is $\cO(\din^2\dout^{1.5}/\eps^2)$.

\textbf{Correctness.} It remains to show that Alg.~\ref{alg-dep} is $1/3$-correct. For this, we need to control two error probabilities. The first one, under the null hypothesis $H_0$, the channel $\cN=\cD$ so for all $k\in [M]$, $\cN(\proj{\phi_k}))=\frac{\dI}{\dout} $. Hence, for all the quantum states test we are under the null hypothesis $h_0$. Therefore 
\begin{align*}
    \mathds{P}_{H_0}(\text{error}) &= \mathds{P}_{H_0}(\exists k\in [M]: i_k=1)
    \\&\le\sum_{k=1}^M\mathds{P}_{h_0}(i_k=1)\le\sum_{k=1}^M\frac{1}{3M}=\frac{1}{3}.
\end{align*}
The second error probability concerns the alternate hypothesis $H_1$. In this case, the channel $\cN$ is $\eps$-far from the depolarizing channel $\cD$. For a given $k\in [M]$, let us lower bound the probability of being under $h_1$ at step $k$ of the algorithm. Let   $X_k=\left\|\cN(\proj{\phi_k}) -\frac{\dI}{\dout}\right\|_2^2$. Recall that Thm.~\ref{thm:PZ-D} says:
    \begin{align*}
    \var(X_k)\le 105\ex{X_k}^2  
\end{align*}
Therefore, the  Paley-Zygmund inequality implies
\begin{align*}
    \pr{X_k\ge \frac{1}{2}\ex{X_k}}&\ge 1-\frac{1}{\frac{1}{4}\frac{\ex{X_k}^2}{\var(X_k)}+1}\ge 1-\frac{1}{\frac{1}{8\times 105}+1}
    \ge 10^{-3}.
\end{align*}
Recall also that $\ex{X_k}= \ex{ \left\|\cN(\proj{\phi_k}) -\frac{\dI}{d}\right\|_2^2}\ge \frac{\eps^2}{2\din^2\dout}$ (Ineq.~\ref{ineq-expectation}). Therefore, with a probability at least $10^{-3}$, we have $\left\|\cN(\proj{\phi_k}) -\frac{\dI}{\dout}\right\|_2\ge \frac{\eps}{2\din\sqrt{\dout}}$ which means we are under the alternate hypothesis $h_1$. Denote the good event of $E_k=\left\{\left\|\cN(\proj{\phi_k}) -\frac{\dI}{\dout}\right\|_2
    \ge \frac{\eps}{2\din\sqrt{\dout} } \right\}$, the probability of error under $H_1$ can be controlled as follows:
    \begin{align*}
      \mathds{P}_{H_1}(\text{error})&= \mathds{P}_{H_1}(\forall k \in [M]: i_k=0) 
     = \prod_{k=1}^{M} \mathds{P}_{H_1}(i_k=0) 
      \\&\le \prod_{k=1}^{M} \left( \mathds{P}_{H_1}(i_k=0 | E_k)\pr{E_k}+\mathds{P}_{H_1}(i_k=0| \bar{E}_k)\pr{\bar{E}_k}\right)
       \\&\le \prod_{k=1}^{M} \left( \mathds{P}_{h_1}(i_k=0 )\pr{E_k}+\pr{\bar{E}_k}\right)
       \\&\le \prod_{k=1}^{M} \left( \mathds{P}_{H_1}(i_k=0 | E_k)\pr{E_k}+1-\pr{E_k}\right)
       \\&=\prod_{k=1}^{M} \left( 1-\pr{E_k}\mathds{P}_{h_1}(i_k=1 | E_k)\right)
      \\&\le \left( 1-10^{-3}\cdot \left(1-\frac{1}{3M}\right)  \right)^{M}
      \le \left(   1-\frac{10^{-3}}{2} \right)^{M}
     \le \frac{1}{3}
    \end{align*}
    for $M=2200=\cO(1)$. This concludes the correctness of  Alg.~\ref{alg-dep}.
\end{proof}

\section{Lower bounds for testing identity to the depolarizing channel}
\subsection{Proof of Thm.~\ref{thm:LB_DEP_non-adap}}\label{thm-app:LB_DEP_non-adap}
\begin{theorem}
  Let $\eps\le 1/32$, $\din \ge 80$ and $\dout\ge 10$. \footnote{We didn't try to optimize these conditions.}  Any incoherent ancilla-assisted non-adaptive algorithm for testing identity to the depolarizing channel (in both the trace and diamond distances) requires, in the worst case, a number of measurements
    \begin{align*}
        N=\Omega\left(\frac{\din^2\dout^{1.5}}{\log(\din\dout/\eps)^2\eps^2}\right).
    \end{align*}
 
\end{theorem}
 \begin{proof}\textbf{Construction.}
    Under the null hypothesis $H_0$ the quantum channel is $\cN(\rho)=\cD(\rho)= \tr(\rho)\frac{\dI}{\dout}$. Under the alternate hypothesis $H_1$, we choose the quantum channel $\cN\sim \cP$ of the form:
	 \begin{align*}
	 	\cN(\rho)= \tr(\rho) \frac{\dI}{\dout} +\frac{\eps}{\dout} \bra{\bar{w}}\rho \ket{\bar{w}} U
	 \end{align*}
	 where $\ket{w}= W\ket{0}$, $W \sim \Haar(\din)$, and $U$ satisfies 
\begin{align*}
   \begin{cases*}
     U_{x,x}=0  & for all   $x\in [\dout]$, \\
       U_{y,x}=\bar{U}_{x,y} \sim  \cN_c(0, \sigma^2) =  \cN(0, \sigma^2/2)+i\, \cN(0, \sigma^2/2) &  for all   $x<y\in [\dout]$, 
\end{cases*} 
\end{align*}
 where the parameter $\sigma$ would be chosen later and we condition on the event \[\cG=\{\|U\|_1\ge \dout, \|U\|_\infty \le 32\}.\] 
  We call this distribution $\cP$ and use the notation $(w,U)\sim \cP$. If we don't condition on the event $\cG$, the distribution of $U$ is denoted $\cP_0$ and we write $U\sim \cP_0$. Random constructions with Gaussian random variables were used for proving lower bounds by \citet{chen2022tight-mixedness,chen2022tight}.
	 Note that $\cN$ is trace preserving since $\tr(U)=0$. In order to prove that $\cN$ is a quantum channel, it remains to show that $\cN$ is completely positive which is equivalent to show that  the corresponding Choi matrix is positive semi-definite. For this we can express the Choi state of the channel $\cN$:
	 \begin{align*}
	 	\cJ_\cN&= \frac{\dI}{\din \dout}+ \frac{\eps}{\din \dout} \sum_{i,j=1}^{\din} \ket{i}\bra{j} \otimes\bra{0}W^\top \ket{i}\bra{j} \bar{W}\ket{0} U
	 	\\&=\frac{\dI}{\din \dout}+ \frac{\eps}{\din \dout} \sum_{i,j} \ket{i}\bra{j} \otimes\bra{i}W \ket{0}\bra{0} W^{\dagger}\ket{j} U
	 		\\&=\frac{\dI}{\din \dout}+ \frac{\eps}{\din \dout} W \ket{0}\bra{0} W^{\dagger}\otimes U
    \\&=\frac{\dI}{D}+ \frac{\eps}{D}  \ket{w}\bra{w} \otimes U
	 \end{align*} 
  where $D=\din \dout$.
\\Observe that $ \|\proj{w}\otimes U  \|_\infty  = \|U \|_\infty \le 32 $  thus $\cJ_\cN\ge 0$ if $\eps \le 1/32$. So under the event $\cG$, the map $\cN$ is a quantum channel. The parameter $\sigma$ should be chosen so that $\dd(\cN,\cD)\ge \eps$. Recall that $\ket{w}= W\ket{0}$, the definition of the diamond distance implies
\begin{align*}
    \dd(\cN,\cD)&\ge \td(\cN, \cD)= \max_{\rho }\|(\cN-\cD)(\rho) \|_1
    \\&\ge \|(\cN-\cD)(\proj{\bar{w}})\|_1
    =\frac{\eps}{\dout}\|U\|_1\ge \eps
\end{align*}
where we use the fact that under the event $\cG$ we have $\|U\|_1\ge \dout$.
Note that the lower bound we prove in this theorem holds for the stronger condition given by the  trace distance. 
\\Now we move to show that the event $\cG$ occurs with high probability.
\begin{lemma}\label{lem:concentration ||U||_1}
There is a constant $c>0$ such that we have:
\begin{align*}
        \pr{|\|U\|_1-\ex{\|U\|_1}|>s }\le \exp\left(-\frac{cs^2}{\dout\sigma^2}\right) 
\end{align*}
\end{lemma}
\begin{proof}
    The function $U\mapsto \|U\|_1$ is $\sqrt{\dout}$-Lipschitz w.r.t. the Hilbert-Schmidt norm. Indeed, by the triangle inequality and the Cauchy Schwarz inequality $|\|U\|_1-\|V\|_1|\le \|U-V\|_1\le \sqrt{\dout}\|U-V\|_2$. The concentration of Lipschitz functions of Gaussian random variables \citep[Theorem 2.26]{wainwright2019high}  yields exactly the desired statement.
\end{proof}
Next, we need a lower bound on the expectation of $\|U\|_1$. By the Hölder's inequality:
\begin{align*}
    \ex{\|U\|_1}\ge \sqrt{\frac{\ex{\|U\|_2^2}^3}{\ex{\|U\|_4^4}}}\ge \sqrt{\frac{(\dout^2\sigma^2)^3}{4\dout^3\sigma^4}}= \frac{\dout\sqrt{\dout}\sigma}{2}. 
\end{align*}
It is sufficient to choose $\sigma=\frac{4}{\sqrt{\dout}}$ so that $\ex{\|U\|_1}\ge 2\dout$
and by Lem.~\ref{lem:concentration ||U||_1}, we have $\|U\|_1\ge \dout$ with a probability  $1-\exp(-\Omega(\dout^2))$. 
It remains to see that, for this choice of $\sigma= \frac{4}{\sqrt{\dout}}$, the event $\{\|U\|_\infty \le 32\}$ 
also occurs with high probability.  Indeed, 
let $\cS$ be a $1/4$-net of $\bS^{\dout}$ of size at most $8^{\dout}$. By the union bound: 
\begin{align*}
    \pr{\|U\|_\infty >32}&=\pr{\exists \phi\in \bS^{\dout}: \bra{\phi} U\ket{\phi}=\|U\|_\infty, \|U\|_\infty >32 }
    \\&\le \pr{\exists \phi\in \cS: \bra{\phi} U\ket{\phi}>\frac{1}{2}\|U\|_\infty, \|U\|_\infty >32 }
    \\&\le |\cS|\pr{ \bra{\phi} U\ket{\phi}>16}
    \le 8^{\dout} e^{-8\dout}\le e^{-4\dout}.
\end{align*}
Finally, with a probability at least $1-\exp(-\Omega(\dout^2))-\exp(-\Omega(\dout))$  the event $\cG$ is satisfied and we have a quantum channel $\cN$ that is $\eps$-far in the diamond distance  from the depolarizing channel $\cD$. A $1/3$-correct algorithm $\cA$ should distinguish between the channels $\cN$ and $\cD$ with a probability of error at most $1/3$. Let $N$ be a sufficient number of measurements for this task and $I_1,\dots, I_N$ be  the observations of the algorithm $\cA$. The Data-Processing inequality applied on the $\TV$-distance gives LeCam's method \citep{lecam1973convergence}:
	\begin{align*}
\TV\left(\mathds{P}_{H_0}^{I_1,\dots, I_N}\Big\| \mathds{P}_{H_1}^{I_1,\dots, I_N}\right)&\ge  \TV\left(\Ber(\mathds{P}_{H_0}(\cA=1))\| \Ber(\mathds{P}_{H_1}(\cA=1))\right) 
\\&\ge \TV(\Ber(1/3)\| \Ber(2/3))= \frac{1}{3}.
\end{align*}
Now, we need to upper bound this $\TV$ distance with an expression involving $N, \din, \dout$ and $\eps$.

\textbf{Upper bound on the $\TV$ distance.} 
The non-adaptive algorithm $\cA$ would choose at step $t$ the input $\rho_t$ and the measurement device $\cM_t=\{\lambda_i^t \proj{\phi_i^t}\}_{i\in \cI_t}$. Observe that we can always reduce w.l.o.g. to such a POVM. 
Under the null hypothesis $H_0$, the quantum channel $\cN=\cD$ so the probability distribution of the outcomes is exactly:
\begin{align*}
   \mathds{P}_{H_0}^{I_1,\dots, I_N}= \left\{\prod_{t=1}^N\lambda_{i_t}^t \bra{\phi_{i_t}^t}\id\otimes \cD (\rho_t) \ket{\phi_{i_t}^t}\right\}_{i_1,\dots, i_N} .
\end{align*}
On the other hand, under the alternate hypothesis $H_1$, the probability of the outcomes is exactly:
\begin{align*}
   \mathds{P}_{H_1}^{I_1,\dots, I_N}= \left\{\prod_{t=1}^N\lambda_{i_t}^t \bra{\phi_{i_t}^t}\id\otimes \cN (\rho_t) \ket{\phi_{i_t}^t}\right\}_{i_1,\dots, i_N}
\end{align*}
Recall that $D=\din \dout$ and
\begin{align*}
   \cJ_\cD= \frac{\dI}{D} ~~\text{ and }~~\cJ_\cN= \frac{\dI}{D}+ \frac{\eps}{D}  \proj{w} \otimes U.
\end{align*}
We can suppose w.l.o.g. that each input state is pure because of the convexity of the $\TV$ distance. That is, for $t\in [N]$, we write $\rho_t= \proj{\psi_t}$. Moreover, we can write each rank one input state and measurement vector as follows:
\begin{align*}
    \ket{\psi_t}= A_t\otimes \dI \ket{\Psi_{\din}} ~~\text{ and }~~\ket{\phi_{i_t}^t}=  B_{i_t}^t\otimes \dI \ket{\Psi_{\dout}}
\end{align*}
where $\ket{\Psi_d}= \frac{1}{\sqrt{d} } \sum_{i=1}^{d} \ket{ii}$ and the matrices $A_t \in \dC^{\danc \times \din }$ and $B_{i_t}^t \in \dC^{\danc \times \dout }$ verify: 
\begin{align*}
   \tr(A_tA_t^\dagger )= \din  ~~\text{ and }~~  \tr(B_{i_t}^tB_{i_t}^{t,\dagger} )= \dout.
\end{align*}
Note that we have for all $t$, for all $X\in \dC^{\danc\otimes \danc}$ wa have $\sum_{i_t} \lambda_{i_t}^t B_{i_t}^{t,\dagger}XB_{i_t}^{t} = \dout\tr(X)\dI_{\dout}$. Indeed, the condition of the POVM $\cM_t$ implies: 
\begin{align*}
 X\otimes \dI=    X\otimes \dI\sum_{i_t}\lambda_{i_t}^t \proj{\phi_{i_t}^t} =  \sum_{i_t}\lambda_{i_t}^t XB_{i_t}^t\otimes \dI \proj{\Psi_{\dout}} B_{i_t}^{t, \dagger}\otimes\dI 
\end{align*}
hence by taking the partial trace
\begin{align*}
    \tr(X) \dI = \frac{1}{\dout}\sum_{i_t, i, j}\lambda_{i_t}^t \bra{j}  B_{i_t}^{t, \dagger}X B_{i_t}^{t} \ket{i}  \ket{i}\bra{j}.
\end{align*}
Finally
\begin{align*}
    \sum_{i_t} \lambda_{i_t}^t B_{i_t}^{t,\dagger}XB_{i_t}^{t} = \dout\tr(X)\dI_{\dout}.
\end{align*}
By taking $X=\dI$  and the partial trace on the second system we have 
\begin{align*}
    \sum_{i_t} \lambda_{i_t}^t B_{i_t}^{t} B_{i_t}^{t, \dagger} =\dout^2 \dI_{\danc}.
\end{align*}
Now we can re-write the  distribution  of the observations under the null hypothesis as follows:
\begin{align*}
   \mathds{P}_{H_0}^{I_1,\dots, I_N}&= \left\{\prod_{t=1}^N\lambda_{i_t}^t \bra{\phi_{i_t}^t}\id\otimes \cD (\rho_t) \ket{\phi_{i_t}^t}\right\}_{i_1,\dots, i_N} 
   \\&=\left\{\prod_{t=1}^N\lambda_{i_t}^t \bra{\Psi_{\dout}}(B_{i_t}^{t,\dagger}A_t \otimes \dI)\cJ_\cD (A_t^\dagger B_{i_t}^{t}\otimes \dI) \ket{\Psi_{\dout}}
   \right\}_{i_1,\dots, i_N} 
   \\&= \left\{\prod_{t=1}^N\frac{\lambda_{i_t}^t}{D} \bra{\Psi_{\dout}}(B_{i_t}^{t,\dagger}A_tA_t^\dagger B_{i_t}^{t} \otimes \dI) \ket{\Psi_{\dout}}
   \right\}_{i_1,\dots, i_N} 
   \\&= \left\{\prod_{t=1}^N\frac{\lambda_{i_t}^t}{D\dout} \tr(B_{i_t}^{t,\dagger}A_tA_t^\dagger B_{i_t}^{t})
   \right\}_{i_1,\dots, i_N} ,
\end{align*}
and under the alternate hypothesis as follows:
\begin{align*}
   \mathds{P}_{H_1}^{I_1,\dots, I_N}&= \left\{\prod_{t=1}^N\lambda_{i_t}^t \bra{\phi_{i_t}^t}\id\otimes \cN (\rho_t) \ket{\phi_{i_t}^t}\right\}_{i_1,\dots, i_N}
  \\&=\left\{\prod_{t=1}^N\lambda_{i_t}^t \bra{\Psi_{\dout}}(B_{i_t}^{t,\dagger}A_t \otimes \dI)\cJ_\cN (A_t^\dagger B_{i_t}^{t}\otimes \dI) \ket{\Psi_{\dout}}
   \right\}_{i_1,\dots, i_N} 
    \\&=\left\{\prod_{t=1}^N\lambda_{i_t}^t \bra{\Psi_{\dout}}(B_{i_t}^{t,\dagger}A_t \otimes \dI)\left(\frac{\dI}{D}+ \frac{\eps}{D}  \proj{w} \otimes U \right)(A_t^\dagger B_{i_t}^{t}\otimes \dI) \ket{\Psi_{\dout}}
   \right\}_{i_1,\dots, i_N}
    \\&=\left\{\prod_{t=1}^N \left( \frac{\lambda_{i_t}^t}{D\dout} \tr(B_{i_t}^{t,\dagger}A_tA_t^\dagger B_{i_t}^{t})+ \frac{\eps\lambda_{i_t}^t}{D} \bra{\Psi_{\dout}}(B_{i_t}^{t,\dagger}A_t  \proj{w} A_t^\dagger B_{i_t}^{t}  \otimes U) \ket{\Psi_{\dout}}\right)
   \right\}_{i_1,\dots, i_N}
\end{align*}
So, we can express the $\TV$ distance as follows:
\begin{align*}
   & 2\TV\left(\mathds{P}_{H_0}^{I_1,\dots, I_N}\Big\| \mathds{P}_{H_1}^{I_1,\dots, I_N}\right)
   \\&=\sum_{i_1,\dots, i_N}\Bigg| \mathds{E}_{(w,U)\sim\cP}\prod_{t=1}^N \left( \frac{\lambda_{i_t}^t}{D\dout} \tr(B_{i_t}^{t,\dagger}A_tA_t^\dagger B_{i_t}^{t})+ \frac{\eps\lambda_{i_t}^t}{D} \bra{\Psi_{\dout}}(B_{i_t}^{t,\dagger}A_t  \proj{w} A_t^\dagger B_{i_t}^{t}  \otimes U) \ket{\Psi_{\dout}}\right)
   \\& \qquad\qquad \qquad\quad\ \ -\prod_{t=1}^N\frac{\lambda_{i_t}^t}{D\dout} \tr(B_{i_t}^{t,\dagger}A_tA_t^\dagger B_{i_t}^{t})\Bigg|
   \\&= \mathds{E}_{\le N}\left| \mathds{E}_{(w,U)\sim\cP}\prod_{t=1}^N \left(1+\eps\dout \frac{\bra{\Psi_{\dout}}(B_{i_t}^{t,\dagger}A_t  \proj{w} A_t^\dagger B_{i_t}^{t}  \otimes U) \ket{\Psi_{\dout}}}{\tr(B_{i_t}^{t,\dagger}A_tA_t^\dagger B_{i_t}^{t})} \right) -1\right|
   \\&= \mathds{E}_{\le N}\left| \mathds{E}_{(w,U)\sim\cP}\prod_{t=1}^N \left(1+\eps\cdot  \frac{\tr(B_{i_t}^{t,\dagger}A_t  \proj{w} A_t^\dagger B_{i_t}^{t}  U)}{\tr(B_{i_t}^{t,\dagger}A_tA_t^\dagger B_{i_t}^{t})} \right) -1\right|
\end{align*}
where we use the notation $\mathds{E}_{\le N}(X(i_1,\dots, i_N))= \sum_{i_1,\dots, i_N}\left(\prod_{t=1}^N \frac{\lambda_{i_t}^t}{D\dout} \tr(B_{i_t}^{t,\dagger}A_tA_t^\dagger B_{i_t}^{t})\right)X(i_1,\dots, i_N)$. 
Let $u$ be a standard Gaussian vector  such that $\ket{w}=\frac{\ket{u}}{\|u\|_2}$. We condition on the event $\cE$ that $u$ satisfies: 
\begin{enumerate}
    \item $\|u\|_2 \;\ge \sqrt{\frac{1}{6}\din}$ ~~ and 
    \item  $\forall t\in [N]:P_u^t\le \frac{(7\log(N))^4}{\din^4 \dout}$
\end{enumerate}
where $P_u^t$ is defined as 
\begin{align*}
P^t_w&= \sum_{i_t,j_t}\frac{\lambda^t_{i_t}\lambda^t_{j_t}}{D^2\dout^2}\left(\frac{|\bra{w}A_t^\dagger B_{i_t}^t  B_{j_t}^{t,\dagger}A_t  \ket{w}|^4}{\tr(B_{i_t}^{t,\dagger}A_tA_t^\dagger B_{i_t}^{t})\tr(B_{j_t}^{t,\dagger}A_tA_t^\dagger B_{j_t}^{t})} \right)
   \\ P_u^t&= \sum_{i_t,j_t}\frac{\lambda^t_{i_t}\lambda^t_{j_t}}{D^2\dout^2\din^4}\left(\frac{|\bra{u}A_t^\dagger B_{i_t}^t  B_{j_t}^{t,\dagger}A_t  \ket{u}|^4}{\tr(B_{i_t}^{t,\dagger}A_tA_t^\dagger B_{i_t}^{t})\tr(B_{j_t}^{t,\dagger}A_tA_t^\dagger B_{j_t}^{t})} \right).
\end{align*}
Note that under the event $\cE$, we have $P_w^t\le 6^4 P_u^t $. We claim that for  $\din \ge 80$, we have with probability at least $9/10$, the event $\cE$ is satisfied. The proof of this claim is deferred to Lemma~\ref{lem: event E sat}.
\\Now, we can distinguish whether the event $\cE$ is verified or not in the $\TV$ distance. Let $$\Psi_{i,w, U}=\prod_{t=1}^N \left(1+\eps\cdot \Phi^{t,i_t}_{w, U} \right)\quad \text{ where }\quad  \Phi^{t,i_t}_{w, U}=\frac{\tr(B_{i_t}^{t,\dagger}A_t  \proj{w} A_t^\dagger B_{i_t}^{t}  U)}{\tr(B_{i_t}^{t,\dagger}A_tA_t^\dagger B_{i_t}^{t})}.  $$
By the triangle inequality:
\begin{align*}
     &\mathds{E}_{\le N}\left| \mathds{E}_{(w,U)\sim\cP}\prod_{t=1}^N \left(1+\eps\cdot \frac{\tr(B_{i_t}^{t,\dagger}A_t  \proj{w} A_t^\dagger B_{i_t}^{t}  U)}{\tr(B_{i_t}^{t,\dagger}A_tA_t^\dagger B_{i_t}^{t})}  \right) -1\right|
     \\ &\le    \mathds{E}_{\le N}\left| \mathds{E}_{\substack{(w,U)\sim\cP}}\left[ \bm{1}\{\cE\}(\Psi_{i,w, U} -1)\right]\right|+  \mathds{E}_{\le N}
     \left| \mathds{E}_{\substack{(w,U)\sim\cP}}\left[ \bm{1}\{\bar{\cE}\}(\Psi_{i,w, U} -1)\right]\right|
      \\ &\le    \mathds{E}_{\le N}\left| \mathds{E}_{\substack{(w,U)\sim\cP}}\left[ \bm{1}\{\cE\}(\Psi_{i,w, U} -1)\right]\right|+  \mathds{E}_{\le N}
     \mathds{E}_{\substack{(w,U)\sim\cP}}\left[ \bm{1}\{\bar{\cE}\}(\Psi_{i,w, U} +1)\right]
     \\&= \mathds{E}_{\le N}\left| \mathds{E}_{\substack{(w,U)\sim\cP}}\left[ \bm{1}\{\cE\}(\Psi_{i,w, U} -1)\right]\right|+ 2\pr{\Bar{\cE}}
\end{align*}
where we use the fact that $\sum_{i_t} \lambda_{i_t}^t B_{i_t}^{t,\dagger}XB_{i_t}^{t} = \dout\tr(X)\dI$ and
\begin{align*}
    \mathds{E}_{\le N}\Psi_{i,w, U} &= \sum_{i_1,\dots, i_N}\left( \prod_{t=1}^N\frac{\lambda_{i_t}^t}{D\dout} \cdot \tr(B_{i_t}^{t,\dagger}A_tA_t^\dagger B_{i_t}^{t})\right)\prod_{t=1}^N \left(1+\eps\cdot \frac{\tr(B_{i_t}^{t,\dagger}A_t  \proj{w} A_t^\dagger B_{i_t}^{t}  U)}{\tr(B_{i_t}^{t,\dagger}A_tA_t^\dagger B_{i_t}^{t})} \right) 
    \\&= \prod_{t=1}^N \sum_{i_t}\left(  \frac{\lambda_{i_t}^t}{D\dout} \cdot \tr(B_{i_t}^{t,\dagger}A_tA_t^\dagger B_{i_t}^{t})+ \frac{\eps\lambda_{i_t}^t}{D\dout}\cdot \tr(B_{i_t}^{t,\dagger}A_t  \proj{w} A_t^\dagger B_{i_t}^{t}  U) \right) 
    \\&= \prod_{t=1}^N \left(1+\frac{\eps}{D} \cdot\tr(A_t  \proj{w} A_t^\dagger )\tr(  U)\right)=1. 
\end{align*}
It remains to upper bound the expectation $\mathds{E}_{\le N}\left| \mathds{E}_{\substack{(w,U)\sim\cP}}\left[ \bm{1}\{\cE\}(\Psi_{i,w, U} -1)\right]\right|$. For this, we follow \citep{bubeck2020entanglement} and apply the Cauchy Schwarz inequality and Hölder's inequality: 
\begin{align*}
  & \left( \mathds{E}_{\le N}\left| \mathds{E}_{\substack{(w,U)\sim\cP}}\left[\bm{1}\{\cE\}(\Psi_{i,w, U}-1)\right]\right|\right)^2+\pr{\cE}^2
   \\&\le \mathds{E}_{\le N}\left( \mathds{E}_{\substack{(w,U)\sim\cP}}\left[\bm{1}\{\cE(w)\}(\Psi_{i,w, U}-1)\right]\right)^2+\pr{\cE}^2
   \\&= \mathds{E}_{\le N} \mathds{E}_{\substack{(w,U)\sim\cP}}\mathds{E}_{\substack{(z,V)\sim\cP}}\bm{1}\{\cE(w)\}\Psi_{i,w, U}\bm{1}\{\cE(z)\}\Psi_{i,z, V} 
    \\&\le  \mathds{E}_{\substack{(w,U)\sim\cP}}\bm{1}\{\cE(w)\}\mathds{E}_{\substack{(z,V)\sim\cP}}\bm{1}\{\cE(z)\} \prod_{t=1}^N \mathds{E}_{i_t}\left(1+\eps \cdot\Phi^{t,i_t}_{w, U}\right)\left(1+\eps \cdot\Phi^{t,i_t}_{z, V} \right) 
    \\&=  \mathds{E}_{\substack{(w,U)\sim\cP}}\bm{1}\{\cE(w)\}\mathds{E}_{\substack{(z,V)\sim\cP}}\bm{1}\{\cE(z)\}\prod_{t=1}^N \left(1+\eps^2\cdot \mathds{E}_{i_t} \Phi^{t,i_t}_{w, U}\Phi^{t,i_t}_{z, V} \right) ~~~~  (\text{ because } ~~\mathds{E}_{i_t} \Phi^{t,i_t}_{w, U}=0)
    \\&\le \max_{1\le t \le N} \mathds{E}_{\substack{(w,U)\sim\cP}}\bm{1}\{\cE(w)\}\mathds{E}_{\substack{(z,V)\sim\cP}}\bm{1}\{\cE(z)\}\left(1+\eps^2 \cdot \mathds{E}_{i_t} \Phi^{t,i_t}_{w, U}\Phi^{t,i_t}_{z, V} \right)^N
     \\&\le  \frac{1}{\pr{\cG}^2}\max_{1\le t\le N} \mathds{E}_{U,V\sim \cP_0} \mathds{E}_{\substack{w,z}}\bm{1}\{\cE(w)\}\bm{1}\{\cE(z)\}\left(1+\eps^2 \cdot \left|\mathds{E}_{i_t} \Phi^{t,i_t}_{w, U}\Phi^{t,i_t}_{z, V}\right| \right)^N
      \\&\le  \frac{1}{(1-e^{-\Omega(\dout)})^2}\max_{1\le t\le N} \mathds{E}_{U,V\sim \cP_0} \mathds{E}_{\substack{w,z}}\bm{1}\{\cE(w)\}\bm{1}\{\cE(z)\}\left(1+\eps^2\cdot \left|\mathds{E}_{i_t} \Phi^{t,i_t}_{w, U}\Phi^{t,i_t}_{z, V}\right| \right)^N
\end{align*}
Note that at the last two inequalities, we don't require anymore that $U$ satisfies $\|U\|_\infty\le 32 $ and $\|U\|_1\ge \dout$.  This is possible because the integrand is positive and $\pr{\cG}\ge 1-e^{-\Omega(\dout)}$. 

For $t\in [N]$, let $Z_t$ be the polynomial in $\{U_{i,j},V_{i,j}\}_{i,j=1}^{\dout}$ defined as follows:
\begin{align*}
    Z_t&=\eps^2 \mathds{E}_{i_t} \Phi^{t,i_t}_{w, U}\Phi^{t,i_t}_{z, V}
    \\&=\eps^2 \sum_{i_t }\left( \frac{\lambda_{i_t}^t}{D\dout} \tr(B_{i_t}^{t,\dagger}A_tA_t^\dagger B_{i_t}^{t})\right) \left(\frac{\tr(B_{i_t}^{t,\dagger}A_t  \proj{w} A_t^\dagger B_{i_t}^{t}  U)}{\tr(B_{i_t}^{t,\dagger}A_tA_t^\dagger B_{i_t}^{t})}\right)\left(\frac{\tr(B_{i_t}^{t,\dagger}A_t  \proj{z} A_t^\dagger B_{i_t}^{t}  V)}{\tr(B_{i_t}^{t,\dagger}A_tA_t^\dagger B_{i_t}^{t})}\right)
     \\&=\eps^2 \sum_{i_t }\left( \frac{\lambda_{i_t}^t}{D\dout} \right) \left(\frac{\tr(B_{i_t}^{t,\dagger}A_t  \proj{w} A_t^\dagger B_{i_t}^{t}  U)\tr(B_{i_t}^{t,\dagger}A_t  \proj{z} A_t^\dagger B_{i_t}^{t}  V)}{\tr(B_{i_t}^{t,\dagger}A_tA_t^\dagger B_{i_t}^{t})}\right)
\end{align*}

Note that $Z_t$ is a polynomial of degree $2$ of expectation $0$. The Hypercontractivity \citep[Proposition $5.48$]{aubrun2017alice} implies for all $k\in \{1,\dots,N\}$:
\begin{align*}
    \ex{|Z_t|^{k}}\le k^k  \ex{Z_t^{2}}^{k/2}. 
\end{align*}
This means that we only need to control the second moment of $Z_t$. We have:
\begin{align*}
    \mathds{E}_{U,V}(Z_t^{2})= \eps^4  \mathds{E}_{U,V}\Bigg[\sum_{i_t,j_t}\frac{\lambda^t_{i_t}\lambda^t_{j_t}}{D^2\dout^2}  &\left(\frac{\tr(B_{i_t}^{t,\dagger}A_t  \proj{w} A_t^\dagger B_{i_t}^{t}  U)\tr(B_{i_t}^{t,\dagger}A_t  \proj{z} A_t^\dagger B_{i_t}^{t}  V)}{\tr(B_{i_t}^{t,\dagger}A_tA_t^\dagger B_{i_t}^{t})}\right) 
    \\&\cdot\left(\frac{\tr(B_{j_t}^{t,\dagger}A_t  \proj{w} A_t^\dagger B_{j_t}^{t}  U)\tr(B_{j_t}^{t,\dagger}A_t  \proj{z} A_t^\dagger B_{j_t}^{t}  V)}{\tr(B_{j_t}^{t,\dagger}A_tA_t^\dagger B_{j_t}^{t})}\right)\Bigg]. 
\end{align*}
Let $\Xi_{i_t, j_t}(w)= \mathds{E}_U \big[\tr(B_{i_t}^{t,\dagger}A_t  \proj{w} A_t^\dagger B_{i_t}^{t}  U)\tr(B_{j_t}^{t,\dagger}A_t  \proj{w} A_t^\dagger B_{j_t}^{t}  U)\big]$, we have :
\begin{align*}
       \mathds{E}_{U,V}(Z_t^{2}) =  \eps^4 \sum_{i_t,j_t}\frac{\lambda^t_{i_t}\lambda^t_{j_t}}{D^2\dout^2}  \left(\frac{\Xi_{i_t, j_t}(w)\Xi_{i_t, j_t}(z)}{\tr(B_{i_t}^{t,\dagger}A_tA_t^\dagger B_{i_t}^{t})\tr(B_{j_t}^{t,\dagger}A_tA_t^\dagger B_{j_t}^{t})}\right)
\end{align*}
For given $i_t, j_t$, we can upper bound the expectation $\Xi_{i_t, j_t}$:
\begin{align*}
   \Xi_{i_t, j_t}( w ) &= \mathds{E}_U \big[\tr(B_{i_t}^{t,\dagger}A_t  \proj{w} A_t^\dagger B_{i_t}^{t}  U)\tr(B_{j_t}^{t,\dagger}A_t  \proj{w} A_t^\dagger B_{j_t}^{t}  U)\big] \notag
   \\&= \sum_{x,y,x',y'} \ex{U_{x,y}U_{x',y'}} \tr(B_{i_t}^{t,\dagger}A_t  \proj{w} A_t^\dagger B_{i_t}^{t} \ket{x}\bra{y})\tr(B_{j_t}^{t,\dagger}A_t  \proj{w} A_t^\dagger B_{j_t}^{t}  \ket{x'}\bra{y'}) \notag
    \\&\le  \frac{16}{\dout}\sum_{x,y} \bra{w}A_t^\dagger B_{i_t}^t \proj{x} B_{j_t}^{t,\dagger}A_t  \ket{w} \bra{w} A_t^\dagger B_{j_t}^{t}  \proj{y}B_{i_t}^{t,\dagger}A_t  \ket{w} \notag  
    \\&= \frac{16}{\dout}\bra{w}A_t^\dagger B_{i_t}^t  B_{j_t}^{t,\dagger}A_t  \ket{w} \bra{w} A_t^\dagger B_{j_t}^{t}B_{i_t}^{t,\dagger}A_t  \ket{w} \notag
      \\&= \frac{16}{\dout}|\bra{w}A_t^\dagger B_{i_t}^t  B_{j_t}^{t,\dagger}A_t  \ket{w}|^2 
\end{align*}
Therefore we can  upper bound the expectation of $Z_t^2$ as follows:  
\begin{align*}
     \mathds{E}_{U,V}(Z_t^{2}) &=  \eps^4 \sum_{i_t,j_t}\frac{\lambda^t_{i_t}\lambda^t_{j_t}}{D^2\dout^2}  \left(\frac{\Xi_{i_t, j_t}(w)\Xi_{i_t, j_t}(z)}{\tr(B_{i_t}^{t,\dagger}A_tA_t^\dagger B_{i_t}^{t})\tr(B_{j_t}^{t,\dagger}A_tA_t^\dagger B_{j_t}^{t})}\right)
    \\&= \eps^4   \left(\frac{ 16^2}{\dout^2}\right)\sum_{i_t,j_t}\frac{\lambda^t_{i_t}\lambda^t_{j_t}}{D^2\dout^2}\left(\frac{|\bra{w}A_t^\dagger B_{i_t}^t  B_{j_t}^{t,\dagger}A_t  \ket{w}|^2|\bra{z}A_t^\dagger B_{i_t}^t  B_{j_t}^{t,\dagger}A_t  \ket{z}|^2}{\tr(B_{i_t}^{t,\dagger}A_tA_t^\dagger B_{i_t}^{t})\tr(B_{j_t}^{t,\dagger}A_tA_t^\dagger B_{j_t}^{t})} \right)
\end{align*}
This implies an upper bound for every moment, that is for $k=1, \dots, N$, we have:
\begin{align*}
   & \mathds{E}_{\substack{w,z}}\bm{1}\{\cE(w)\}\bm{1}\{\cE(z)\}\mathds{E}_{U,V}(|Z_t|^{k})    \\&\le \mathds{E}_{\substack{w,z}}\bm{1}\{\cE(w)\}\bm{1}\{\cE(z)\} k^k  \Bigg[\eps^4   \left(\frac{16^2}{\dout^2}\right)\sum_{i_t,j_t}\frac{\lambda^t_{i_t}\lambda^t_{j_t}}{D^2\dout^2}\left(\frac{|\bra{w}A_t^\dagger B_{i_t}^t  B_{j_t}^{t,\dagger}A_t  \ket{w}|^2|\bra{z}A_t^\dagger B_{i_t}^t  B_{j_t}^{t,\dagger}A_t  \ket{z}|^2}{\tr(B_{i_t}^{t,\dagger}A_tA_t^\dagger B_{i_t}^{t})\tr(B_{j_t}^{t,\dagger}A_tA_t^\dagger B_{j_t}^{t})} \right)\Bigg]^{k/2}
    \\&\le k^k  \mathds{E}_{\substack{w,z}}\bm{1}\{\cE(w)\}\bm{1}\{\cE(z)\}\Bigg[\eps^4   \left(\frac{ 16^2}{\dout^2}\right)\sum_{i_t,j_t}\frac{\lambda^t_{i_t}\lambda^t_{j_t}}{D^2\dout^2}\left(\frac{|\bra{w}A_t^\dagger B_{i_t}^t  B_{j_t}^{t,\dagger}A_t  \ket{w}|^4+|\bra{z}A_t^\dagger B_{i_t}^t  B_{j_t}^{t,\dagger}A_t  \ket{z}|^4}{2\tr(B_{i_t}^{t,\dagger}A_tA_t^\dagger B_{i_t}^{t})\tr(B_{j_t}^{t,\dagger}A_tA_t^\dagger B_{j_t}^{t})} \right)\Bigg]^{k/2}
    \\&= k^k \mathds{E}_{\substack{w,z}}\bm{1}\{\cE(w)\}\bm{1}\{\cE(z)\} \Bigg[\eps^4   \left(\frac{ 16^2}{\dout^2}\right)\frac{(P_w^t +P_z^t)}{2} \Bigg]^{k/2}
    \\&\le k^k \mathds{E}_{\substack{w,z}}\bm{1}\{\cE(w)\}\bm{1}\{\cE(z)\} \Bigg[\eps^4   \left(\frac{ 16^2\cdot 6^4}{\dout^2}\right)\frac{(P_u^t +P_v^t)}{2} \Bigg]^{k/2}
    \\&\le k^k \mathds{E}_{\substack{w,z}}\bm{1}\{\cE(w)\}\bm{1}\{\cE(z)\} \Bigg[\eps^4   \left(\frac{ 16^2\cdot 6^4}{\dout^2}\right)  \left( \frac{(7\log(N))^4}{\din^4\dout}\right) \Bigg]^{k/2}
    \le k^k  \Bigg[ \frac{C\eps^4\log(N)^4}{\din^4\dout^3}\Bigg]^{k/2}
\end{align*}
where $C>0$ is a universal constant and  we used that we are under the events $\cE(w)$ and $\cE(z)$.
\\Now, grouping the lower bound and upper bounds on the $\TV$ distance, we obtain:
\begin{align*}
   (1-e^{-\Omega(\dout)})^2\left( \frac{7^2}{15^2}+\frac{9^2}{10^2}\right)&\le  (1-e^{-\Omega(\dout)})^2 \left[\left(2\TV\left(\mathds{P}_{H_0}^{I_1,\dots, I_N}\Big\| \mathds{P}_{H_1}^{I_1,\dots, I_N}\right)-2\pr{\bar{E}}   \right)^2+\pr{\cE}^2\right]
   \\&\le (1-e^{-\Omega(\dout)})^2 \left[\left( \mathds{E}_{\le N}\left| \mathds{E}_{\substack{(w,U)\sim\cP}}\left[\bm{1}\{\cE\}(\Psi_{i,w, U}-1)\right]\right|\right)^2+\pr{\cE}^2\right]
   \\&\le \max_t\ex{(1+|Z_t|)^N}\\&\le  \max_t \sum_{k=0}^N\binom{N}{k} k^k \left(\frac{C\eps^4\log(N)^4}{\din^4\dout^3}\right)^{k/2}
    \le \sum_{k=0}^N  \left(\frac{C' N\eps^2\log(N)^2}{\din^{2}\dout^{1.5}}\right)^{k}
\end{align*}
where we used $\binom{N}{k}\le \frac{N^ke^k}{k^k}$ and $C'=\sqrt{C}e$. If $N\log(N)^2\le \frac{\din^{2}\dout^{1.5}}{101C'\eps^2}$ the  RHS is upper bounded by $\sum_{k\ge 0}\frac{1}{101^k}= 1.01$ but the LHS is at least $(1- e^{-\Omega(\dout)})^2\left( \frac{7^2}{15^2}+\frac{9^2}{10^2}\right)\ge 1.02$ for $\dout\ge \Omega(1)$  which is a contradiction. Hence $N\log(N)^2\ge \frac{\din^{2}\dout^{1.5}}{101C'\eps^2}$ and finally:
\begin{align*}
    N\ge \Omega\left(\frac{\din^{2}\dout^{1.5}}{\log(\din\dout/\eps)^2\eps^2}\right). 
\end{align*}
\end{proof}
We move to prove our claim:
\begin{lemma}\label{lem: event E sat}
  Let $\din \ge 80$.   We have with probability at least $9/10$, the event $\cE$ is satisfied.
\end{lemma}
\begin{proof}
Fix $t\in [N]$,  denote $(i,j)=(i_t,j_t)$, $X_i= B^\dagger_i A_t$ and recall that $\ket{w}=\frac{\ket{u}}{\|u\|_2}$ such that $u$ is a  standard Gaussian vector. 
Recall that the event $\cE$ is  $u$ satisfies: 
\begin{enumerate}
    \item $\|u\|_2\; \ge \frac{1}{2}\sqrt{\frac{2}{3}\din} $ \quad \text{ and }
    \item $\forall t\in [N]:P_u^t\le \frac{(7\log(N))^4}{\din^4 \dout}$
\end{enumerate}
where $P_u^t$ is defined as 
\begin{align*}
    P_u^t
    &=\sum_{i_t,j_t}\frac{\lambda^t_{i_t}\lambda^t_{j_t}}{D^2\dout^2\din^4}\left(\frac{|\bra{u}A_t^\dagger B_{i_t}^t  B_{j_t}^{t,\dagger}A_t  \ket{u}|^4}{\tr(B_{i_t}^{t,\dagger}A_tA_t^\dagger B_{i_t}^{t})\tr(B_{j_t}^{t,\dagger}A_tA_t^\dagger B_{j_t}^{t})} \right)=\sum_{i,j} \frac{\lambda_{i}\lambda_{j}}{D^2\dout^2\din^4}\left(  \frac{|\bra{u}X_i^\dagger  X_j   \ket{u}|^4}{\|X_i\|_2^2 \|X_j\|_2^2} \right)
\end{align*}
First note that  with high probability $\|u\|_2\ge  \frac{1}{2}\sqrt{\frac{2}{3}\din}$. Indeed, the function $u\mapsto \|u\|_2$ is Lipschitz so $\pr{ \|u\|_2\le \ex{\|u\|_2} -s }\le e^{-s^2/2}$ \citep[Theorem 2.26]{wainwright2019high}. On the other hand we have by Hölder's inequality $\ex{\|u\|_2}\ge \sqrt{\frac{\ex{\|u\|_2^2}^3}{\ex{\|u\|_2^4}}}=  \sqrt{\frac{\din^3}{\din^2+2\din}}\ge \sqrt{\frac{2}{3}\din }  $ hence 
\begin{align*}
    \pr{\|u\|_2\le\frac{1}{2}\sqrt{\frac{2}{3}\din} }\le   \pr{\|u\|_2-\ex{\|u\|_2}\le- \frac{1}{2}\sqrt{\frac{2}{3}\din} } \le e^{-\din/32}. 
\end{align*}
For the second point in the event $\cE$, we can see that $P_u$ is a  Gaussian polynomial in the entries of $u$  of degree $8$. So we can use the concentration inequality of Gaussian polynomials \cite[Corollary 5.49]{aubrun2017alice}, for any $s>(2e)^{1/4}$ we have:
\begin{align*}
    \pr{|P-\ex{P}| \ge s \sqrt{\var(P)}} \le \exp\left(-\frac{4}{e} s^{1/4} \right)
\end{align*}
so with probability at least $1-\frac{1}{N^2}$ we have: 
\begin{align*}
    P&\le \ex{P} + (e\log(N)/2)^4 \sqrt{\var(P)}
    \le  5\log^4(N) \sqrt{\ex{P^2}}.
\end{align*}
Let us control the expectation of $P^2=P_u^2$, recall that we denote $(i,j)=(i_t,j_t)$ and $X_i= B^\dagger_i A_t$,
\begin{align*}
    &\ex{P^2} = \mathds{E}\Bigg[ \Bigg(\sum_{i,j} \frac{\lambda_{i}\lambda_{j}}{D^2\dout^2\din^4}\left( \|X_i\|_2^{-2} \|X_j\|_2^{-2}|\bra{u}X_i^\dagger  X_j   \ket{u}|^4 \right)\Bigg)^2\Bigg]
    \\&= \mathds{E}\Bigg[ \sum_{i,j,k,l} \frac{\lambda_{i}\lambda_{j}\lambda_k\lambda_l}{D^4\dout^4\din^8}\left( \|X_i\|_2^{-2} \|X_j\|_2^{-2}\|X_k\|_2^{-2} \|X_l\|_2^{-2} |\bra{w}X_i^\dagger  X_j   \ket{w}|^4|\bra{w}X_k^\dagger  X_l   \ket{w}|^4 \right)\Bigg] \ex{\|u\|_2^{16}}
     \\&\le  \sum_{i,j,k,l} \frac{\lambda_{i}\lambda_{j}\lambda_k\lambda_l}{D^4\dout^4\din^8}\cdot\frac{ \sum_{\alpha\in \fS_8} \tr_\alpha(X_i^\dagger  X_j, X_j^\dagger X_i, X_i^\dagger  X_j, X_j^\dagger X_i, X_k^\dagger  X_l, X_l^\dagger X_k, X_k^\dagger  X_l, X_l^\dagger X_k)}{(\din+7)![(\din-1)!]^{-1}\|X_i\|_2^{2} \|X_j\|_2^{2}\|X_k\|_2^{2} \|X_l\|_2^{2}}\cdot\din^8 
      \\&\le \max_{\alpha\in \fS_8} \sum_{i,j,k,l} \frac{8!\lambda_{i}\lambda_{j}\lambda_k\lambda_l}{D^4\dout^4\din^8}\cdot\frac{  |\tr_\alpha(X_i^\dagger  X_j, X_j^\dagger X_i, X_i^\dagger  X_j, X_j^\dagger X_i, X_k^\dagger  X_l, X_l^\dagger X_k, X_k^\dagger  X_l, X_l^\dagger X_k)|}{\|X_i\|_2^{2} \|X_j\|_2^{2}\|X_k\|_2^{2} \|X_l\|_2^{2}}
\end{align*}
Then we claim that:
\begin{lemma}\label{lem: tr_alpha}
    Let $A= X_i^\dagger  X_j$ and $B= X_k^\dagger  X_l$, we have for all $\alpha\in \fS_8$,
\begin{align*}
    &|\tr_\alpha(A,A^\dagger,A,A^\dagger, B, B^\dagger, B, B^\dagger)|^{1/2} 
  \\&\le \max\{ |\tr(A)|^4, \tr(AA^\dagger)^2\}  \max\{ |\tr(B)|^4, \tr(BB^\dagger)^2   \}
    \\&\le \|X_i\|_2^{2} \|X_j\|_2^{2}\|X_k\|_2^{2} \|X_l\|_2^{2}.
\end{align*}
\end{lemma}
\textbf{Sketch of the proof.}
The strategy to prove this lemma depends on $\alpha$, we can give an outline:
\begin{itemize}
    \item Traces containing one element remain unchanged.
    \item If we have two successive $A$ (or $B$) in a trace term,  we  upper bound this trace with a $2$-norm $\tr(AA^\dagger)$ times the operator norm of the remaining elements. Then we can upper bound the operator norms with the $2$-norms.
    \item Otherwise, we have successive $(A,B)$ (with possible adjoint) inside the trace term. We can apply a Cauchy Schwarz inequality then we find two successive $A$ or $B$ and apply the previous strategy.
\end{itemize}
For instance we have for $\alpha=(1)(5372648)$:
\begin{align*}
    \tr(A)\tr(BABA^\dagger B^\dagger A^\dagger B^\dagger )&\le |\tr(A)| \tr(BB^\dagger) \| ABA^\dagger B^\dagger A^\dagger \|_\infty
    \\&\le  |\tr(A)| \tr(BB^\dagger) \|BB^\dagger \|_\infty\|AA^\dagger\|^{3/2}_\infty
    \\&\le  |\tr(A)| \tr(BB^\dagger) \tr(BB^\dagger )\tr(AA^\dagger)^{3/2}
     \\&\le \max\{ |\tr(A)|^4, \tr(AA^\dagger)^2\}  \max\{ |\tr(B)|^4, \tr(BB^\dagger)^2   \}.
\end{align*}
Another example for $\alpha=(15372648)$:
\begin{align*}
    \tr(ABABA^\dagger B^\dagger A^\dagger B^\dagger) &\le \tr(ABABB^\dagger A^\dagger B^\dagger A^\dagger )^{1/2}\tr(BABA A^\dagger B^\dagger A^\dagger B^\dagger )^{1/2}
    \\&\le \tr(BB^\dagger) \|AA^\dagger\|^2_\infty \|BB^\dagger\|_\infty
    \le \tr(BB^\dagger)^2 \tr(AA^\dagger)^2.
\end{align*}
Using Lemma~\ref{lem: tr_alpha}, we have:
\begin{align*}
    &\ex{P^2} \le \max_{\alpha\in \fS_8} \sum_{i,j,k,l} \frac{8!\lambda_{i}\lambda_{j}\lambda_k\lambda_l}{D^4\dout^4\din^8}\cdot\frac{  |\tr_\alpha(X_i^\dagger  X_j, X_j^\dagger X_i, X_i^\dagger  X_j, X_j^\dagger X_i, X_k^\dagger  X_l, X_l^\dagger X_k, X_k^\dagger  X_l, X_l^\dagger X_k)|}{\|X_i\|_2^{2} \|X_j\|_2^{2}\|X_k\|_2^{2} \|X_l\|_2^{2}}
    \\&\le  \max_{\alpha\in \fS_8} \sum_{i,j,k,l} \frac{8!\lambda_{i}\lambda_{j}\lambda_k\lambda_l}{D^4\dout^4\din^8}\cdot \sqrt{|\tr_\alpha(X_i^\dagger  X_j, X_j^\dagger X_i, X_i^\dagger  X_j, X_j^\dagger X_i, X_k^\dagger  X_l, X_l^\dagger X_k, X_k^\dagger  X_l, X_l^\dagger X_k)|}
    \\&\le  \sum_{i,j,k,l} \frac{8!\lambda_{i}\lambda_{j}\lambda_k\lambda_l}{D^4\dout^4\din^8}\cdot \max\{ |\tr(X_i^\dagger X_j)|^2, \tr(X_i^\dagger X_j X_j^\dagger X_i)\}  \max\{ |\tr(X_k^\dagger X_l)|^2, \tr(X_k^\dagger X_lX_l^\dagger X_k) \} 
\end{align*}
So we have $4$ cases to consider, two of them are similar. Let us define four subsets of indices: 
\begin{align*}
    C_1&=\{(i,j,k,l): |\tr(X_i^\dagger X_j)|^2>\tr(X_i^\dagger X_j X_j^\dagger X_i), |\tr(X_k^\dagger X_l)|^2>\tr(X_k^\dagger X_lX_l^\dagger X_k)\}, \\
     C_2&=\{(i,j,k,l):|\tr(X_i^\dagger X_j)|^2\le \tr(X_i^\dagger X_j X_j^\dagger X_i), |\tr(X_k^\dagger X_l)|^2\le \tr(X_k^\dagger X_lX_l^\dagger X_k)\},\\
    C_3&=\{(i,j,k,l):|\tr(X_i^\dagger X_j)|^2> \tr(X_i^\dagger X_j X_j^\dagger X_i), |\tr(X_k^\dagger X_l)|^2\le\tr(X_k^\dagger X_lX_l^\dagger X_k)\}, \\
    C_4&=\{(i,j,k,l):|\tr(X_i^\dagger X_j)|^2\le \tr(X_i^\dagger X_j X_j^\dagger X_i), |\tr(X_k^\dagger X_l)|^2>\tr(X_k^\dagger X_lX_l^\dagger X_k)\}.
\end{align*}
For the first case where the indices $(i,j,k,l)\in C_1$, we have:
\begin{align*}
    &\sum_{i,j,k,l\in C_1}  \frac{8!\lambda_{i}\lambda_{j}\lambda_k\lambda_l}{D^4\dout^4\din^8}\cdot \max\{ |\tr(X_i^\dagger X_j)|^2, \tr(X_i^\dagger X_j X_j^\dagger X_i)\}  \max\{ |\tr(X_k^\dagger X_l)|^2, \tr(X_k^\dagger X_lX_l^\dagger X_k) \} 
    \\&\le\sum_{i,j,k,l} \frac{8!\lambda_{i}\lambda_{j}\lambda_k\lambda_l}{D^4\dout^4\din^8}\cdot |\tr (X_i^\dagger  X_j)|^2\tr( X_k^\dagger  X_l)|^2 
    \\&=\sum_{i,j,k,l}\sum_{x,y,x',y'} \frac{8!\lambda_{i}\lambda_{j}\lambda_k\lambda_l}{D^4\dout^4\din^8}\cdot \bra{x}X_i^\dagger  X_j\ket{x}\bra{y}X_j^\dagger  X_i\ket{y}   \bra{x'}X_k^\dagger  X_l\ket{x'}\bra{y'}X_l^\dagger  X_k\ket{y'}  
    \\&= \sum_{x,y,x',y'} \frac{8!}{D^4\dout^4\din^8} \cdot \dout^4\tr[\tr(A_t\ket{x}\bra{y} A_t^\dagger)\dI_{\dout} \tr(A_t\ket{y}\bra{x} A_t^\dagger)\dI_{\dout} ]   
    \\&\qquad\qquad\qquad\qquad\qquad\quad \cdot \tr[\tr(A_t\ket{x'}\bra{y'} A_t^\dagger)\dI_{\dout} \tr(A_t\ket{y'}\bra{x'} A_t^\dagger)\dI_{\dout} ] 
    \\&= \sum_{x,y,x',y'} \frac{8!}{D^4\dout^4\din^8}\cdot \dout^6 \bra{y} A_t^\dagger A_t \ket{x}\bra{x} A_t^\dagger A_t \ket{y} \bra{y'} A_t^\dagger A_t \ket{x'}\bra{x'} A_t^\dagger A_t \ket{y'}
     \\&= \frac{8!}{D^4\dout^4\din^8}\cdot \dout^6  \tr( A_t^\dagger A_t A_t^\dagger A_t)^2
\le \frac{8!}{D^4\dout^4\din^8}\cdot \dout^6  \tr( A_t^\dagger A_t )^4
    \\& =  \frac{8!}{D^4\dout^4\din^8}\cdot \dout^6 \din^4= \frac{8!}{\din^8\dout^2}
\end{align*}
because $\sum_{i} \lambda_i X_i M X_i^\dagger = \sum_{i} \lambda_i B_i^\dagger A_t M A_t^\dagger B_i= \dout\tr(A_t MA_t^\dagger )\dI_{\dout}$ and $\tr(A_tA_t^\dagger )=\din$.
\\ For the second case where the indices $(i,j,k,l)\in C_2$, we have:
\begin{align*}
  &\sum_{i,j,k,l\in C_2}  \frac{8!\lambda_{i}\lambda_{j}\lambda_k\lambda_l}{D^4\dout^4\din^8}\cdot \max\{ |\tr(X_i^\dagger X_j)|^2, \tr(X_i^\dagger X_j X_j^\dagger X_i)\}  \max\{ |\tr(X_k^\dagger X_l)|^2, \tr(X_k^\dagger X_lX_l^\dagger X_k) \} 
  \\&\le \sum_{i,j,k,l} \frac{\lambda_{i}\lambda_{j}\lambda_k\lambda_l}{D^4\dout^4\din^8}\cdot\tr(X_i^\dagger  X_j X_j^\dagger X_i)\tr( X_k^\dagger  X_l X_l^\dagger X_k)
     \\&= \frac{8!}{D^4\dout^4\din^8} \tr(A_tA_t^\dagger ) ^2 \tr(\dI) \tr(A_tA_t^\dagger ) ^2 \tr(\dI) \dout^4
    \\& = \frac{8!}{D^4\dout^4\din^8}  \din^4 \dout^6= \frac{8!}{\din^8\dout^2}. 
\end{align*}
\\ For the third case where the indices $(i,j,k,l)\in C_3$, we have: 
\begin{align*}
    &\sum_{i,j,k,l\in C_3}  \frac{8!\lambda_{i}\lambda_{j}\lambda_k\lambda_l}{D^4\dout^4\din^8}\cdot \max\{ |\tr(X_i^\dagger X_j)|^2, \tr(X_i^\dagger X_j X_j^\dagger X_i)\}  \max\{ |\tr(X_k^\dagger X_l)|^2, \tr(X_k^\dagger X_lX_l^\dagger X_k) \} 
  \\&\le  \sum_{i,j,k,l} \frac{\lambda_{i}\lambda_{j}\lambda_k\lambda_l}{D^4\dout^4\din^8}\cdot|\tr(X_i^\dagger  X_j)|^2\tr( X_k^\dagger  X_l X_l^\dagger X_k)
   \\&= \sum_{i,j,k,l} \frac{\lambda_{i}\lambda_{j}\lambda_k\lambda_l}{D^4\dout^4\din^8}\cdot \sum_{x,y}\bra{x}X_i^\dagger  X_j\ket{x}\bra{y}X_j^\dagger  X_i\ket{y}\tr( X_k^\dagger  X_l X_l^\dagger X_k)
   \\&= \sum_{i,j,k,l} \frac{\lambda_{i}\lambda_{j}\lambda_k\lambda_l}{D^4\dout^4\din^8}\cdot \sum_{x,y} \dout^3\bra{y}A_t^\dagger A_t\ket{x} \bra{x}A_t^\dagger A_t\ket{y}  \dout^2 \tr(A_tA_t^\dagger)^2 \tr(\dI)
   \\&\le  \sum_{i,j,k,l} \frac{\lambda_{i}\lambda_{j}\lambda_k\lambda_l}{D^4\dout^4\din^8}\cdot  \dout^6 \tr(A_tA_t^\dagger)^4 
     = \frac{8!}{D^4\dout^4\din^8} \dout^6\din^4 
   = \frac{8!}{\din^8\dout^2}. 
\end{align*}
The fourth case, where the indices $(i,j,k,l)\in C_4$, is similar to the previous one.
\\To sum up, we have 
\begin{align*}
    \ex{P^2} 
    &\le  \sum_{i,j,k,l} \frac{8!\lambda_{i}\lambda_{j}\lambda_k\lambda_l}{D^4\dout^4\din^8}\cdot \max\{ |\tr(X_i^\dagger X_j)|^2, \tr(X_i^\dagger X_j X_j^\dagger X_i)\}  \max\{ |\tr(X_k^\dagger X_l)|^2, \tr(X_k^\dagger X_lX_l^\dagger X_k) \} 
    \\&\le \frac{4\times 8!}{\din^8\dout^2}.
\end{align*}
So by the union bound, we have with probability at least $1-\frac{1}{N}- e^{-\din/32}\ge \frac{9}{10}$, for all $t\in [N]$:
\begin{align*}
    P_u\le 5\log^4(N)^4 \sqrt{\ex{P_u^2}} \le 5\log^4(N) \sqrt{\frac{4\times 8!}{\din^8\dout^2}}\le \frac{(7\log(N))^4 }{\din^4\dout}
\end{align*}
and $\|u\|_2\ge \sqrt{\din/6}$, that is the event $\cE$ is satisfied.

\end{proof}
\subsection{Proof of Thm.~\ref{thm:LB_DEP_adap}}\label{thm-app:LB_DEP_adap}
\begin{theorem}
    Any incoherent ancilla-assisted adaptive algorithm for testing identity to the depolarizing channel requires, in the worst case, 
    \begin{align*}
        N=\Omega\left(\frac{\din^2\dout+ \dout^{1.5}}{\eps^2}\right)
    \end{align*}
  measurements.
\end{theorem}
\begin{proof}
    We use the same construction as in the proof of Thm.~\ref{thm:LB_DEP_non-adap}.  Mainly, under the null hypothesis $H_0$ the quantum channel is $\cN(\rho)=\cD(\rho)= \tr(\rho)\frac{\dI}{\dout}$. Under the alternate hypothesis $H_1$, a quantum channel $\cN\sim \cP$ has the form:
	 \begin{align*}
	 	\cN(\rho)= \tr(\rho) \frac{\dI}{\dout} +\frac{\eps}{\dout} \bra{w}\rho \ket{w} U
	 \end{align*}
	 where $\ket{w}= W\ket{0}$ and $W \sim \Haar(\din)$ and for all $i,j\in [\dout]$,  $U_{j,i}=U_{i,j} \sim \mathbf{1}\{i\neq j\}  \cN_c\left(0, \sigma^2=\frac{16}{\dout}\right)$ and we condition on the event $\cG=\{\|U\|_1\ge \dout, \|U\|_\infty \le 32\}$. We call this distribution $\cP$ and use the notation $(w,U)\sim \cP$. Recall that $\pr{\cG}\ge 1-\exp(-\Omega(\dout))-\exp(-\Omega(\dout^2))$ and under $\cG$, the map $\cN$ is a valid quantum channel $\eps$-far from the quantum channel $\cD$. 
  
  Now, given a set of observations $i_{<t}=(i_1,\dots, i_{t-1})$. The adaptive algorithm $\cA$ would choose at step $t$ the input $\rho_t^{i_{<t}}$ and the measurement device $\cM_t^{i_{<t}}=\left\{\lambda_{i_t}^{t, i_{<t} } \proj{\phi_{i_t}^{t, i_{<t}}}\right\}_{i_t\in \cI_t}$. 

Under the null hypothesis $H_0$, the quantum channel $\cN=\cD$ so the probability of the outcomes is exactly:
\begin{align*}
   \mathds{P}_{H_0}^{I_1,\dots, I_N}= \left\{\prod_{t=1}^N\lambda_{i_t}^{t, i_{<t} } \bra{\phi_{i_t}^{t, i_{<t}}}\id\otimes \cD (\rho_t^{i_{<t}}) \ket{\phi_{i_t}^{t, i_{<t}}}\right\}_{i_1,\dots, i_N} .
\end{align*}
On the other hand, under the alternate hypothesis $H_1$, the probability of the outcomes is exactly:
\begin{align*}
   \mathds{P}_{H_1, (w,U)}^{I_1,\dots, I_N}= \left\{\prod_{t=1}^N\lambda_{i_t}^{t, i_{<t} } \bra{\phi_{i_t}^{t, i_{<t}}}\id\otimes \cN_{(w,U)} (\rho_t^{i_{<t}}) \ket{\phi_{i_t}^{t, i_{<t}}}\right\}_{i_1,\dots, i_N}
\end{align*}
Recall that
\begin{align*}
   \cJ_\cD= \frac{\dI}{D} ~~\text{ and }~~\cJ_\cN= \frac{\dI}{D}+ \frac{\eps}{D}  \proj{w} \otimes U.
\end{align*}
We can suppose w.l.o.g. that each input state is pure because of the convexity of the $\KL$ divergence. For $t\in [N]$ and ${i_{<t}}$, let $\rho_t^{i_{<t}}= \proj{\psi_t^{i_{<t}}}$. Moreover, we can write each rank one input state and measurement vector as follows:
\begin{align*}
    \ket{\psi_t^{i_{<t}}}= A_t^{i_{<t}}\otimes \dI \ket{\Psi_{\din}} ~~\text{ and }~~\ket{\phi_{i_t}^{t, i_{<t}}}=  B_{i_t}^{i_{<t}}\otimes \dI \ket{\Psi_{\dout}}
\end{align*}
where $\ket{\Psi_d}= \frac{1}{\sqrt{d} } \sum_{i=1}^{d} \ket{ii}$ and the matrices $A_t^{i_{<t}} \in \dC^{\danc \times \din }$ and $B_{i_t}^{i_{<t}} \in \dC^{\danc \times \dout }$ verify: 
\begin{align*}
   \tr(A_t^{i_{<t}}A_t{i_{<t}, \dagger} )= \din  ~~\text{ and }~~  \tr(B_{i_t}^{i_{<t}} B_{i_t}{i_{<t}, \dagger} )= \dout.
\end{align*}
Note that we have for all $t$, $\sum_{i_t} \lambda_{i_t}^{t, i_{<t} }B_{i_t}^{i_{<t}, \dagger}XB_{i_t}^{i_{<t}} = \dout\tr(X)\dI$. Indeed, the condition of the POVM $\cM_t$ implies: 
\begin{align*}
 X\otimes \dI=    X\otimes \dI\sum_{i_t}\lambda_{i_t}^t \proj{\phi_{i_t}^t} =  \sum_{i_t}\lambda_{i_t}^t XB_{i_t}^t\otimes \dI \proj{\Psi_{\dout}} B_{i_t}^{t, \dagger}\otimes\dI 
\end{align*}
hence by taking the partial trace
\begin{align*}
    \tr(X) \dI = \frac{1}{\dout}\sum_{i_t, i, j}\lambda_{i_t}^t \bra{j}  B_{i_t}^{t, \dagger}X B_{i_t}^{t} \ket{i}  \ket{i}\bra{j}.
\end{align*}
Finally
\begin{align*}
    \sum_{i_t} \lambda_{i_t}^t B_{i_t}^{t,\dagger}XB_{i_t}^{t} = \dout\tr(X)\dI.
\end{align*}
Now we can re-write the  distribution  of the observations under the null hypothesis as:
\begin{align*}
   \mathds{P}_{H_0}^{I_1,\dots, I_N}&= \left\{\prod_{t=1}^N\lambda_{i_t}^t \bra{\phi_{i_t}^t}\id\otimes \cD (\rho_t) \ket{\phi_{i_t}^t}\right\}_{i_1,\dots, i_N} 
   \\&=\left\{\prod_{t=1}^N\lambda_{i_t}^t \bra{\Psi_{\dout}}(B_{i_t}^{t,\dagger}A_t \otimes \dI)\cJ_\cD (A_t^\dagger B_{i_t}^{t}\otimes \dI) \ket{\Psi_{\dout}}
   \right\}_{i_1,\dots, i_N} 
   \\&= \left\{\prod_{t=1}^N\frac{\lambda_{i_t}^t}{D} \bra{\Psi_{\dout}}(B_{i_t}^{t,\dagger}A_tA_t^\dagger B_{i_t}^{t} \otimes \dI) \ket{\Psi_{\dout}}
   \right\}_{i_1,\dots, i_N} 
   \\&= \left\{\prod_{t=1}^N\frac{\lambda_{i_t}^t}{D\dout} \tr(B_{i_t}^{t,\dagger}A_tA_t^\dagger B_{i_t}^{t})
   \right\}_{i_1,\dots, i_N} ,
\end{align*}
and under the alternate hypothesis as: 
\begin{align*}
   \mathds{P}_{H_1}^{I_1,\dots, I_N}&= \left\{\prod_{t=1}^N\lambda_{i_t}^t \bra{\phi_{i_t}^t}\id\otimes \cN (\rho_t) \ket{\phi_{i_t}^t}\right\}_{i_1,\dots, i_N}
  \\&=\left\{\prod_{t=1}^N\lambda_{i_t}^t \bra{\Psi_{\dout}}(B_{i_t}^{t,\dagger}A_t \otimes \dI)\cJ_\cN (A_t^\dagger B_{i_t}^{t}\otimes \dI) \ket{\Psi_{\dout}}
   \right\}_{i_1,\dots, i_N} 
    \\&=\left\{\prod_{t=1}^N\lambda_{i_t}^t \bra{\Psi_{\dout}}(B_{i_t}^{t,\dagger}A_t \otimes \dI)\left(\frac{\dI}{D}+ \frac{\eps}{D}  \proj{w} \otimes U \right)(A_t^\dagger B_{i_t}^{t}\otimes \dI) \ket{\Psi_{\dout}}
   \right\}_{i_1,\dots, i_N}
    \\&=\left\{\prod_{t=1}^N \left( \frac{\lambda_{i_t}^t}{D\dout} \tr(B_{i_t}^{t,\dagger}A_tA_t^\dagger B_{i_t}^{t})+ \frac{\eps\lambda_{i_t}^t}{D} \bra{\Psi_{\dout}}(B_{i_t}^{t,\dagger}A_t  \proj{w} A_t^\dagger B_{i_t}^{t}  \otimes U) \ket{\Psi_{\dout}}\right)
   \right\}_{i_1,\dots, i_N}
\end{align*}

We can express the $\KL$ divergence as follows: 
\begin{align*}
   & \KL\left(\mathds{P}_{H_0}^{I_1,\dots, I_N}\Big\| \mathds{P}_{H_1, (w,U)}^{I_1,\dots, I_N}\right)
    \\&=\sum_{i_1,\dots, i_N}\prod_{t=1}^N \frac{\lambda_{i_t}^t}{D\dout} \tr(B_{i_t}^{t,\dagger}A_tA_t^\dagger B_{i_t}^{t})
    \\&\quad\quad\quad\quad\quad\quad \times\log\left(\frac{\prod_{t=1}^N \frac{\lambda_{i_t}^t}{D\dout} \tr(B_{i_t}^{t,\dagger}A_tA_t^\dagger B_{i_t}^{t})}{\prod_{t=1}^N \left( \frac{\lambda_{i_t}^t}{D\dout} \tr(B_{i_t}^{t,\dagger}A_tA_t^\dagger B_{i_t}^{t})+ \frac{\eps\lambda_{i_t}^t}{D} \bra{\Psi_{\dout}}(B_{i_t}^{t,\dagger}A_t  \proj{w} A_t^\dagger B_{i_t}^{t}  \otimes U) \ket{\Psi_{\dout}}\right)}\right)
   \\&=\sum_{t=1}^N  \mathds{E}_{\le N}(-\log)\left(1+\frac{\frac{\eps\lambda_{i_t}^t}{D} \bra{\Psi_{\dout}}(B_{i_t}^{t,\dagger}A_t  \proj{w} A_t^\dagger B_{i_t}^{t}  \otimes U) \ket{\Psi_{\dout}}}{\frac{\lambda_{i_t}^t}{D\dout} \tr(B_{i_t}^{t,\dagger}A_tA_t^\dagger B_{i_t}^{t})} \right)  
     \\&=\sum_{t=1}^N  \mathds{E}_{\le N}(-\log)\left(1+\eps\frac{ \tr((B_{i_t}^{t,\dagger}A_t  \proj{w} A_t^\dagger B_{i_t}^{t} U)) }{ \tr(B_{i_t}^{t,\dagger}A_tA_t^\dagger B_{i_t}^{t})} \right)  
\end{align*}
where we use the notation $\mathds{E}_{\le N}(X(i_1,\dots, i_N))= \sum_{i_1,\dots, i_N}\left(\prod_{t=1}^N \frac{\lambda_{i_t}^t}{D\dout} \tr(B_{i_t}^{t,\dagger}A_tA_t^\dagger B_{i_t}^{t})\right)X(i_1,\dots, i_N)$. 
Using the inequality $(-\log(1+x))\le -x+x^2 $ valid for $x\ge -1/2$ and since for $\eps \le 64$, we have $\eps\frac{ \tr((B_{i_t}^{t,\dagger}A_t  \proj{w} A_t^\dagger B_{i_t}^{t} U)) }{ \tr(B_{i_t}^{t,\dagger}A_tA_t^\dagger B_{i_t}^{t})} \ge -1/2$, we can upper bound the previous integrand as follows:
\begin{align*}
    &(-\log) \left(1+ \eps\frac{ \tr((B_{i_t}^{t,\dagger}A_t  \proj{w} A_t^\dagger B_{i_t}^{t} U)) }{ \tr(B_{i_t}^{t,\dagger}A_tA_t^\dagger B_{i_t}^{t})}\right)
    \\&\le \eps\frac{ \tr((B_{i_t}^{t,\dagger}A_t  \proj{w} A_t^\dagger B_{i_t}^{t} U)) }{ \tr(B_{i_t}^{t,\dagger}A_tA_t^\dagger B_{i_t}^{t})} +\eps^2\frac{ \tr((B_{i_t}^{t,\dagger}A_t  \proj{w} A_t^\dagger B_{i_t}^{t} U))^2 }{ \tr(B_{i_t}^{t,\dagger}A_tA_t^\dagger B_{i_t}^{t})^2}. 
\end{align*}
Observe that the first term vanishes under the expectation:
\begin{align*}
 & \mathds{E}_{\le N} \left(\eps\frac{ \tr((B_{i_t}^{t,\dagger}A_t  \proj{w} A_t^\dagger B_{i_t}^{t} U)) }{ \tr(B_{i_t}^{t,\dagger}A_tA_t^\dagger B_{i_t}^{t})}\right)
  \\&=   \mathds{E}_{\le t-1}\sum_{i_t} \frac{\lambda_{i_t}^t}{D\dout} \tr(B_{i_t}^{t,\dagger}A_tA_t^\dagger B_{i_t}^{t})\cdot \eps\frac{ \tr((B_{i_t}^{t,\dagger}A_t  \proj{w} A_t^\dagger B_{i_t}^{t} U)) }{ \tr(B_{i_t}^{t,\dagger}A_tA_t^\dagger B_{i_t}^{t})}
  \\&=   \mathds{E}_{\le t-1}\sum_{i_t} \frac{\lambda_{i_t}^t}{D\dout} \cdot\eps\tr((B_{i_t}^{t,\dagger}A_t  \proj{w} A_t^\dagger B_{i_t}^{t} U)) 
  \\&= \mathds{E}_{\le t-1} \frac{\eps}{D} \tr(A_t  \proj{w} A_t^\dagger)\tr(U)=0
\end{align*}
where we used $  \sum_{i_t} \lambda_{i_t}^t B_{i_t}^{t,\dagger}XB_{i_t}^{t} = \dout\tr(X)\dI.$
For the second term, we will instead upper bound its expectation under $U$. Observe that this term is nonnegative, so we can safely remove the condition on the event $\cG$ and then we compute the expectation under $\Haar$  distributed vector $w$ and Gaussians $\{U_{i,j}\}$.
\begin{align*}
    \ex{\eps^2 \frac{ \tr((B_{i_t}^{t,\dagger}A_t  \proj{w} A_t^\dagger B_{i_t}^{t} U))^2 }{ \tr(B_{i_t}^{t,\dagger}A_tA_t^\dagger B_{i_t}^{t})^2}}
    &\le \frac{32}{\dout}\eps^2\ex{ \frac{\bra{w}A_t^{ \dagger }B_{i_t}^t B_{i_t}^{t, \dagger}  A_t^{ }\ket{w}^2}{\tr(B_{i_t}^{t,\dagger}A_tA_t^\dagger B_{i_t}^{t})^2} }
    \\&\le \frac{64\eps^2}{\dout \din^2} \cdot \frac{\tr(A_t^{ \dagger }B_{i_t}^{t} B_{i_t}^{t, \dagger}  A_t^{ })^2}{\tr(B_{i_t}^{t,\dagger}A_tA_t^\dagger B_{i_t}^{t})^2}
    = \frac{64\eps^2}{\din^2\dout}. 
\end{align*}
Therefore 
\begin{align*}
    \mathds{E}_{(w,U)\sim\cP}\KL\left(\mathds{P}_{H_0}^{I_1,\dots, I_N}\Big\| \mathds{P}_{H_1, (w,U)}^{I_1,\dots, I_N}\right)&\le  \sum_{t=1}^N\mathds{E}_{\le N} \mathds{E}_{(w,U)\sim\cP}(-\log) \left(1+\eps\frac{ \tr((B_{i_t}^{t,\dagger}A_t  \proj{w} A_t^\dagger B_{i_t}^{t} U)) }{ \tr(B_{i_t}^{t,\dagger}A_tA_t^\dagger B_{i_t}^{t})} \right)
    \\&\le \sum_{t=1}^N\mathds{E}_{\le N} \frac{64\eps^2}{\din^2\dout}= \frac{64N\eps^2}{\din^2\dout}. 
\end{align*}
On the other hand, the Data-Processing inequality applied on the $\KL$ divergence (see Prop.~\ref{KL-sec}) writes:
 \begin{align*}
   \mathds{E}_{(w,U)\sim\cP}  \KL\left(\mathds{P}_{H_0}^{I_1,\dots, I_N}\Big\| \mathds{P}_{H_1, (w,U)}^{I_1,\dots, I_N}\right)
     &\ge \mathds{E}_{(w,U)\sim\cP}\KL(\mathds{P}_{H_0}\left(\cA=0\right)\|\mathds{P}_{H_1, (w,U)}\left(\cA=0\right))
     \\&\ge \KL(2/3\|1/3)=\frac{2}{3}\log(2)-\frac{1}{3}\log(2)=\frac{1}{3}\log(2).
 \end{align*}
 Grouping the lower and upper bounds on the $\KL$ divergence:
 \begin{align*}
  \frac{128N\eps^2}{\din^2\dout}  \ge \mathds{E}_{(w,U)\sim\cP} \KL\left(\mathds{P}_{H_0}^{I_1,\dots, I_N}\Big\| \mathds{P}_{H_1, (w,U)}^{I_1,\dots, I_N}\right)\ge \frac{1}{3}\log(2)
 \end{align*}
which yields the lower bound:
\begin{align*}
    N=\Omega\left(\frac{\din^2\dout}{\eps^2}\right). 
\end{align*}
\end{proof}

\section{Testing identity of quantum states}\label{app: testing states}
In this section, we show how to reduce testing quantum states to testing discrete distributions. The same result with another proof can be found in \citep{chen2022toward}. 
For a POVM $\cM$ and a quantum state $\rho$, let  $\rho(\cM)$ denotes the classical probability distribution $\{\tr(\rho M_i)\}_{i}$.
The following lemma captures the main ingredient of the reduction:
\begin{lemma}\label{TV}
For all $\delta>0$, let $l=\frac{1}{4}\log(2/\delta)$ and  $U^1,U^2,\dots,U^l\in \mathbb{C}^{d\times d} $ be $\Haar$ -random unitary matrices of columns $\left\{\ket{U_i^j}\right\}_{1\le i\le d, 1\le j \le l} $, $\cM=\left\{\frac{1}{l}\proj{U_i^j}\right\}_{i,j}$   is a POVM 
and for all quantum states $\rho$ and $\sigma$ we have with a probability at least $1-\delta$:
\begin{align*}
    \TV(\rho(\cM),\sigma(\cM)) \ge\frac{\|\rho-\dI/d\|_2}{20}.
\end{align*}
\end{lemma}
A similar statement can be found in \citep{matthews2009distinguishability} where the authors 
analyze the uniform POVM and a POVM defined by a spherical $4$-designs. However, for our reduction, it is important to minimize the number of outcomes of the POVM.
\begin{proof}
Let $\xi=\rho -\sigma$, we have $U\ket{e_i}=\ket{U_i}$ and we use  Weingarten Calculus \ref{lem:Wg} and \ref{lem:wg2} to calculate 
\begin{align*}
\mathds{E}\left[\bra{U_i}\xi\ket{U_i}^2\right]&=\mathds{E}\left[\bra{U_i}\xi\ket{U_i}\bra{U_i}\xi\ket{U_i}\right]
=\mathds{E}\left[\tr(\xi\proj{U_i}\xi\proj{U_i})\right]
\\&=\mathds{E}\left[\tr(\xi U\proj{e_i}U^*\xi U\proj{e_i}U^*)\right]
=\mathds{E}\left[\tr(U^*\xi U\proj{e_i}U^*\xi U\proj{e_i})\right]
\\&=\sum_{\alpha,\beta \in \cS_2}\W(\beta\alpha^{-1},d)\tr_{\beta^{-1}}(\xi,\xi)\tr_{\alpha}(\proj{e_i},\proj{e_i})
=\frac{1}{d(d+1)}\tr(\xi^2).
\end{align*}
Similarly 
\begin{align*}
\mathds{E}\left[\bra{U_i}\xi\ket{U_i}^4\right]&=\mathds{E}\left[\bra{U_i}\xi\ket{U_i}\bra{U_i}\xi\ket{U_i}\bra{U_i}\xi\ket{U_i}\bra{U_i}\xi\ket{U_i}\right]
\\&=\mathds{E}\left[\tr(\xi\proj{U_i}\xi\proj{U_i}\xi\proj{U_i}\xi\proj{U_i})\right]
\\&=\mathds{E}\left[\tr(\xi U\proj{e_i}U^*\xi U\proj{e_i}U^*\xi U\proj{e_i}U^*\xi U\proj{e_i}U^*)\right]
\\&=\mathds{E}\left[\tr(U^*\xi U\proj{e_i}U^*\xi U\proj{e_i}U^*\xi U\proj{e_i}U^*\xi U\proj{e_i})\right]
\\&=\sum_{\alpha,\beta \in \fS_4}\W(\beta\alpha^{-1},d)\tr_{\beta^{-1}}(\xi,\xi,\xi,\xi)\tr_{\alpha}(\proj{e_i},\proj{e_i},\proj{e_i},\proj{e_i})
\\&=\frac{1}{d(d+1)(d+2)(d+3)}(6\tr(\xi^2)^2+6\tr(\xi^4)) .
\\&\le\frac{c'}{d(d+1)(d+2)(d+3)}\tr(\xi^2)^2 .
\end{align*}
We can now conclude by Hölder's inequality:
\begin{align*}
2\mathds{E}\left[ \TV(\rho(\cM),\sigma(\cM)) \right]&=\sum_{i=1}^d \mathds{E}\left[ |\bra{U_i}\xi\ket{U_i}| \right]
\ge \sum_{i=1}^d  \sqrt{\frac{\left(\mathds{E}\left[\bra{U_i}\xi\ket{U_i}^2\right]\right)^3}{\mathds{E}\left[\bra{U_i}\xi\ket{U_i}^4\right]}}
\\&\ge \sum_{i=1}^d  \sqrt{\frac{(d^{-1}(d+1)^{-1}\tr(\xi^2))^3}{c' d^{-1}(d+1)^{-1}(d+2)^{-1}(d+3)^{-1}\tr(\xi^2)^2}}
\\&\ge \sum_{i=1}^d c \frac{\sqrt{\tr{(\xi ^2)}}}{d}
\ge  c \sqrt{\tr{(\rho -\sigma)^2}}.
\end{align*}

Let $f(U)=\TV(\rho(\cM),\sigma(\cM))$, we first show that $f$ is Lipschitz by using the triangle and Cauchy Schwarz inequalities:
\begin{align*}
2|f(U)-f(V)|&=\left| \sum_{1\le i\le d, 1\le j\le l} \frac{1}{l}|\tr(\proj{U_i^j}\xi)|-|\tr(\proj{V_i^j}\xi)| \right|
\\&\le  \sum_{1\le i\le d, 1\le j\le l} \frac{1}{l}\left|\tr((\proj{U_i^j}-\proj{V_i^j})\xi)\right|
\\&\le  \sum_{1\le i\le d, 1\le j\le l} \frac{1}{l}\sqrt{\tr(\xi^2)}\sqrt{\tr((\proj{U_i^j}-\proj{V_i^j})^2)} 
\\&\le  \sqrt{\frac{d}{l}}\sqrt{\tr(\xi^2)}\sqrt{\sum_{1\le i\le d, 1\le j\le l} \tr((\proj{U_i^j}-\proj{V_i^j})^2)} 
\\&\le  \sqrt{\frac{d}{l}}\sqrt{\tr(\xi^2)}\sqrt{\sum_{ 1\le j\le l} \tr((U^j-V^j)^2)} 
\\&\le  \sqrt{\frac{d}{l}}\sqrt{\tr(\xi^2)}\|U-V\|_{2,\text{HS}},
\end{align*}
hence $f$ is $L=\sqrt{\frac{d}{2l}}\sqrt{\tr(\xi^2)}$-Lipschitz, therefore by Thm.~\ref{thm:concentration}: 
\begin{align*}
\pr{\left|f(U)-\ex{f(U)}\right|>\frac{c}{2}\sqrt{\tr{(\xi^2)}}} \le e^{-\frac{dc^2\tr(\xi^2)}{48L^2}}=e^{-lc^2/24}=\delta/2,
\end{align*}
for $l=24\log(2/\delta)/c^2$. Finally with high probability (at least $ 1-\delta/2$) we have 
\begin{align*}
\TV(\rho(\cM),\sigma(\cM))&\ge \ex{\TV(\rho(\cM),\sigma(\cM))}-|\TV(\rho(\cM),\sigma(\cM))-\ex{\TV(\rho(\cM),\sigma(\cM))}|
\\&\ge c\sqrt{\tr(\xi^2)}-\frac{c}{2}\sqrt{\tr(\xi^2)}
\ge \frac{c}{2}\sqrt{\tr(\xi^2)}.
\end{align*}

\end{proof}
Let $\eta = \|\rho-\dI/d\|_2$.  Under the alternative hypothesis $H_2$, the $\TV$ distance between $P$ and $U_n$ can be lower bounded by $\TV(P,U_n)\ge \frac{1}{20}\|\rho-\dI/d\|_2$. So Lem.~\ref{TV} gives a POVM for which our problem reduces to testing identity: $P=U_n$ vs $\TV(P,U_n)\ge \frac{\eta}{20}$  with high probability, where $n=\frac{1}{4}d\log(2/\delta)$ and $P=\cM(\rho)$.
Therefore we can apply the classical testing uniform result of \citep{diakonikolas2017optimal} to obtain a copy complexity: 
\begin{align*}
\mathcal{O}\left( \frac{\sqrt{d}\log(1/\delta)}{\eta^2}\right).
\end{align*}

\begin{algorithm}
\caption{Testing identity of quantum states  in the $2$-norm with an error probability at most $\delta$.}\label{alg-states}
\begin{algorithmic}[t!]
\STATE $l=\frac{1}{4}\log(2/\delta)$.
\STATE Sample   $U^1,U^2,\dots,U^l\in \mathbb{C}^{d\times d} $  from  $\Haar(d)$ distribution. 
\STATE Let  $\left\{\ket{U_i^j}\right\}_{1\le i\le d, 1\le j \le l} $ be the columns of the unitary matrices $U^1,U^2,\dots,U^l$
\STATE Measure the quantum state $\rho$ using the POVM $\cM=\left\{\frac{1}{l}\proj{U_i^j}\right\}_{i,j}$ and observe $\mathcal{O}\left( \sqrt{d}\log(1/\delta)/\eta^2\right)$
samples from $\rho(\cM)$. 
 \STATE Test whether $h_0:\rho(\cM)= U_{ld}$ or $h_1:\TV(\rho(\cM), U_{ld})\ge \eta/20$ using the testing identity of discrete distributions of \citep{diakonikolas2017optimal}, with an error probability $\delta$, and answer accordingly.
\end{algorithmic}
\end{algorithm}

\section{Technical lemmas}\label{sec:technical}
In this section, we group technical lemmas useful for the previous proofs of this article.
\subsection{Weingarten Calculus}
Since we use generally a uniform POVM, which consists in sampling a $\Haar$ -unitary matrix, we need some facts from Weingarten calculus in order to compute $\Haar$ -unitary integrals.  If $\pi$ a permutation of $[n]$, let $\W(\pi,d)$ denotes the Weingarten function of dimension $d$. The following lemma is crucial for our results. 
\begin{lemma}\label{lem:Wg}\citep{gu2013moments} Let $U$ be a $d\times d$ $\Haar$ -distributed unitary matrix and $\{A_i,B_i\}_i$ a sequence of $d\times d$ complex matrices. We have the following formula 
\begin{align*}
&\ex{\tr(UB_1U^*A_1U\dots UB_nU^*A_n)}
\\&=\sum_{\alpha,\beta \in \fS_n}\W(\beta\alpha^{-1},d)\tr_{\beta^{-1}}(B_1,\dots,B_n)\tr_{\alpha\gamma_n}(A_1,\dots,A_n),
\end{align*}where $\gamma_n=(12\dots n)$ and $\tr_{\sigma}(M_1,\dots,M_n)=\Pi_j \tr(\Pi_{i\in C_j} M_i)$ for $\sigma=\Pi_j C_j $ and $C_j$ are cycles. 

\end{lemma}
We need also some values of Weingarten function:
\begin{lemma}\label{lem:wg2}
\begin{itemize}
    \item $\W((1),d)=\frac{1}{d}$,
    \item $\W((12),d)=\frac{-1}{d(d^2-1)}$,
    \item $\W((1)(2),d)=\frac{1}{d^2-1}$,
    \item $\W((123),d)=\frac{2}{d(d^2-1)(d^2-4)}$,
    \item $\W((12)(3),d)=\frac{-1}{(d^2-1)(d^2-4)}$,
    \item $\W((1)(2)(3),d)=\frac{d^2-2}{d(d^2-1)(d^2-4)}$.
\end{itemize}
\end{lemma}
In particular it is  known that the sum of the Weingarten function has a closed expression:
\begin{lemma}\citep{collins2012integration}\label{lem:wg3}
Let $d,k\in \mathbb{N}^*$. We have 
\begin{align*}
    \sum_{\alpha\in \fS_k} \W(\alpha, d)= \frac{1}{d(d+1)\cdots(d+k-1)}.
\end{align*}
    
\end{lemma}
\subsection{Concentration inequalities for Haar-random unitary matrices}
The following concentration inequality is important for our results.
\begin{theorem}\label{thm:concentration}\citep{meckes2013spectral}
Let $M=\bU(d)^k$ endowed by the $L_2$-norm of Hilbert-Schmidt metric. If $F:M\rightarrow \mathbb{R}$ is $L$-Lipschitz, then for any $t>0$ 
\begin{align*}
\pr{|F(U_1,\dots,U_k)-\ex{F(U_1,\dots,U_k)}|\ge t} \le e^{-dt^2/12L^2},
\end{align*}
where $U_1,\dots,U_k$ are independent $\Haar$-distributed unitary matrices. 
\end{theorem}
\subsection{Kullback-Leibler divergence}\label{KL-sec}
\begin{definition}[Kullback Leibler divergence]
The Kullback Leibler divergence is defined for two distributions $P$ and $Q$ on $[d]$ as 
\begin{align*}
\KL(P\|Q)=\sum_{i=1}^d P_i\log\left(\frac{P_i}{Q_i}\right)\;.
\end{align*}We denote by $\KL(p\|q)=\KL(\Ber(p)\|\Ber(q))$.

\end{definition}
Kullback-Leibler’s divergence is non-negative and satisfies the Data-Processing  property: 
\begin{proposition} [\citep{van2014renyi}]
 Let $P,P',Q$ and $Q'$ distributions on $[d]$, we have 
 \begin{itemize}
     \item \textbf{Non negativity} $\KL(P\|Q)\ge0$. 
     \item \textbf{Data processing} Let $X$ a random variable and $g$ a function. Define the random variable $Y=g(X)$,  we have 
\begin{align*}
\KL\left(P^X\|Q^X\right)\ge \KL\left(P^Y\|Q^Y\right).
\end{align*}
 \end{itemize}
 
\end{proposition}

\end{document}